\documentclass[11pt,a4paper]{article}
\usepackage{amsmath} 
\usepackage{amsfonts}
\usepackage{amsthm}
\usepackage{thm-restate}

\usepackage{enumitem}

\usepackage{tikz}

\usetikzlibrary{topaths,calc,patterns}
\usepackage{pgfplots}

\usepackage{algorithm}
\usepackage{algorithmicx}

\usepackage[noend]{algpseudocode}
\usepackage[margin=2cm]{geometry}
\usepackage{graphicx}
\usepackage{xcolor}
\usepackage{blindtext}
\usepackage{caption}
\usepackage{subcaption}
\usepackage{bbm}
\usepackage{amsfonts}
\usepackage{amsmath}
\usepackage{amssymb}
\usepackage{amsthm}
\usepackage{lineno}

\DeclareMathOperator{\Ima}{Im}

\DeclareMathOperator{\uniform}{unif}

\DeclareMathOperator{\poly}{poly}
\newcommand{\emptyvec}{{\epsilon}}
\newcommand{\reals}{\mathbb{R}}

\usepackage{amssymb}
\newtheorem{theorem}{Theorem}
\newtheorem{lemma}[theorem]{Lemma}
\newtheorem{claim}[theorem]{Claim}
\newtheorem{obs}[theorem]{Observation}
\newtheorem{definition}[theorem]{Definition}

\DeclareMathOperator*{\argmax}{arg\,max}
\DeclareMathOperator*{\argmin}{arg\,min}
\usepackage[hidelinks]{hyperref}
\begin{document}
\pagenumbering{gobble}%
\clearpage
\thispagestyle{empty}
\newcommand{\floor}[1]{\left\lfloor #1 \right\rfloor}	
\newcommand{\LP}{\mbox{LP}}
\newcommand{\sol}{\mbox{sol}}
\newcommand{\empt}{\mbox{empty}}
\newcommand{\set}[1]{\{#1\}}
\newcommand{\eps}{\varepsilon}
\newcommand{\tE}{\tilde{E}}
\newcommand{\OPT}{\mbox{OPT}}

\newenvironment{claimproof}[1][\proofname]{\begin{proof}[#1]\renewcommand{\qedsymbol}{$\lrcorner$}}{\end{proof}}
\newcommand{\tp}{\tilde{p}}
\newcommand{\nonneg}{\reals_{\geq0}}
\newcommand{\ba}{{\bar{a}}}
\newcommand{\bb}{{\bar{b}}}
\newcommand{\bc}{{\bar{c}}}
\newcommand{\bn}{{\bar{n}}}
\newcommand{\bm}{{\bar{m}}}
\newcommand{\be}{{\bar{e}}}
\newcommand{\bp}{{\bar{p}}}
\newcommand{\br}{{\bar{r}}}
\newcommand{\bq}{{\bar{q}}}
\newcommand{\bk}{{\bar{k}}}
\newcommand{\bt}{{\bar{t}}}
\newcommand{\bdel}{{\bar{\delta}}}
\newcommand{\bup}{{\bar{\upsilon}}}
\newcommand{\newop}{\bullet}
\newcommand{\inner}{\cdot}
\renewcommand{\log}{\ln}
\newcommand{\dist}{\mathcal{D}}
\newcommand{\D}[2]{D\left({#1} \middle\|{#2}\right)}
\newcommand{\entropy}{\mathcal{H}}
\newcommand{\terment}{\mathcal{H}_{\textnormal{\texttt{term}}}}
\newcommand{\mD}[2]{D_e\left({#1} \middle\|{#2}\right)}
\newcommand{\Chi}{{\mathcal{X}}}
\newcommand{\Yi}{{\mathcal{Y}}}
\newcommand{\bChi}{{\bar{\Chi}}}
\newcommand{\bYi}{{\bar{\Yi}}}
\newcommand{\bgam}{{\bar{\gamma}}}
\newcommand{\blam}{{\bar{\lambda}}}
\newcommand{\bv}{{\bar{v}}}
\newcommand{\ariel}[1]{{\color{red} (Ariel: #1)}} 
\newcommand{\Rdi}{R_{\Delta}^{\infty}}
\newcommand{\Rdn}{R_{\Delta}^{n}}
\newcommand{\Rdk}{R_{\Delta}^{\leq k}}
\newcommand{\Rgi}{R_{\Gamma}^{\infty}}
\newcommand{\Rgn}{R_{\Gamma}^{n}}
\newcommand{\Rgk}{R_{\Gamma}^{\leq k}}
\newcommand{\Rqk}{R_{Q}^{\leq k}}
\newcommand{\li}[1]{ {\left\| #1\right\|^{\infty}}}
\newcommand{\critical}{\textnormal{\texttt{critical}}}

\newcommand{\terms}{\textnormal{\texttt{terms}}}
\newcommand{\type}{\textnormal{\texttt{type}}}

\newcommand{\ceil}[1]{{\left\lceil #1 \right\rceil}}
\newcommand{\lkr}{{\blam_k^R}}
\newcommand{\Pkr}{{P_k^R}}
\newcommand{\nkr}{{n_k^R}}
\newcommand{\Ckr}{{C_k^R}}
\newcommand{\mC}{\mathcal{C}}
\newcommand{\mT}{\mathcal{T}}
\newcommand{\mK}{\mathcal{K}}
\newcommand{\ulim}{{\unif{R} \lim_{k\rightarrow \infty}}}
\newcommand{\kulim}{{\unif{R\in \mR(k)} \lim_{k\rightarrow \infty}}}
\newcommand{\unif}[1]{{\underset{#1}{\uniform} }}
\newcommand{\tO}{O^*}
\newcommand{\skipfig}[1]{#1}
\newcommand{\mF}{\mathcal{F}} 
\newcommand{\mR}{\c{R}}
\newcommand{\mS}{\mathcal{S}}
\newcommand{\mI}{\mathcal{I}}
\newcommand{\mN}{\mathcal{N}}
\newcommand{\kmax}{{k_{\max}}}
\newcommand{\bmax}{{b_{\max}}}
\newcommand{\dR} {{\ddot{R}}}
\newcommand{\NeGra}{\text{{\sf NG}}}
\newcommand{\Ind}{\text{{\sf Ind}}}
\newcommand{\cG}{\mathcal{G}}
\newcommand{\one}{\mathbbm{1}}
\renewcommand{\Pr}{\mathrm{Pr}}
\newcommand{\E}{\mathrm{E}}

\begin{titlepage}
\title{Analysis of Two-variable Recurrence Relations \\ with Application to Parameterized Approximations
}
\date{}
\author{	
	Ariel Kulik\thanks{Department of Industrial Engineering and Management, Ben-Gurion University of the Negev, Beer-Sheva 8410501, Israel \mbox{kulik@bgu.ac.il }}
	\and
	Hadas Shachnai\thanks{Computer Science Department, Technion, Haifa 3200003,
		Israel. \mbox{E-mail: {\tt hadas@cs.technion.ac.il}}.
	}
}
\maketitle
\begin{abstract}

In this paper we introduce randomized branching as a tool for parameterized approximation and develop the mathematical machinery for its analysis. 
Our algorithms improve the best known running times of parameterized approximation algorithms for Vertex Cover and $3$-Hitting Set for a wide range of approximation ratios.
One notable example is
a simple parameterized random $1.5$-approximation algorithm for Vertex Cover, 
whose running time of $\tO(1.01657^k)$ substantially improves the 
best known runnning time of $\tO(1.0883^k)$ [Brankovic and Fernau, 2013]. 
For $3$-Hitting Set we present a parameterized random $2$-approximation algorithm with
running time of $\tO(1.0659^k)$, improving the best known $\tO(1.29^k)$ algorithm of
[Brankovic and Fernau, 2012].

The running times of our algorithms are derived from
an asymptotic analysis of a wide class of two-variable recurrence relations  of the form:
$$p(b,k) = \min_{1\leq j \leq N} \sum_{i=1}^{r_j} \bgam_i^j \cdot p(b-\bb^j_i, k-\bk_i^j),$$
where $\bb^j$ and  $\bk^j$ are vectors of natural numbers, and $\bgam^j$ is a probability distribution over $r_j$ elements, for $1\leq j \leq N$.   Our main theorem asserts that for any $\alpha>0$, 
$$\lim_{k \rightarrow \infty } \frac{1}{k} \cdot \log p(\floor{\alpha k},k) = -\max_{1\leq j \leq N} M_j,$$
 where $M_j$ depends only on $\alpha$, $\bgam^j$, $\bb^j$ and $\bk^j$, and can be efficiently calculated by solving a simple numerical optimization problem. 
To prove the theorem we show an equivalence between the recurrence and a stochastic process.
We analyze this process using the {\em method of types}, by introducing
an adaptation of Sanov's theorem to our setting. We believe our novel analysis of
recurrence relations which is of independent interest is a main contribution of this paper.

\end{abstract}
\end{titlepage}

\clearpage
\pagenumbering{arabic}%

\section{Introduction}
\label{sec:intro}

In search of tools for deriving efficient parameterized approximations, we 
explore the power of randomization in branching algorithms.
Recall that a {\em vertex cover} (or simply a {\em cover})  of an undirected graph $G=(V,E)$ is a subset $S\subseteq V$ such that for any $(u,v)\in E$ it holds that $S\cap \{u,v\} \neq \emptyset$. The {\em Vertex Cover} problem is to find 
a  cover of minimum cardinality for $G$.
In Vertex Cover parameterized by the solution size, $k$, we are given an integer parameter $k \geq 1$, and 
we wish to determine if $G$ has a vertex cover of size $k$ in 
time $\tO(f(k))$, for some computable function $f$.\footnote{The notation $O^*$ hides factors polynomial in the input size.}

Consider the following 
simple algorithm for the problem. Recursively pick a vertex $v$ of degree at least $3$, and 
branch over the following two options: $v$ is in the cover, or three of $v$'s neighbors are in the cover. If the maximal degree is $2$ or less then find a minimal vertex cover in polynomial time. The algorithm has a running time  $\tO(1.4656^k)$ (see Chapter 3 in \cite{CFKLMPPS15} for more details). 

The randomized branching version of this algorithm replaces branching by a random selection with some probability $\gamma \in (0,1) $. In each recursive call the algorithm selects either  $v$ or three of its neighbors into the solution, with probabilities $\gamma$ and $1-\gamma$, respectively (see Algorithm \ref{alg:VC3} for a formal description).
If $v$ is in a minimal cover then the algorithm has probability $\gamma$ to
decrease the minimal cover size by one, and probability $1-\gamma$ to select three vertices into the solution, possibly with no decrease in the minimal cover size.
A similar argument holds in case $v$ is not in a minimal cover. This suggests that the function $p(b,k)$ defined in  equation \eqref{eq:vc3intro}  lower bounds the probability the above algorithm returns a cover of size $b$, given a graph which has a cover of size $k$.
\begin{equation}
\label{eq:vc3intro}
\begin{aligned}
p(b,k) &=&\min\Bigg\{
 \begin{aligned}
 &\gamma \cdot p(b-1, k-1)~&&+ ~~(1-\gamma ) \cdot p(b-3,k) \\
&\gamma \cdot p(b-1,k) &&+ ~~(1-\gamma) \cdot  p(b-3,k-3)  
\end{aligned}
\\
p(b,k) &= 0 &~~~~\forall b<0, k\in \mathbb{Z}\\
p(b,k) &= 1& ~~~~\forall b\geq 0, k \leq 0
\end{aligned}
\end{equation}

Thus, for any $\alpha >1$, we can obtain an
$\alpha$-approximation with  constant probability  by repeating the randomized branching process $\frac{1}{p(\alpha k , k )}$ times.
While $p(b,k)$ can be evaluated using dynamic programming for any $b,k \geq 0$, finding the asymptotic behavior of $\frac{1}{p(\alpha k,k)}$ as $k \rightarrow \infty$, which dominates the running time 
of our algorithm, is less trivial.

\subsection{Our Results}

In this paper we show that randomized branching is a highly efficient tool in the development of parameterized approximation algorithms for Vertex Cover and $3$-Hitting Set, leading to significant improvements in running times over algorithms developed by using existing tools.\footnote{See Section~\ref{sec:results_vc_3hs} for a formal definition of $3$-Hitting Set.} 
One notable example is
a simple parameterized random $1.5$-approximation algorithm for Vertex Cover, 
whose running time of $\tO(1.01657^k)$ substantially 
improves the currently best known $\tO(1.0883^k)$ algorithm for the problem \cite{BF13}. 

To evaluate the running times of our algorithms, we develop mathematical tools for analyzing the asymptotic behavior of a wide class of two-variable recurrence relations generalizing the relation in~\eqref{eq:vc3intro}.
To this end, we introduce 
an adaptation of Sanov's theorem~\cite{S58} (see also~\cite{Co06}) to our setting, which facilitates the use of  {\em method of types} and {\em information theory}
for the first time in the analysis of {\em branching} algorithms.
We believe our novel analysis of recurrence relations
which is of independent interest is a main contribution of this paper.

\subsubsection{Vertex Cover and $3$-Hitting Set}
\label{sec:results_vc_3hs}

\begin{figure}
	\centering
	\begin{subfigure}{.5\textwidth}
		\centering
		\caption{Vertex Cover}
		\skipfig{
		\begin{tikzpicture}[scale = 0.9]
\begin{axis}[
xmin = 1, xmax=2, ymin =0.99 , ymax=1.35 , xlabel= 	{approximation ratio}, 
ylabel={exponent base}, samples=50]
\addplot[blue, ultra thick] (2-x, 1.2738^x ); 

\addplot[black, ultra thick] coordinates {
( 1.0  ,  1.2738 )
( 1.01  ,  1.2697966946339028 )
( 1.02  ,  1.2658059708770486 )
( 1.03  ,  1.2618277891878902 )
( 1.04  ,  1.2578621101491503 )
( 1.05  ,  1.2539088944674337 )
( 1.06  ,  1.249968102972836 )
( 1.07  ,  1.2460396966185563 )
( 1.08  ,  1.2421236364805097 )
( 1.09  ,  1.2382198837569434 )
( 1.1  ,  1.2343283997680499 )
( 1.11  ,  1.2304491459555849 )
( 1.12  ,  1.2265820838824855 )
( 1.13  ,  1.222727175232489 )
( 1.1400000000000001  ,  1.218884381809753 )
( 1.15  ,  1.215053665538477 )
( 1.16  ,  1.211234988462526 )
( 1.17  ,  1.2074283127450531 )
( 1.18  ,  1.2036336006681256 )
( 1.19  ,  1.199850814632351 )
( 1.2  ,  1.196079917156504 )
( 1.21  ,  1.1923208708771553 )
( 1.22  ,  1.188573638548303 )
( 1.23  ,  1.1848381830410002 )
( 1.24  ,  1.1811144673429903 )
( 1.25  ,  1.1774024545583384 )
( 1.26  ,  1.1737021079070669 )
( 1.27  ,  1.1700133907247903 )
( 1.28  ,  1.1663362664623513 )
( 1.29  ,  1.1626706986854614 )
( 1.3  ,  1.1590166510743358 )
( 1.31  ,  1.1553740874233371 )
( 1.32  ,  1.1517429716406147 )
( 1.33  ,  1.148123267747748 )
( 1.34  ,  1.1445149398793886 )
( 1.35  ,  1.1409179522829078 )
( 1.3599999999999999  ,  1.137332269318038 )
( 1.37  ,  1.1337578554565237 )
( 1.38  ,  1.130194675281768 )
( 1.3900000000000001  ,  1.1266426934884801 )
( 1.4  ,  1.1231018748823276 )
( 1.4100000000000001  ,  1.1195721843795874 )
( 1.42  ,  1.1160535870067976 )
( 1.43  ,  1.1125460479004097 )
( 1.44  ,  1.1090495323064467 )
( 1.45  ,  1.105564005580155 )
( 1.46  ,  1.1020894331856639 )
( 1.47  ,  1.0986257806956408 )
( 1.48  ,  1.0951730137909526 )
( 1.49  ,  1.091731098260324 )
( 1.5  ,  1.0883 )
};

\addplot[mark=*, mark options={scale=0.5}, only marks, black, forget plot] coordinates {
	(1.5,1.0883) 
	(1.666666, 1.04)
	(1.75, 1.024)
	(1.8, 1.017)
	(1.8333 , 1.01208191)
	( 1.875 ,  1.00734187 )
	( 1.88888888889 ,  1.0060641 )
	( 1.9 ,  1.00503861 )
	( 1.90909090909 , 1.00425981 )
	( 1.91666666667 , 1.00365343 )
	( 1.92307692308 , 1.00317141 )
	( 1.92857142857 , 1.00278148 )
	( 1.93333333333 , 1.00246127 )
	( 1.9375 ,  1.00221275)
	( 1.94117647059 , 1.00198584 )
	( 1.94444444444 ,  1.0017931)
	( 1.94736842105 , 1.0016279 )
	( 1.95 ,   1.00148516 )
};  %

\addplot[red, ultra thick] coordinates{ 
( 1.000000001 , 1.324717943439684 )
( 1.01 , 1.2930971872855843 )
( 1.02 , 1.2709553562950033 )
( 1.03 , 1.2524114847946912 )
( 1.04 , 1.236184666676916 )
( 1.05 , 1.221659939266156 )
( 1.06 , 1.208472688759161 )
( 1.07 , 1.196381051359742 )
( 1.08 , 1.1852128031155162 )
( 1.09 , 1.1748391359811732 )
( 1.1 , 1.1651601693734248 )
( 1.11 , 1.1560962764757634 )
( 1.12 , 1.1475825675011355 )
( 1.13 , 1.1395652100002955 )
( 1.1400000000000001 , 1.1319988798097702 )
( 1.15 , 1.1248449379905316 )
( 1.16 , 1.118070092453454 )
( 1.17 , 1.1116453932805612 )
( 1.18 , 1.1056611400607466 )
( 1.19 , 1.1005318066137195 )
( 1.2 , 1.0956590093845138 )
( 1.21 , 1.0910239370064931 )
( 1.22 , 1.0866099008693249 )
( 1.23 , 1.082402020072795 )
( 1.24 , 1.0783869639245827 )
( 1.25 , 1.0745527396574905 )
( 1.26 , 1.0708885156125345 )
( 1.27 , 1.0673844735995928 )
( 1.28 , 1.0640316839126145 )
( 1.29 , 1.0608219997059647 )
( 1.3 , 1.0577479666683236 )
( 1.31 , 1.0548027455227311 )
( 1.32 , 1.0519800455276038 )
( 1.33 , 1.0492740664240956 )
( 1.34 , 1.0466794482170065 )
( 1.35 , 1.0441912273818021 )
( 1.3599999999999999 , 1.0418047983557137 )
( 1.37 , 1.0395158794903296 )
( 1.38 , 1.037320483486005 )
( 1.3900000000000001 , 1.0352148905347995 )
( 1.4 , 1.0331956250110093 )
( 1.4100000000000001 , 1.0312594344906947 )
( 1.42 , 1.0294032711038863 )
( 1.43 , 1.027624274373305 )
( 1.44 , 1.0259197567144593 )
( 1.45 , 1.0242871892896974 )
( 1.46 , 1.022724190237944 )
( 1.47 , 1.0211061148691052 )
( 1.48 , 1.0194724241333957 )
( 1.49 , 1.0179187691522171 )
( 1.5 , 1.0165674569904897 )
( 1.51 , 1.0154641391259127 )
( 1.52 , 1.0144087091873695 )
( 1.53 , 1.0133997923158904 )
( 1.54 , 1.012436083058512 )
( 1.55 , 1.0114547596039911 )
( 1.56 , 1.0104819722592377 )
( 1.57 , 1.009576591546044 )
( 1.58 , 1.0088566788288071 )
( 1.5899999999999999 , 1.0081704955747772 )
( 1.6 , 1.007517193989135 )
( 1.6099999999999999 , 1.0068589214444632 )
( 1.62 , 1.0062168673941576 )
( 1.63 , 1.0056706102535804 )
( 1.6400000000000001 , 1.0051821501633222 )
( 1.65 , 1.004716811708091 )
( 1.6600000000000001 , 1.0042450021290232 )
( 1.67 , 1.003826644950794 )
( 1.6800000000000002 , 1.0034626045761648 )
( 1.69 , 1.0031096093998462 )
( 1.7000000000000002 , 1.0027685566578624 )
( 1.71 , 1.002483044291485 )
( 1.72 , 1.002207736466775 )
( 1.73 , 1.001951261875834 )
( 1.74 , 1.001728912143514 )
( 1.75 , 1.0015114858320462 )
( 1.76 , 1.0013276079407811 )
( 1.77 , 1.001150938002623 )
( 1.78 , 1.0009985138018003 )
( 1.79 , 1.0008579215400477 )
( 1.8 , 1.0007314241337721 )
( 1.81 , 1.0006210015967711 )
( 1.82 , 1.0005224850122387 )
( 1.83 , 1.0004353015474952 )
( 1.8399999999999999 , 1.000358773954474 )
( 1.85 , 1.0002921261368678 )
( 1.8599999999999999 , 1.0002352464447102 )
( 1.87 , 1.00018616407132 )
( 1.88 , 1.0001448402740132 )
( 1.8900000000000001 , 1.0001104438422828 )
( 1.9 , 1.0000820672258894 )
( 1.9100000000000001 , 1.0000592265304418 )
( 1.92 , 1.0000411663331896 )
( 1.9300000000000002 , 1.0000272851986407 )
( 1.94 , 1.0000170095656638 )
( 1.9500000000000002 , 1.0000097419997498 )
( 1.96 , 1.0000049368308817 )
( 1.97 , 1.000002061518123 )
( 1.98 , 1.0000005075929264 )
( 1.99 , 1.0000000000237768 )
( 1.9999999 , 1.000000000000005 )
};

\addlegendentry[no markers, blue]{FKRS \cite{FKRS18}}
\addlegendentry[no markers, black]{BF \cite{BF13}}
\addlegendentry[no markers, red]{This paper}

\end{axis}
\end{tikzpicture}
		}
		\label{fig:vc_results}
	\end{subfigure}%
	\begin{subfigure}{.5\textwidth}
		\centering
		\caption{$3$-Hitting Set}
		\skipfig{
		\begin{tikzpicture}[scale =0.9]
\begin{axis}[
xmin = 1, xmax=3.0, ymin =0.99 , ymax=2.7 , xlabel= 	{approximation ratio}, 
ylabel={
	exponent base
}, samples=50]
\addplot[blue, ultra thick] (3-2*x, 2.076^x ); 
\addplot[mark=*, mark options={scale=0.5}, black, ultra thick] coordinates {(2,1.29)};  %

\addplot[red, ultra thick] coordinates {
 (1.000100, 2.56947)	
(1.02, 2.33899355)
(1.04, 2.18878798)
(1.06, 2.07146573)
(1.08, 1.97362555)
(1.10, 1.88931003)
(1.12, 1.81538977)
(1.14, 1.75042413)
(1.16, 1.69100892)
(1.18, 1.63780126)
(1.20, 1.58981271)
(1.22, 1.54556289)
(1.24, 1.50492419)
(1.26, 1.46752551)
(1.28, 1.43283758)
(1.30, 1.40082990)
(1.32, 1.37140465)
(1.34, 1.34383936)
(1.36, 1.31802952)
(1.38, 1.30053342)
(1.40, 1.28394017)
(1.42, 1.27107745)
(1.44, 1.25768523)
(1.46, 1.24619737)
(1.48, 1.23334420)
(1.50, 1.22234041)
(1.52, 1.21063610)
(1.54, 1.19977383)
(1.56, 1.18960297)
(1.58, 1.18078432)
(1.60, 1.17290836)
(1.62, 1.16461203)
(1.64, 1.15676366)
(1.66, 1.14852382)
(1.68, 1.14189353)
(1.70, 1.13586483)
(1.72, 1.12912413)
(1.74, 1.12325692)
(1.76, 1.11762499)
(1.78, 1.11291300)
(1.80, 1.10712792)
(1.82, 1.10238824)
(1.84, 1.09805720)
(1.86, 1.09338133)
(1.88, 1.08905733)
(1.90, 1.08478085)
(1.92, 1.08085640)
(1.94, 1.07722764)
(1.96, 1.07322820)
(1.98, 1.06953912)
(2.00, 1.06587612)
(2.02, 1.06264353)
(2.04, 1.05957415)
(2.06, 1.05680487)
(2.08, 1.05378660)
(2.10, 1.05067894)
(2.12, 1.04803169)
(2.14, 1.04576359)
(2.16, 1.04314166)
(2.18, 1.04042956)
(2.20, 1.03824326)
(2.22, 1.03614700)
(2.24, 1.03413029)
(2.26, 1.03192147)
(2.28, 1.03013903)
(2.30, 1.02815843)
(2.32, 1.02628631)
(2.34, 1.02455146)
(2.36, 1.02334167)
(2.38, 1.02155238)
(2.40, 1.02004664)
(2.42, 1.01856280)
(2.44, 1.01741481)
(2.46, 1.01588359)
(2.48, 1.01467354)
(2.50, 1.01331457)
(2.52, 1.01234996)
(2.54, 1.01132808)
(2.56, 1.01035226)
(2.58, 1.00955754)
(2.60, 1.00842094)
(2.62, 1.00761580)
(2.64, 1.00738664)
(2.66, 1.00611551)
(2.68, 1.00545855)
(2.70, 1.00491500)
(2.72, 1.00436377)
(2.74, 1.00365422)
(2.76, 1.00313634)
(2.78, 1.00273339)
(2.80, 1.00252976)
(2.82, 1.00213331)
(2.84, 1.00163226)
(2.86, 1.00186557)
(2.88, 1.00121823)
(2.90, 1.00113434)
(2.92, 1.00079890)
(2.94, 1.00050466)
(2.96, 1.00016512)
(2.98, 1.00007352)
};
\addlegendentry[no markers, blue]{FKRS \cite{FKRS18}}
\addlegendentry[no markers, black]{BF \cite{BF12}}
\addlegendentry[no markers, red]{This paper}
\end{axis}
\end{tikzpicture}
		}
		\label{fig:sub2}
	\end{subfigure}
	\caption{Results for Vertex Cover and $3$-Hitting Set. A dot at $(\alpha,c)$ means that the respective algorithm outputs $\alpha$-approximation in time $\tO(c^k)$ or $\tO\left( (c+\eps)^k\right)$ for any $\eps>0$.}
	\label{fig:HS}
\end{figure}

We say that an algorithm $\cal A$  is a {\em parameterized random $\alpha$-approximation for Vertex Cover} if, given a graph $G$ and a parameter $k$ such that $G$ has a vertex 
cover of size $k$, $\cal A$ returns a vertex cover $S$ of $G$ satisfying $|S| \leq \alpha k$ with constant probability $\lambda >0$, and has running time $\tO(f(k))$. 
We refer the reader to \cite{FKRS18, BF12, LPRS17} for similar and more general definitions.

\paragraph{Vertex Cover:} Our results for Vertex Cover include two parameterized random $\alpha$-approximation algorithms,  {\sc EnhancedVC3*} and {\sc BetterVC} (presented in Sections \ref{sec:warmup} and \ref{sec:VC}, respectively). Algorithm {\sc EnhancedVC3*} uses a single branching rule (either $v$ or $N(v)$ are in a minimal cover) and has
the best running times for approximation ratios 
greater than $1.4$. 
We note that this simple algorithm outputs a $1.5$-approximation in time $\tO(1.01657^k)$.  

Algorithm {\sc BetterVC} is more complex.  It is based on a parameterized $\tO(1.33^k)$ algorithm
for Vertex Cover presented in \cite{Nied06}. {\sc BetterVC}
achieves the best running  times for approximation ratios smaller than
$1.4$. This algorithm shows that 
applying randomization in a sophisticated branching algorithm can result in 
an excellent tradeoff between approximation and time complexity for 
approximation ratios approaching $1$.

The table below compares the running time of the best algorithm presented in this paper for a given approximation ratio to the previous best results due to Brankovic and Fernau \cite{BF13}. A value of $c$ for ratio $\alpha$ means that the respective algorithm yields an $\alpha$-approximation with running time $\tO(c^k)$.
The set of values selected for $\alpha$ matches the set of
approximation ratios listed in~\cite{BF13}.  The running times presented in this paper are always rounded up. 
\begin{center}
	\begin{tabular}{ c|| c |c |c | c | c | c | c | c|  c| c} 
		ratio %
		& $1.1$ & $1.2$ &  $1.3$ &  $1.4$ &  $1.5$ &  $1.666$ & $1.75$ & $1.8$ & $1.9$ \\
		\hline
		BF \cite{BF13} &
		1.235 &
		1.197 &
		1.160 &
		1.1232 &
		1.0883  &
		1.0396 &
		1.0243 &
		1.0166 &
		1.0051 \\
		\hline 
		This paper &
		1.1652 &
		1.096 &
		1.058 &
		1.0332 & 
		1.0166 & 
		1.004 &
		1.0016 &
		1.00074 &
		1.000083  \\
		\hline	
		
	\end{tabular}
\end{center}
Figure~\ref{fig:vc_results} shows a graphical comparison between our results 
and the previous best known results~\cite{BF13, FKRS18}.\footnote{The running times 
	presented in Figures~ \ref{fig:HS}, \ref{fig:warmup} and \ref{fig:bettervc} are extrapolations of numerically evaluated running times for $99$ evenly spaced approximation ratios over the relevant range, with possible additional approximation ratios close to the endpoints of this range.}

\paragraph{$3$-Hitting Set:} 
The input for $3$-Hitting Set is a hypergraph $G=(V,E)$, where each hyperedge $e$ 
contains at most $3$ vertices, i.e., $|e|\leq 3$. We refer to such hypergraph as {\em $3$-hypergraph}. 
We say that a 
subset $S\subseteq V$ is a {\em hitting set} if, for every $e\in E$, 
$e\cap S \neq \emptyset$. The objective is to find
a hitting set of minimum cardinality. In the parameterized version, the 
goal is to determine if the input graph has a hitting set 
of at most $k$ vertices, where $k \geq 1$ is the parameter.

We say that an algorithm ${\cal A}$ is a {\em parameterized random $\alpha$-approximation for $3$-Hitting Set} if, given a $3$-hypergraph $G$ and a parameter $k$, such that $G$ has a hitting
set of size $k$, ${\cal A}$ returns a hitting set $S$ of $G$ satisfying $|S| \leq \alpha k$ with constant probability $\lambda >0$, and has running time $\tO(f(k))$.

In Section \ref{sec:HS} we present a parameterized random $\alpha$-approximation algorithm for $3$-Hitting Set for any $1<\alpha<3$.
The algorithm, \textsc{3HS} (Algorithm \ref{alg:3HS}), 
can be viewed as an adaptation of \textsc{EnhancedVC3*} to 
hypergraphs, using the following observation.
For any $v\in V$ we define the {\em neighbors graph} of $v$ as the hypergraph in which 
$\{u,w\}$ (or $\{u\}$) is an edge if $\{u,v,w\}$ ($\{u,v\}$) is an edge in the original hypergraph.
It holds that for any hitting set $S$ either $v\in S$ or $S$ contains a hitting set of the neighbors graph of $v$.  
The actual branching rules of \textsc{3HS} were determined via computer-aided search tree generation, using the above observation.

 While \textsc{3HS} may not be the best for approximation
ratios close to $1$, it yields a significant improvement over previous results for higher approximation ratios. 
For $\alpha = 2$ the running time is $\tO(1.0659^k)$, substantially improving 
the best known result of $\tO(1.29^k)$ due to \cite{BF12}.
Figure~\ref{fig:sub2} 
gives a graphical comparison between the running times achieved in this paper and the results of  \cite{BF12} and \cite{FKRS18}.

We note that while our algorithms yield significant improvements in running times for both Vertex Cover and $3$-Hitting Set over the algorithms of~\cite{BF12,BF13} and~\cite{FKRS18}, the previous algorithms are deterministic; our algorithms use randomization as a key tool.

The parameterized approximation algorithms presented in this work can also be used to derive exponential time (non-parameterized) 
$\alpha$-approximation algorithms for Vertex Cover and $3$-Hitting Set. In a recent work Esmer et al.~\cite{EKMNS22} used the parameterized approximation algorithms presented in the conference version of this paper~\cite{KS20},
along with {\em approximate monotone local search}, to derive faster exponential time approximations for Vertex Cover and $3$-Hitting Set. We refer the reader to \cite{EKMNS22} for further details.

\subsubsection{Recurrence Relations}
\label{sec:rec_intro}

The objective of our algorithms is to find a vertex cover of a 
graph under the restriction that this cover must not exceed a  given budget. The algorithms proceed by recursive application 
of a random branching step. Each time this step is executed
it adds vertices to the solution, thereby decreasing the available budget, and possibly reducing 
the number of vertices required to complete the solution.
To analyze the running times of our algorithms, we need to evaluate the probability of obtaining a cover
satisfying the budget constraint.

Similar to branching algorithms,
this property can be formulated using a recurrence relation. 
We define a function $p:\mathbb{Z}\times \mathbb{Z}\rightarrow [0,1]$ for every set of {\em terms} $\{(\bb^j,\bk^j,\bgam^j)~|~1\leq j\leq N\} $ and refer to $p$ as the {\bf composite recurrence} of $\{(\bb^j,\bk^j,\bgam^j)~|~1\leq j\leq N\} $. We require that each of the terms $(\bb^j, \bk^j,\bgam^j)$ satisfies the following technical conditions: $\bb^j \in \mathbb{N}_+^{r_j}$, 
$\bk^j \in \mathbb{N}^{r_j}$, $\bk^j$ is not the all zeros vector,  and $\bgam^j \in \mathbb{R}_{+}^{r_j}$ with $\sum_{i=1}^{r_j} \bgam^j_i= 1$.\footnote{
	Throughout the paper
	we use $\mathbb{N}$ (resp. $\mathbb{N}_+$) to denote the  non-negative (resp. positive) integers
	($\mathbb{N}=\mathbb{N}_+ \cup \{0\}$).
}
The function  $p:\mathbb{Z} \times \mathbb{Z} \rightarrow [0,1]$ is  defined by the following equations.
\begin{equation}
\label{eq:rec_relation}
\begin{aligned}
p(b,k) &= \min_{1\leq j \leq N}  \sum_{i=1}^{r_j} \bgam^j_i \cdot p(b-\bb^j_i, k-\bk^j_i) &\\
p(b,k) &= 0 & \forall b<0, k \in \mathbb{Z} \\
p(b,k) &= 1 & \forall  b\geq 0, k \leq 0
\end{aligned}
\end{equation}
 For example, the function $p$  defined in \eqref{eq:vc3intro} is the composite recurrence of 
 $\left\{(\bb^j, \bk^j, \gamma^j)|~ j= 1,2 \right\}$ with $\bb^1 = \bb^2 = (1,3)$, $\bgam^1 = \bgam^2 = (\gamma, 1-\gamma)$, $\bk^1 = (1,0)$ and $\bk^2= (0,3)$.
  
  In the context of  our randomized branching algorithms, the number of terms, 
  $N$, corresponds to the number of possible branching {\em states} (which differs from the number of {\em branching rules}). For example, in Algorithm \ref{alg:VC3} %
  (See Section \ref{sec:warmup}
   and an informal outline at the beginning of  Section \ref{sec:intro}) there are two  possible states: either $v$ is in an optimal cover, or its neighbors are. 
   Indeed, the analysis of the algorithm utilizes a composite relation with $N=2$ as given in \eqref{eq:vc3intro}.

To evaluate the running times of our algorithms we need to analyze the
asymptotic behavior of $p(\floor{\alpha k} , k )$ for a fixed $\alpha$ as $k$ grows to infinity. 
With some surprise, we did not find an existing analysis of this behavior, even for $N=1$.
The main technical contribution of this paper is Theorem \ref{thm:rec} that gives
such analysis for any $N\geq 1$. We emphasize that while the recurrence relations we want to solve are derived from coverage problems, our solution is generic and can be used for any composite recurrence.

 We say that a vector $\bq \in \nonneg^r$ is a {\em distribution} if $\sum_{i=1}^{r} \bq_i = 1$ and use $\D{\cdot}{\cdot}$ to denote {\bf Kullback}-{\bf Leibler divergence} \cite{Co06}. That is, for every $\bar{c},\bar{d} \in \mathbb{R}^k$ define\footnote{Throughout the paper we refer by $\log$ to the natural logarithm.} $$\D{\bar{c}}{\bar{d}} = \sum_{i=1}^k \bar{c}_i \log \frac{\bar{c}_i}{\bar{d}_i}.$$
 To state our main result we need the next definitions.
 For short, associate the term $(\bb, \bk ,\bgam)$ with the expression $\sum_{i=1}^{r}  \bgam_i \cdot p(b-\bb_i, k-\bk_i)$.
We first associate a {\em critical ratio} with each term. If $\alpha$ is strictly smaller than the critical ratio of any of the terms which define the composite recurrence then it can be easily shown that $p(\floor{\alpha k},k) =0$. 
 \begin{restatable}{definition}{critialdef}
 	\label{def:critical_ratio}
 	Let $\bb \in \mathbb{N}_+^{r}$, 
 	$\bk \in \mathbb{N}^{r}\setminus \{0\}$ and $\bgam \in \nonneg^r$ with $\sum_{i=1}^{r} \bgam_i= 1$.
 	The {\bf critical ratio} of the term $(\bb, \bk, \bgam)$ is
 	\begin{equation*}
 	\critical(\bb,\bk,\bgam)= \min_{1\leq i\leq r:~\bk_i\neq 0}  \frac{\bb_i}{\bk_i }.
 \end{equation*}
 \end{restatable}

We associate an {\em $\alpha$-branching number} with every term $(\bb,\bk,\bgam)$. In Theorem~\ref{thm:rec} we show that the value of $p(\floor{\alpha,k},k)$ is dominated by the maximum $\alpha$-branching number of its terms.
\begin{restatable}{definition}{alphabranching}
	\label{def:alpha_branching}
	Let $\bb \in \mathbb{N}_+^{r}$, 
	$\bk \in \mathbb{N}^{r}\setminus \{0\}$ and $\bgam \in \nonneg^r$ with $\sum_{i=1}^{r} \bgam_i= 1$.
	Then for $\alpha > \critical(\bb,\bk,\bgam)$,
	the {\bf $\alpha$-branching number} of the term $(\bb, \bk, \bgam)$ 
	is the optimal value $M^*$ of the following minimization problem over $\bdel \in \nonneg^r$:
	\begin{equation}
	\label{eq:alpha_num}
	M^*= \min \left\{\frac{1} {\sum_{i=1}^r \bdel_i \cdot \bk_i}\cdot  \D{\bdel}{ \bgam} \middle| ~ \sum_{i=1}^r \bdel_i\cdot \bb_i \leq \alpha \sum_{i=1}^r \bdel_i \cdot \bk_i,  \text {~$\bdel$ is a distribution} \right\}
	\end{equation}
\end{restatable}
The formula in \eqref{eq:alpha_num}  arise from an interpretation of the composite recurrence as a random walk, and the $\alpha$-branching number provides the probability for  a rare event in this walk. Observe that the condition $\alpha>\critical(\bb,\bk,\bgam)$ ensures that the feasibility region for the optimization problem in \eqref{eq:alpha_num} is not empty.
Our main result is the following.
\begin{restatable}{theorem}{main}
	\label{thm:rec}
	Let $p$ be the composite recurrence of $\{(\bb^j, \bk^j, \bgam^j ) |~ 1\leq j \leq N\}$, and  $\alpha>0$ such that $\alpha > \critical(\bb^j,\bk^j,\bgam^j)$ for $1\leq j\leq N$. Denote by $M_j$ the
	$\alpha$-branching number of $(\bb^j, \bk^j, \bgam^j)$, 
	and let $M=\max\{M_j | 1\leq j \leq N\}$.  
	Then,
	\[\lim_{k\rightarrow \infty} \frac{\log p\left(\floor{\alpha k} , k\right)}{k}  = -M. \]
\end{restatable}

Intuitively, Theorem \ref{thm:rec} asserts that $p(\alpha k , k )\approx \exp(-M)^k$.  Furthermore, it shows that the asymptotics of $p(\alpha k , k )\approx \exp(-M)^k$ is dominated by the ``worst'' term in $p$. 
We note that the optimization problem~\eqref{eq:alpha_num} is {\em quasiconvex}. Furthermore, all of the numerical problems in this paper arising as consequences of \eqref{eq:alpha_num} and Theorem \ref{thm:rec} are {quasiconvex}, and as such  can be solved efficiently using standard tools (these problems involve the optimization of $\bgam^j$ as well). We also note that most of these problems have a nearly closed form solution.

It is easy to show that for $p$ as defined in \eqref{eq:rec_relation} and every $b,k,n\in \mathbb{N}_+$ it holds that $p(nb,nk)\geq \left(p(b,k) \right)^n$. 
This suggests that $p$ can be lower bounded empirically
by  $p(\alpha k , k ) = \Omega( c^k)$ where $c=\left({p(\alpha k_0, k_0)}\right)^{\frac{1}{k_0}}$  for any fixed $k_0$. Indeed, this simple approach can be used in practice to derive a fairly good lower bound for $p$ in  simple cases such as \eqref{eq:vc3intro}. However, it lacks both the scale and insight required to derive the algorithmic results presented in this paper.
Furthermore, Theorem \ref{thm:rec} readily gives the desired solution, thus eliminating the need for an empirical approach as described above.

The observation that the asymptotic behavior of $p(b,k)$ is  dominated by the highest $\alpha$-branching number of the terms in $p$ served as a main guiding rule for designing the algorithms in this paper.
Most notably, the $1.5$-approximation for Vertex Cover was explicitly derived by this insight (see Section \ref{sec:eVC3star}). 
In addition, Theorem \ref{thm:rec} reduces the problem of optimizing the values of $\bgam^j$ of the terms of $p$ (e.g., the selection of $\gamma$ in \eqref{eq:vc3intro}) to multiple simple continuous quasiconvex optimization problems. In contrast, the empirical approach provides no tools for optimizing the distributions $\bgam^j$. This was crucial for deriving all of our algorithmic results, in particular the results for $3$-Hitting Set (see Section \ref{sec:HS})  which involve multiple (computer generated) branching rules.

The proof of Theorem \ref{thm:rec} is given in Section \ref{sec:rec} that is written as a stand-alone part in this paper.

We note that the requirement $\alpha>\critical(\bb^j,\bk^j,\bgam^j)$ in 
the statement of Theorem~\ref{thm:rec} is essential; 
indeed, $\lim_{k\rightarrow \infty} \frac{1}{k}\cdot \log p (\floor{\alpha k},k)$ may not exist when $\alpha$ is the critical ratio for one of the terms in $p$.\footnote{
	Consider, for example, the  recurrence $p(b,k)=p(b-4, k-2)$
	with $p(b,k)=0$ for $b<0$ and $p(b,k)=1$ for $k\leq 0 \leq b$.
	In this case, for every odd $k$ it holds that $p(2k,k)=0$ while $p(2k,k)=1$ for every even $k$, therefore $\lim_{k \rightarrow \infty } \frac{1}{k}\log p(2k,k)$ does not exist.}
In the conference version of the paper~\cite{KS20} we tackled this corner case in a different way, by slightly modifying the definition of composite recurrences~\eqref{eq:rec_relation}. Avoiding this corner case, by requiring that 
$\alpha>\critical(\bb^j,\bk^j,\bgam^j)$,
led to a significantly simpler proof
for Theorem~\ref{thm:rec}. We remark that in all of our recurrence relations the critical ratio of each term is equal to $1$, while $\alpha >1$. Hence, the requirement that $\alpha>\critical(\bb^j,\bk^j,\bgam^j)$ does not 
affect our algorithms or their analyses.

\subsection{Recurrences, Random Walks and Types}
\label{sec:int_walk}
In the following we give a brief and informal introduction to the tools and ideas used in  the proof of Theorem~\ref{thm:rec}. To do so, we focus on a specific simple instance of a  composite recurrence,  show how it can be viewed through the lens of a random walk, and explain how to analyze the random walk using  the method of types.

Let $\bgam= (\bgam_1,\bgam_2)\in \nonneg^2$ be a distribution ($\bgam_1+\bgam_2=1$), and consider the composite recurrence:
\begin{equation}
	\label{eq:example_rec}
	\begin{aligned}
		p(b,k) &= \bgam_1 \cdot p(b-4,k-3)+ \bgam_2\cdot p(b-2,k-1)\\
		p(b,k) &= 0 & \forall b<0, k \in \mathbb{Z} \\
		p(b,k) &= 1 & \forall  b\geq 0, k \leq 0
	\end{aligned}
\end{equation}
That is, $p$ is the composite recurrence of $\{ ( \bb, \bk,\bgam)\}$ where $\bb=(4,2)$ and $\bk=(3,1)$.  Observe that $\critical(\bb,\bk,\bgam) = \frac{4}{3}$. Our objective is to evaluate $\lim_{k\rightarrow \infty} \frac{1}{k}\cdot \ln p\left( \floor{\alpha\cdot k} ,k\right)$. In this informal introduction we  focus on finding $M$  such that $\liminf_{k\rightarrow \infty} \frac{1}{k}\cdot \ln p\left( \floor{\alpha\cdot k} ,k\right)>-M$ which implies $p(\floor{\alpha \cdot k } ,k)\gtrsim \left(\exp \left( -M\right) \right)^k $.  We note such lower bound suffices for all our algorithmic applications. 

\paragraph{A Random Walk.}
We associate a {\em random walk} with $p$. The walk starts at $(X_0,Y_0) = (0,0)$. At the $n$-th step of the walk a random variable $A_n\in \{1,2\}$ is sampled with $\Pr(A_n=1)=\bgam_1$ and $\Pr(A_n=2)=\bgam_2$. If $A_n=1$ then the next location of the walk is
$(X_n,Y_n) = (X_{n-1} + 4,Y_{n-1}+3) = (X_{n-1} + \bb_1, Y_{n-1}+\bk_1)$, and if $A_n=2$ then the next location is $(X_n,Y_n) = (X_{n-1} + 2,Y_{n-1}+1) = (X_{n-1} + \bb_2, Y_{n-1}+\bk_2)$. That is, with probability $\bgam_1$ the in the $n$-th step the position of that walk changes  by $(4,3)=(\bb_1,\bk_1)$, and with probability $\bgam_2$ the $n$-th step the position changes by $(2,1)=(\bb_2,\bk_2)$. See illustration in Figure~\ref{fig:walk}.

The random walk is tightly related to the recurrence $p(b,k)$.
\begin{lemma}
	\label{lem:example_equiv}
	$p(b,k) = \Pr(\exists n\in \mathbb{N}_{\geq 0}:~X_n\leq b \textnormal{ and } Y_n\geq k)$.
\end{lemma}

\begin{figure}
	\centering
\begin{tikzpicture}[scale=0.75]
	\draw[step=1cm,gray,very thin,dashed] (0,0) grid (15,8);
	\draw[thick,->] (0,0) -- (15,0) ;
	\draw[thick,->] (0,0) -- (0,8);

	\pgfsetlinewidth{0.0pt} 
	
    \draw[pattern=north east lines, pattern color=yellow] (0,6) -- (0,8) -- (9,8) -- (9,6) -- cycle;

	\foreach \x in {1,2,3,4,5,6,7,8,9,10,11,12,13,14,15}
	\draw (\x cm,2pt) -- (\x cm,-2pt)  node[anchor=north] {$\x$};
	
	\foreach \y in { 1,2,3,4,5,6,7,8}
	\draw (2pt,\y cm) -- (-2pt,\y cm) node[anchor=east] {$\y$};
	\draw[blue,->,thin,dashed] (0,0) -- (3.9,2.925);
	\filldraw[blue] (4,3) circle (3pt) node[anchor=north west] {$(X_1,Y_1)=(4,3)$};
	\draw[blue,->,thin,dashed] (4.1,3.05) -- (6-0.1,4-0.05);
	\filldraw[blue] (6,4) circle (3pt) node[anchor=north west] {$(X_2,Y_2)=(6,4)$};
	\draw[blue,->,thin,dashed] (6.1,4.075) -- (10-0.1,7-0.075);
	\filldraw[blue] (10,7) circle (3pt) node[anchor=north west] {$(X_3,Y_3)=(10,7)$};
	\filldraw[blue] (0,0) circle (3pt);

\end{tikzpicture}
\caption{In blue: an instance of the first three steps of the random walk, corresponding to  $A_1=1$, $A_2=2$ and $A_3=1$. In yellow: the area to which the walk should reach for the event $\{\exists n\in \mathbb{N}: X_n\leq 9 \textnormal{ and } Y_n\geq 6\}$ to occur.
}
\label{fig:walk}
\end{figure}

That is, $p(b,k)$ is the probability that the random walk $(X_n,Y_n)$ crossed the value  $k$ on the $y$-axis {\em before} it crossed the value $b$ on the $x$-axis (see Figure~\ref{fig:walk}). To show Lemma~\ref{lem:example_equiv} we consider the probability of the event $\{\exists n\in \mathbb{N}_{\geq 0}:~X_n\leq b \textnormal{ and } Y_n\geq k\}$ depending on $A_1=1$ or $A_1=2$. For example, if $A_1=1$ then $X_n\leq b \textnormal{ and } Y_n\geq k$ holds if and only if $X_n-X_1\leq b-4 \textnormal{ and } Y_n-Y_1\geq k-3$, as $(X_1,Y_1) = (4,3)$ in this case. That is, for $b\geq 0$ and $k>0$ we have,
\begin{equation}
	\label{eq:ex_rand_walk}
\begin{aligned}
	\Pr&(\exists n\in \mathbb{N}_{\geq 0}:~X_n\leq b \textnormal{ and } Y_n\geq k) \\[0.1cm]
	=\,& 
	\Pr(\exists n\in \mathbb{N}_{\geq 0}:~X_n\leq b \textnormal{ and } Y_n\geq k  \textnormal{ and }A_1=1)\\
	&\,+  \Pr(\exists n\in \mathbb{N}_{> 0}:~X_n\leq b \textnormal{ and } Y_n\geq k  \textnormal{ and } A_1=2)\\[0.1cm]
	=\,& \Pr(\exists n\in \mathbb{N}_{> 0}:~X_n-X_1\leq b-4 \textnormal{ and } Y_n-Y_1\geq k-3   \textnormal{ and }A_1=1)\\
	&\,+  \Pr(\exists n\in \mathbb{N}_{> 0}:~X_n-X_1\leq b-2 \textnormal{ and } Y_n-Y_1\geq k-1   \textnormal{ and }A_1=2)\\[0.1cm]
	=\,&\Pr(A_1 = 1) \cdot \Pr(\exists n\in \mathbb{N}_{> 0}:~X_n-X_1\leq b-4 \textnormal{ and } Y_n-Y_1\geq k-3)\\
	&\,+\Pr(A_1=2)\cdot  \Pr(\exists n\in \mathbb{N}_{> 0}:~X_n-X_1\leq b-2 \textnormal{ and } Y_n-Y_1\geq k-1)\\[0.1cm]
	=\,&\bgam_1\cdot \Pr(\exists n\in \mathbb{N}_{> 0}:~X_n-X_1\leq b-4 \textnormal{ and } Y_n-Y_1\geq k-3)\\
	&\,+\bgam_2\cdot  \Pr(\exists n\in \mathbb{N}_{> 0}:~X_n-X_1\leq b-2 \textnormal{ and } Y_n-Y_1\geq k-1)
\end{aligned}
\end{equation}
The third equality holds as $A_1$ and $(X_n -X_1, Y_n-Y_1)$ (which only depends on $A_2,\ldots ,A_n$) are independent. We further  observe that the distribution of $(X_n,Y_n)$ is identical to the distribution of $(X_{n+m},Y_{n+m}) - (X_m,Y_m)$ for every $n,m\in \mathbb{N}$. That is, for every $b,k,\in \mathbb{Z}$ and $m,n\in \mathbb{N}$ it holds that 
$$\Pr\left( (X_n,Y_n) = (b,k) \right) \,=\,  \Pr\left( (X_{n+m},Y_{n+m})-(X_m,Y_m) = (b,k)\right)
$$
and therefore,
$$
\Pr( \exists n\in \mathbb{N}_{\geq 0}: X_n\leq b \textnormal{ and } Y_n\geq k ) ~= ~\Pr( \exists n\in \mathbb{N}_{\geq 0}: X_{m+n}-X_m\leq b \textnormal{ and } Y_{m+n}-Y_m\geq k ).
$$
By the above equality and \eqref{eq:ex_rand_walk}, we have
\begin{equation}
	\label{eq:ex_rand_walk2}
	\begin{aligned}
		\Pr&(\exists n\in \mathbb{N}_{\geq 0}:~X_n\leq b \textnormal{ and } Y_n\geq k) \\[0.1cm]
		=\,&\bgam_1\cdot \Pr(\exists n\in \mathbb{N}_{\geq 0}:~X_n\leq b-4 \textnormal{ and } Y_n\geq k-3)\\
		&\,+\bgam_2\cdot  \Pr(\exists n\in \mathbb{N}_{> 0}:~X_n\leq b-2 \textnormal{ and } Y_n\geq k-1).
	\end{aligned}
\end{equation}
Lemma~\ref{lem:example_equiv} follows from \eqref{eq:ex_rand_walk2} via a simple induction. 

\paragraph{Types.} Fix $\alpha>\frac{4}{3}$. We use the {\em method of types} to estimate the probability of the event  $\Pr(\exists n\in \mathbb{N}_{\geq 0}:~X_n\leq b \textnormal{ and } Y_n\geq k)$. The {\em type} of $(a_1,\ldots ,a_n)\in \{1,2\}^{n}$ is the vector $T\in \mathbb{R}^{2}_{\geq 0}$ defined by $T_i = \frac{1}{n}\cdot \left| \{\ell\,|\,a_{\ell}=i\}\right|$. That is, $T_1$ ($T_2$) is the relative frequency of $1$ ($2$) in $(a_1,\ldots,a_n)$. For example, the type of $(1,2,1,2,1)$ is $\left(\frac{3}{5},\,\frac{2}{5}\right)$, as $1$ appears thrice and $2$ appears twice in  $(1,2,1,2,1)$. 
We use $\type(a_1,\ldots,a_n)$ to denote the type of $(a_1,\ldots,a_n)$. 

Our analysis relies on the property that the $n$-th location of the walk  only depends on the type of $A_1,\ldots, A_n$. For every type $T\in \mathbb{R}^2_{\geq 0}$ define $\beta(T)= 4\cdot T_1 +2\cdot T_2 = \sum_{i=1}^{2} T_i\cdot \bb_i$ and $\kappa(T) = 3\cdot T_1 +1\cdot T_2 = \sum_{i=1}^{2} T_i\cdot \bk_i$.  The values $\beta(T)$ ($\kappa(T)$) can be interpreted as the average step size of the walk on the $x$-axis ($y$-axis) if we re-adjust the probabilities such that $\Pr(A_n=1) =T_1$ and $\Pr(A_n=2) =T_2$.   Fix $n\in \mathbb{N}$ and let $T=\type(A_1,\ldots,A_n)$. Then,
$$
X_n \,=\, 4\cdot \left|\{1\leq  \ell \leq n  | A_{\ell} =1  \}\right|+2\cdot \left|\{1\leq  \ell \leq n  | A_{\ell} =2  \}\right| \, = \, n\cdot 4\cdot T_1 + n \cdot 2 \cdot T_2 = n\cdot \beta(T),
$$
where the first equality holds as $X_{n}$ advances by $4$ when $A_n =1$ and by $2$ when $A_n=2$. The second equality follows from the definition of types.  Similarly, it can be shown that $Y_n = n\cdot \kappa(T)$. Therefore,
\begin{equation}
	\label{eq:example_pr_type}
\begin{aligned}
p(\floor{\alpha \cdot k},k )\,&=\, \Pr(\exists n\in \mathbb{N}_{\geq 0}:~X_n\leq \alpha\cdot k \textnormal{ and } Y_n\geq k)\\
&=\,  \Pr\left(\exists n\in \mathbb{N}_{\geq 0}:~\beta(\type(A_1,\ldots,A_n))\leq \frac{\alpha\cdot k}{n} \textnormal{ and } \kappa(\type(A_1,\ldots,A_n))\geq \frac{k}{n}\right).
\end{aligned}
\end{equation}
That is, in \eqref{eq:example_pr_type} we showed the event  $\left\{\exists n\in \mathbb{N}_{\geq 0}:~X_n\leq \alpha\cdot k \textnormal{ and } Y_n\geq k\right\}$ only depends on the types of the random vectors $(A_1,\ldots, A_n)$ for various values of $n$.  In the following we use Sanov's theorem \cite{S58} to lower bound the probability of the event in \eqref{eq:example_pr_type}.

We can arbitrarily lower bound the probability of the event in the last expression of \eqref{eq:example_pr_type} by focusing on a specific value  for $n$. We would guess that $n^*= \rho \cdot k$  is a useful choice, and later optimize the value of $\rho$. 
\begin{equation}
	\label{eq:example_pr_type2}
	\begin{aligned}
		p(\floor{\alpha \cdot k},k )\,
		&=\,  \Pr\left(\exists n\in \mathbb{N}_{\geq 0}:~\beta(\type(A_1,\ldots,A_n))\leq \frac{\alpha\cdot k}{n} \textnormal{ and } \kappa(\type(A_1,\ldots,A_n))\geq \frac{k}{n}\right)\\
		&\geq  \,\Pr\left(\beta(\type(A_1,\ldots,A_{\rho \cdot k }))\leq \frac{\alpha\cdot k}{\rho \cdot k} \textnormal{ and } \kappa(\type(A_1,\ldots,A_{\rho \cdot k}))\geq \frac{k}{\rho \cdot k}\right).
	\end{aligned}
\end{equation}
Define 
$$Q_{\rho}\,=\,\left\{ T\in \mathbb{R}^{2}_{\geq 0} ~|~ \beta(T) \leq \frac{\alpha}{\rho} \textnormal{ and } \kappa(T)\geq \frac{1}{\rho} \textnormal{ and  $T$ is a distribution}\right\},$$
then by \eqref{eq:example_pr_type2} we have 
\begin{equation}
	\label{eq:example_pr_type3}
	\begin{aligned}
		p(\floor{\alpha \cdot k},k )\,
		\geq \, \Pr\left(\type(A_1,\ldots,A_{\rho\cdot k})\in Q_{\rho}\right).
	\end{aligned}
\end{equation}

The following is Sanov's theorem  \cite{S58} (see also \cite{Co06}) stated  to our setting.
\begin{theorem}[Sanov's Theorem]
	\label{thm:sanov}
	Let $R\subseteq\{ T\in \mathbb{R}_{\geq 0}^{2 }\,|\,\textnormal{$T$ is a distribution}\} $ be a set of distributions, such that $R$ is the closure of its interior.
	 Then, $$\lim_{n\rightarrow \infty } \frac{1}{n} \ln \Pr(\type(A_1,\ldots, A_n) \in R) = -\min_{T\in R } \D{T}{\bgam}.$$ 
\end{theorem}
	Intuitively, the theorem states that the probability that the type of  a sequence of $n$ independent and identically distributed ({\em i.i.d}) random variables, 
    distributed by $\bgam$, is in a set $R$ is dominated  by the distance between $\bgam$ and the distribution closest to $\bgam$ in $R$; the distance is measured using  the Kullback-Leibler divergence.  As $Q_{\rho}$ is the closure of its interior, it holds that 
\begin{equation}
	\label{eq:example_lower}
\liminf_{k\rightarrow \infty} \frac{1}{k}\cdot \ln p(\floor{\alpha\cdot k},k) \,\geq \, \lim_{k\rightarrow \infty} \rho \cdot \frac{1}{\rho\cdot k}\cdot \ln \Pr(\type(A_1,\ldots, A_{\rho\cdot k}) \in Q_{\rho}) \,=\,-\rho \cdot \min_{T\in Q_{\rho} } \D{T}{\bgam}.
\end{equation}
The inequality is by \eqref{eq:example_pr_type3} and the equality follows from Theorem~\ref{thm:sanov}. 
The above provides a lower bound for $\liminf_{k\rightarrow \infty} \frac{1}{k}\cdot \ln p(\floor{\alpha\cdot k},k)$ which depends on $\rho$. By selecting the optimal value of $\rho$ we can match the lower bound in Theorem~\ref{thm:rec} which uses branching numbers,  
 and attain a matching upper bound (we omit the details).

\paragraph{The General Case.} In the above we outlined concepts and ideas which can be used to evaluate the asymptotic behavior of the recurrence in \eqref{eq:example_rec}.  The proof of Theorem~\ref{thm:rec} uses a similar outline, while  overcoming major challenges which arise when considering composite recurrences which involve multiple terms (i.e., $N>1$).  

The random walk  used for the proof of Theorem~\ref{thm:rec} involves an {\em adversary} which selects the term $(\bb^j,\bk^j,\bgam^j)$ that will be used for the $n$-th step of the random walk. 
That is, the $n$-th step is $(X_n,Y_n) = (X_{n-1},Y_{n-1})+(\bb^j_i,\bk^j_i)$ where $j$ is selected by the adversary (and may depend on the first $n-1$ steps of the random walk ) and then $i$ is sampled according to $\bgam^j$ ($\Pr(i=i') = \bgam^j_{i'}$).
We model the adversary's behavior using a strategy and show equivalence between the recurrences and the random walk with an optimal adversary. 

We subsequently use the method of types for the analysis of the random walk. However, due to the adversary, the steps of the random walk do not form a sequence of i.i.d, and the vanilla techniques from the method of types, and specifically Sanov's theorem (Theorem~\ref{thm:sanov}), do not apply.  We show that several properties of types can be adjusted to our setting, and use those to prove Theorem~\ref{thm:rec}. The proof is inspired by the proof of Sanov's theoem, and can be viewed as an adaptation of the proof to our setting.

\subsection{Related Work}

Vertex Cover is one of the fundamental problems in computer science, and a 
testbed for new techniques in parameterized complexity.
The problem admits a polynomial time $2$-approximation, which 
cannot be improved under the {\em Unique Games Conjecture (UGC)}~\cite{KR08}. 
Vertex Cover has been widely studied from the viewpoint of parameterized 
complexity. 
We say that a problem (with a particular parameter $k$) is {\it fixed-parameter tractable (FPT)} if it can be solved in time $f(k) \cdot \poly(n)$, where $f$ is 
some computable function depending only on $k$. 
Vertex Cover parameterized by the solution size is well known to be FPT (see, e.g.,~\cite{Nied06}). 
Until very recently, the fastest known running time of an FPT
algorithm for the problem was $O^*(1.273^k)$, due to Chen et at.~\cite{CKX10}.
The current best result is $O^*(1.25288^k)$ due to Harris and Narayanaswamy \cite{HN22}.
Also, there is no $2^{o(k)} \cdot \poly(n)$ 
algorithm for the problem, under the {\em exponential time hypothesis (ETH)}. 

In~\cite{BRKP15} it was shown that there is no $(7/6 - \eps)$
approximation for Vertex Cover with running time $O(2^{n^{1-\delta}})$ for
any $\delta>0$ under ETH. In~\cite{MT18}   Manurangsi and
Trevisan showed a $\left(2-{1}/{O(r)}\right)$-approximation for the problem
with running time $\tO ( \exp({n}\cdot{2^{-r^2}}))$, improving
upon earlier results of~\cite{BCLNN18}. To the best of our knowledge,
the existence of a $(2-\eps)$-approximation 
for Vertex Cover with running time $2^{o(n)}$ is still open.

The above results suggest that for $\alpha< 7/6$ subexponential $\alpha$-approximation algorithms are unlikely to exist, and even as the approximation
ratio approaches $2$ the barrier of exponential running time remains unbreached.  
This motivates our study of parameterized $\alpha$-approximation algorithms for Vertex Cover, for
$1 < \alpha < 2$, whose running times are exponential in the solution size, $k$. 

Brankovic and Fernau presented in~\cite{BF13} a branching 
algorithm that yields a parameterized $1.5$-approximation for Vertex Cover with 
running time $O^*(1.0883^k)$. 
In~\cite{FKRS18} Fellows et al. presented an $\alpha$-approximation algorithm whose running time is 
$O^*(1.273^{(2-\alpha)k})$, for any $1\leq \alpha \leq 2$. 
A similar result was obtained in~\cite{BEP11} using a different technique.

Similar to Vertex Cover, $3$-Hitting Set cannot be approximated 
within a constant factor better than $3$ under UGC~\cite{KR08}, and there is no subexponential algorithm for the problem 
under ETH. The best known parameterized algorithm for the problem has running time of $\tO(2.076^k)$ \cite{Wa07}. 
Previous works on parameterized approximation for $3$-Hitting Set
resulted in an $\alpha$-approximation in time $\tO(2.076^{k(3-\alpha)/2})$
due to \cite{FKRS18}, for any $1 \leq \alpha \leq 3$, and a $2$-approximation in time $\tO(1.29^k)$ using a branching 
algorithm by Brankovic and Fernau \cite{BF12}.

Randomized branching is a well known approach for algorithm design (see, e.g, \cite{BBG99, BE05,LRS18}). 
Often, 
the analysis of such algorithms is narrowed to evaluating the probability that in {\em every} branching step the algorithm makes a correct branching choice (in contrast, in our analysis the aim is to bound the number of incorrect steps). This leads to a  one-variable recurrence  which can be simply solved.
Randomized branching was used for approximation in \cite{BCLNN18}, along with a tailored analysis for the approximation ratio. 

The idea of sampling leaves from a branching tree was studied in the past from a different perspective.
Specifically, it was used in \cite{DNS16}   to justify  
one-sided probabilistic polynomial algorithms as a computational model for 
branching algorithms. 
Within this model, the authors derived lower bounds for branching algorithms.

Previous works on parameterized approximations for both  Vertex Cover and $3$-Hitting Set either considered 
approximative preprocessing \cite{FKRS18} or used approximative (worsening) steps within branching algorithms \cite{BF12,BF13}. While these techniques
use the approximative step explicitly at given stages of the algorithm execution, in randomized branching the approximative step takes the form of an incorrect branching decision, which may add unnecessary vertices to the solution. As incorrect branching is not restricted to a specific stage, a degree of freedom is added to the number of 
{\em good paths} within the branching tree. This degree of freedom in turn increases the probability of finding an approximate solution. 
This gives some intuition 
to the improved performance of our algorithms.

\subsubsection{Recurrence Relations and the Method of Types}

The analysis of single variable recurrence relations (e.g., $f(n) = \sum_{i=1}^{N} f(n-a_i)$) is a cornerstone in the analysis of parameterized branching algorithms 
that is often included in introductory textbooks on parameterized algorithms (see, e.g.,~\cite{Nied06,CFKLMPPS15}). 

In~\cite{Epp06} Eppstein introduced a technique for computing the asymptotic
behavior of multivariate recurrences of the form $f(x) = \max_i \sum_j f(x-\delta_{i,j})$, where $f:\mathbb{Z}^d\rightarrow \mathbb{Z}$ and $\delta_{i,j}\in \mathbb{N}^d$. For any $t\in \mathbb{N}^d$, the technique
shows how to compute a constant $c$ such that $f(nt) \approx c^n$ up to a polynomial factor. The technique is based on a tight reduction of the multivariate recurrence 
to a solvable single variable recurrence, where the reduction is computed using a quasiconvex program.
A matching lower bound  to the result of the quasiconvex program is derived using a random walk, which bears some similarity to the reduction used in this paper from a recurrence to a stochastic process. Nevertheless, the analysis in this paper is significantly different. 
 
The result in~\cite{Epp06} is commonly used in the analysis  of parameterized algorithms, and specifically within the context of Measure and Conquer~\cite{FGK09} as a black box. 

We emphasize that the recurrences considered in~\cite{Epp06} are different from the recurrences studied in this paper. The difference seems to be more than merely technical. The recurrences in~\cite{Epp06}
commonly measure the size of a branching tree, while our recurrence relations are aimed at bounding the number of leaves adhering to certain property within the tree. In fact, the size of the branching trees
considered in this paper can be easily evaluated using standard single variable recurrence relations.
We are not aware of other works relating to the analysis of similar multivariate recurrences.

The {\em method of types} is a powerful technique developed mostly within the context of information theory in a line of works, starting  from the early works of Sanov \cite{S58} 
and Hoeffding \cite{Ho65}.
The current form of the method is attributed to the works of Csiszar et al. \cite{C98}.  Along with Sanov's theorem, the prominent results attained using the method of types are universal block coding and 	hypothesis testing (we refer the reader to the survey in \cite{C98} and to  Chapter 11 in \cite{Co06}). While the method of types is considered a basic tool in
information theory, it seems much less known in theoretical computer science.

\subsection{Organization}

Section~\ref{sec:warmup} includes a technical introduction to randomized branching using several algorithms for Vertex Cover, which gradually 
reveal the main algorithmic ideas presented in this paper.
 The algorithmic results for $3$-Hitting Set and a more sophisticated
algorithm for Vertex Cover are given in Sections~\ref{sec:HS} and~\ref{sec:VC}.
An overview of the numerical tools used to calculate the running times of our algorithms, based on Theorem~\ref{thm:rec}, is given in 
Section \ref{sec:numerical}. %
Section~\ref{sec:rec} gives the proof of Theorem~\ref{thm:rec}. 
Finally, in Section~\ref{sec:discussion} we discuss open problems and some directions for future work.

\section{Our Technique: Warm-up}
\label{sec:warmup}
We start by completing the analysis of the algorithm presented in Section \ref{sec:intro}.  
A formal description of the algorithm, $\text{VC3}_\gamma$, is given in Algorithm \ref{alg:VC3}. 
While the performance of Algorithm \ref{alg:VC3} can be significantly improved, as we show below,
 it demonstrates the main tools and concepts developed in this paper, and its analysis involves only few technical details. Interestingly, already this simple algorithm 
  improves the  state-of-the-art results for a wide range of approximation ratios. 
Sections \ref{sec:VC3star} and \ref{sec:eVC3star} present variants of Algorithm \ref{alg:VC3}, which perform even better. Each
section introduces some new ideas. The results of the algorithms presented in this section are 
depicted in Figure~\ref{fig:warmup}.

Clearly, Algorithm \ref{alg:VC3} has a polynomial running time. Also, it always returns a vertex cover of the input graph $G$.
  The algorithm depends on a configuration parameter $\gamma \in (0,1)$ which determines the probability that the set is $S=\{v\}$ or $S=\{u_1,u_2,u_3\}$ in Step~\ref{VC3:S}. We analyze the algorithm for an arbitrary $\gamma$ and show how to select an optimal value for it later. 
  Let $\mathcal{G}_k$ be the set of graphs with a vertex cover of size $k$ or less. Also, 
let $P_{\gamma}(b,k)$ be the minimal probability that Algorithm \ref{alg:VC3} returns a solution of size at most $b$, given a graph $G\in \mathcal{G}_k$. 
That is, $P_{\gamma}(b,k) = \min_{G\in \mathcal{G}_k}  \Pr \left[ ~ \left| \text{VC3}_{\gamma}(G) \right| \leq b ~ \right] $.
Using the arguments given in Section~\ref{sec:intro}, it is easy
to show by induction that $P(b,k)\geq p_{\gamma}(b,k)$, where $p_{\gamma}(b,k)$ is 
defined by the following recurrence relation.

\begin{equation}
\label{eq:VC3_rec}
\begin{aligned}
p_{\gamma}(b,k) &=&\min&
\begin{cases}
\gamma \cdot p_{\gamma} (b-1, k-1) + (1-\gamma )\cdot  p_{\gamma}(b-3,k) &  \\
\gamma\cdot  p_{\gamma}(b-1,k) + (1-\gamma)\cdot  p_{\gamma}(b-3,k-3) & 
\end{cases} \\
 p_\gamma(b,k) &= 0& &  &\forall b<0\\
 p_\gamma(b,k) &= 1& & &\forall b\geq 0, k \leq 0
\end{aligned}
\end{equation}
That is, $p_{\gamma}$ is the composite recurrence of 
$\left\{(\bb^j, \bk^j, \gamma^j)|~ j= 1,2 \right\}$ with $\bb^1 = \bb^2 = (1,3)$, $\bgam^1 = \bgam^2 = (\gamma, 1-\gamma)$, $\bk^1 = (1,0)$ and $\bk^2= (0,3)$. Note that in this case $N=2$ and $r_1 = r_2 = 2$ (recall that a composite recurrence is defined in Section~\ref{sec:rec_intro}).

Hence, by repeating the execution of Algorithm \ref{alg:VC3} for 
$p_\gamma(b,k)^{-1}$ times, 
we have a constant probability to find a cover of size $b$ 
or less, for any $G\in \mathcal{G}_k$. This is achieved by using Algorithm \ref{alg:alpha_approx}, taking 
Algorithm~\ref{alg:VC3} as $\mathcal{A}$ and $p=p_{\gamma}$. We call the
resulting algorithm \textsc{$\alpha$-VC3}. 

\begin{algorithm}
	\caption{\sc VC3$_\gamma$} 
	\label{alg:VC3}
	\hspace*{\algorithmicindent} \noindent \textbf{Input:} An undirected graph $G$
	
	\begin{algorithmic}[1]	
		\If {$G$ has a vertex $v$ with degree $3$ or more \label{VC3:if}} 
		\State  Let $u_1, u_2, u_3$ be $3$ of $v$'s neighbors.
		\State
		With probability $\gamma$ set $S= \{v\}$ and $S=\{u_1, u_2, u_3\}$ with probability $1-\gamma$. 		\label{VC3:S}
		\State 
		Use a recursive call to evaluate $R=\text{VC3}_\gamma(G\setminus S)$, and return $R\cup S$. 
		\Else { the maximal degree in $G$ is not greater than $2$}
		\State Find an optimal cover $S$ of $G$ in polynomial time and return it.
		\EndIf	
	\end{algorithmic}
\end{algorithm}

\begin{algorithm}
	\caption{\sc $\alpha$-Approx} 
	\label{alg:alpha_approx}
\noindent \textbf{Input:} An undirected graph $G$, a parameter $k$, an algorithm $\mathcal{A}$ and a recurrence relation
 	 $p$. 
	
	\begin{algorithmic}[1]	
		\State Evaluate $r=p(\floor{\alpha k} , k)$ using dynamic programming.
		\State Execute  $\mathcal{A}(G)$ for $\ceil{r^{-1}}$ times. Return the minimal cover found. 
	\end{algorithmic}
\end{algorithm}

We note that if $G\in \mathcal{G}_k$ then \textsc{$\alpha$-VC3} returns a cover of size at most $\alpha k$ with constant probability.
Clearly, the running time of the algorithm is $\tO( (p_{\gamma}(\alpha k, k))^{-1} )$. We resort to Theorem \ref{thm:rec} to obtain a better understanding of the running time. 

It can be easily verified that the critical ratio (Definition~\ref{def:critical_ratio}) of each of the terms in \eqref{eq:VC3_rec} is equal to $1$. Thus,
for any $\alpha>1$ and  $\gamma\in (0,1)$, we can calculate the $\alpha$-branching numbers $M^{\alpha, \gamma}_1, M^{\alpha, \gamma}_2$
of $(\bb^1, \bk^1, \bgam^1),(\bb^2, \bk^2, \bgam^2)$, respectively, by numerically  solving the optimization problem \eqref{eq:alpha_num}.
Let $M^{\alpha, \gamma}= \max \{M^{\alpha, \gamma}_1, M^{\alpha, \gamma}_2\}$. 
Therefore, by Theorem \ref{thm:rec} we have $\lim_{k\rightarrow \infty} \frac{\log p_{\gamma}(\alpha k , k )}{k} = -M^{\alpha, \gamma}$.
Thus, for any $\eps>0$ and large enough $k$, it holds that 
$\frac{\log p_{\gamma}(\alpha k , k )}{k} > -M^{\alpha, \gamma} -\eps $,
and equivalently  $ (p_{\gamma}(\alpha k, k) ) ^{-1} < \exp(M^{\alpha, \gamma}+\eps)^k$.
We conclude that the running time 
of  \textsc{$\alpha$-VC3} is  $\tO( (p_{\gamma}(\alpha k, k))^{-1} ) = \tO(\exp(M^{\alpha, \gamma}+\eps)^k)$ for any $\eps>0$.

For any $\alpha>1$, we can numerically find  the value of $\gamma$ 
for which $M^{\alpha, \gamma}$ is minimal. Let $\gamma_\alpha$ be this value. Then, for any 
$\alpha>1$ algorithm  \textsc{$\alpha$-VC3} is a parameterized
random $\alpha$-approximation for Vertex Cover
with running time  $\tO(\exp(M^{\alpha, \gamma_\alpha}+\eps)^k)$ (for any $\eps>0$). 
For example, for $\alpha =1.5$ we get that \textsc{$\alpha$-VC3} has a running time of 
$\tO(1.04364^k)$. In Figure~\ref{fig:warmup} the value of $\exp(M^{\alpha, \gamma_\alpha})$ is presented as a function of $\alpha$. 
An overview of the methods used for the numerical optimizations is given in Section \ref{sec:numerical}.
 
\subsection{A Refined Analysis of Randomized  Branching}
\label{sec:VC3star}

Standard branching algorithms derive several simpler sub-instances from a given instance with a guarantee that
an optimal solution to one (specific yet unknown) of the sub-instances leads to an 
optimal solution. Therefore, the analysis is focused on this specific sub-instance and ignores the effect of other sub-instances on the 
optimum. This is not the case when using randomized branching 
for {\em approximation}, where the reduction in the minimal 
cover size by an incorrect branching can lead to an improved
running time, as we demonstrate below.

Consider the following observation.
 If $v$ is a vertex of degree exactly $3$ and the algorithm (e.g., Algorithm \ref{alg:VC3}) selects its three neighbors $\{u_1,u_2,u_3\}$ to the cover, then even if none of $\{u_1,u_2,u_3\}$ belongs to an optimal cover, the size of the optimal cover decreases by one (as $v$ is part of an optimal cover, but is no more required). This observation can be extended to any fixed degree of $v$. 

Algorithm \ref{alg:VC3star} takes advantage of this property by using  a different probability for selecting $v$ or its neighbors depending on its degree, as well as selecting all the neighbors of $v$ in case the degree of $v$ is smaller than $\Delta$, for some fixed $\Delta \in \mathbb{N}$. 

{\tiny 
\begin{algorithm}
	\caption{\sc VC3*$_{\gamma_3, \gamma_4, \ldots, \gamma_\Delta}$} 
	\label{alg:VC3star}
	\hspace*{\algorithmicindent} \noindent \textbf{Input:} An undirected graph $G$
	
	\begin{algorithmic}[1]	
		\If {$G$ has a vertex $v$ with degree $3$ or more}
		\State Let $d= \min \{ deg(v), \Delta\}$.
		\State %
		If $d<\Delta$ let $U=N(v)$, otherwise let $U$ be a subset of $N(v)$ of size exactly $\Delta$.
		\State
		With probability $\gamma_d$ set $S= \{v\}$ and $S=U$ with probability $1-\gamma_d$. 
		\State 
		Use a recursive call to evaluate $R=\text{VC3*}_{\gamma_3,\gamma_4, \ldots, \gamma_\Delta} (G\setminus S)$, and return $R\cup S$. 
		\Else { the maximal degree in $G$ is $2$}
		\State Find an optimal cover $S$ of $G$ in polynomial time and return $S$.
		\EndIf	
	\end{algorithmic}
\end{algorithm}}

Clearly, Algorithm \ref{alg:VC3star} is polynomial and always returns a cover of $G$.
Similar to Algorithm \ref{alg:VC3}, it can be shown that the probability Algorithm \ref{alg:VC3star} returns a solution of size $b$, given a 
graph $G\in \mathcal{G}_k$, is at least $p(b,k)$, where $p$ is given by
\begin{equation}
\label{eq:vcstar_rec}
p(b,k) = \min \begin{cases}
	\gamma_d\cdot p(b-1, k-1) ~+~ (1-\gamma_d)\cdot  p(b-d, k-1) &  3\leq  d < \Delta \\ 
	\gamma_d \cdot p(b-1, k) ~+~ (1-\gamma_d)\cdot  p(b-d, k-d) &  3\leq  d <\Delta\\ 
	\gamma_\Delta \cdot  p(b-1, k-1) ~+~ (1-\gamma_\Delta) \cdot p(b-\Delta, k) &  \\ 
	\gamma_\Delta \cdot  p(b-1, k) ~+~ (1-\gamma_\Delta) \cdot p(b-\Delta, k-\Delta) &  \\ 
\end{cases}
\end{equation}
with $p(b,k)=0$ for $b<0$ and $p(b,k)= 1 $ for $b\geq 0$ and $k\leq0$. Observe $p$ depends on $\gamma_3, \gamma_4, \ldots, \gamma_{\Delta}$. 
Clearly, $p$ is a composite recurrence relation of the $N=2(\Delta -2)$ terms (triplets) 
\begin{equation}
	\label{eq:terms_vcstar}
\begin{aligned}
	\terms= &\{ ~((1,d), (1,1), (\gamma_d, 1-\gamma_d))~|~ 3\leq d < \Delta~\} ~~\cup\\ &  
	 \{~((1,d), (0,d), (\gamma_d, 1-\gamma_d))~|~ 3\leq d < \Delta ~\}~~ \cup\\
	&\{~ ((1, \Delta), (1, 0), (\gamma_\Delta, 1-\gamma_\Delta)),~
	((1, \Delta), (0, \Delta), (\gamma_\Delta, 1-\gamma_\Delta)) ~\} \mbox{.}
\end{aligned}
\end{equation}
  
And as before, we can derive an approximation algorithm 
by using Algorithm \ref{alg:alpha_approx} with Algorithm \ref{alg:VC3star} as $\mathcal{A}$ and $p$ as defined in \eqref{eq:vcstar_rec}. Let 
 \textsc{$\alpha$-VC3*} be this algorithm. Clearly, \textsc{$\alpha$-VC3*} is a random parameterized $\alpha$-approximations
algorithm for Vertex Cover.

Arbitrarily, we select $\Delta=100$. Observe that the critical ratio of all the terms in \eqref{eq:vcstar_rec} is $1$. 
As before, for every $1<\alpha<2$ and $1\leq d <\Delta$ we can find the value $\gamma_{\alpha,d}$ such that the maximal 
$\alpha$-branching number of  $( (1,d), (1,1), (\gamma_{\alpha,d}, 1-\gamma_{\alpha,d}))$ and
 $( (1,d), (0,d), (\gamma_{\alpha,d}, 1-\gamma_{\alpha,d}))$ is minimal. Let $M_{\alpha, d}$ be this value. Also, 
 we can find the value $\gamma_{\alpha,\Delta}$ such that the maximal 
 $\alpha$-branching number of  $( (1,\Delta), (1,0), (\gamma_{\alpha,\Delta}, 1-\gamma_{\alpha,\Delta}))$ and
 $( (1,\Delta), (0,\Delta), (\gamma_{\alpha,\Delta}, 1-\gamma_{\alpha,\Delta}))$ is minimal and let $M_{\alpha, \Delta}$ be this value. Let $M_\alpha$ be the maximal branching
 number of these triplets for a given value of $\alpha$ and $3\leq d\leq \Delta$ ($M_\alpha = \max_{3\leq d \leq \Delta} M_{\alpha, d}$). Then by Theorem \ref{thm:rec}, for any $\eps>0$ and large enough
 $k$, it holds that $p(\alpha  k  , k )\geq \exp \left(  -M_\alpha -\eps \right)$, and therefore the running time of  \textsc{$\alpha$-VC3*} is $\tO\left(\exp \left(M_\alpha+\eps \right)^k  \right)$. For $\alpha = 1.5$ the running time is $\tO( 1.0172^k)$.  
 Figure \ref{fig:warmup} shows $\exp( M_\alpha)$ as a function of $\alpha$.

\begin{figure}
	
	\centering
	\skipfig{
		\input{figure5.tex}
	}
	\caption{Results of Section \ref{sec:warmup}. A dot at $(\alpha,c)$ means that the respective algorithm provides $\alpha$-approximation for Vertex Cover with running time $\tO(c^k)$ or $\tO\left( (c+\eps)^k\right)$ for every $\eps>0$. }
	\label{fig:warmup}
\end{figure}

\subsection{Further Insights from using $\alpha$-Branching Numbers}
\label{sec:eVC3star}

In the context of classic branching algorithms, the running time of an algorithm is dominated by the highest branching number of the branching rules used by the algorithm (see, e.g., \cite{Nied06,CFKLMPPS15}). 
This observation is commonly used in the design of (exact) branching algorithms.
Theorem \ref{thm:rec} asserts that essentially the same  holds for 
 parameterized approximation using randomized branching.
 In the following we show how to use it to improve the running time of \textsc{VC3*}.

Consider  algorithm \textsc{$\alpha$-VC3*} of Section \ref{sec:VC3star}, whose time complexity is the inverse of the function in~\eqref{eq:vcstar_rec}. 
As an  example, for $\alpha = 1.5$ we can sort the values 
$M_{\alpha,d}$ to understand which value of $d$ dominates 
the running time. We show the nine highest values in the table below (the values are rounded up).
\begin{center}
	\begin{tabular}{ c|| c |c |c | c | c | c | c | c|  c|} 
		 $d$ & 5 & 6 &  4 &  7 &  8 &  9 & 10 & 11&  3  \\
		\hline
		 $\exp\left(M_{1.5, d}\right)$ & 1.0172& 1.0166 & 1.0165 & 1.0157 & 1.0147 &
		 1.0137 & 1.0129 & 1.0121 & 1.0119 
	\end{tabular}
\end{center}

This suggests that avoiding branching over degree $5$ vertices leads to an $\tO(1.0166^k)$ algorithm. 
In fact, tools to do so have already been used in previous works, such as \cite{NR99}. 
The basic idea is that as long as there is a vertex  $v$ in 
the graph of degree different than $5$ the algorithm branches on 
it. If all vertices in the graph are of degree $5$ the algorithm has to perform a branching on a degree $5$ vertex; however, such event cannot happen more than once along a branching path. Therefore, the algorithm can use non-randomized branching in this case while maintaining a polynomial running time.

\begin{algorithm}
	\caption{\sc EnhancedVC3*} 
	\label{alg:EnVC3}
	 \noindent \textbf{Input:} An undirected graph $G=(V,E)$ \\ 
	 \noindent \textbf{Configuration Parameters:} The algorithm depends on several parameters that should be configured. These include 
	$\Delta \in \mathbb{N}$, $\delta\in \mathbb{N}$, $2\leq \delta <\Delta$, and  $\gamma_2, \ldots, \gamma_{\delta-1} , \gamma_{\delta+1} , \ldots, \gamma_\Delta \in (0,1)$. 
	\begin{algorithmic}[1]	
		\State If the empty set is a cover of $G$ return $\emptyset$.
		\If  {$G$ is not connected}
		\State Let $C$ be a %
		component of $G$. 
		Return $\text{\textsc{EnhancedVC3*}}(C)\cup \text{\textsc{EnhancedVC*}}(G- C)$.  \label{eVC3:split}
		\EndIf
		\State \label{eVC:deg2}If $G$ has a vertex $v$ of degree $1$, let $u$ be its neighbor. Return $\text{\sc EnhancedVC3*}(G\setminus \{u\})\cup\{u\}$. %
		\If {$G$ has a vertex $v$ of degree $d\neq \delta$}
		\State Let $U=N(v)$ if $d<\Delta$ and $U\subseteq N(v)$ with $|U|=\Delta$ otherwise.
		\State Let $S=\{v\}$ with probability $\gamma_d$ and $S=U$ otherwise.
		\State Return $\text{\sc EnhancedVC3*}(G\setminus S)\cup S$
		\label{eVC3:randbranch} 
		\EndIf
		\State If $G$ is a regular graph (of degree $\delta$), select an arbitrary edge $(v_1,v_2) \in E$. Evaluate $S_1= \text{\textsc{EnhancedVC3*}}(G\setminus \{v_1\}) \cup \{v_1 \}$  and $S_2= \text{\textsc{EnhancedVC3*}}(G\setminus \{v_2 \}) \cup \{v_2 \}$. Return the smaller set between $S_1$ and $S_2$.  \label{eVC:regular}
		
	\end{algorithmic}
\end{algorithm}

Consider Algorithm \ref{alg:EnVC3}. 
 It can be shown that its running time is polynomial (similar to the proof of Lemma \ref{lem:bvc_poly} in Section \ref{sec:VC}).
 The 
probability that the algorithm returns a solution of size $b$, given that $G\in \mathcal{G}_k$, is at least   
\begin{equation}
\label{eq:EVC3_rec}
p(b,k) = \min \begin{cases}
\gamma_d\cdot p(b-1, k-1) ~+~ (1-\gamma_d)\cdot  p(b-d, k-1) &  2\leq  d < \Delta, d\neq \delta \\ 
\gamma_d \cdot p(b-1, k) ~+~ (1-\gamma_d)\cdot  p(b-d, k-d) &  2\leq  d <\Delta, d\neq \delta, \\ 
\gamma_\Delta \cdot  p(b-1, k-1) ~+~ (1-\gamma_\Delta) \cdot p(b-\Delta, k) &  \\ 
\gamma_\Delta \cdot  p(b-1, k) ~+~ (1-\gamma_\Delta) \cdot p(b-\Delta, k-\Delta) &   \\ 
\end{cases}
\end{equation}
As before, we use the lower bound derived from the recurrence relation to obtain a random parameterized $\alpha$-approximation algorithm  with running time $\tO\left( \frac{1}{p(\alpha k , k)}\right)$ 
by using Algorithm \ref{alg:alpha_approx} with Algorithm~\ref{alg:EnVC3} as $\mathcal{A}$  and the recurrence relation $p$ as given in~\eqref{eq:EVC3_rec}.
Let \textsc{$\alpha$-EnhancedVC3*} be this algorithm.

As in previous cases, the critical ratio of all the terms in \eqref{eq:EVC3_rec} is equal to $1$. 
For any $1<\alpha<2$ and $2\leq d \leq \Delta$ we can find 
the value $M_{\alpha,d}$ as in Section \ref{sec:VC3star}. 
If $\delta' = \argmax_{2\leq d \leq N } M_{\alpha,d} \neq  \Delta$ we 
can set $\delta =\delta'$; therefore, the run time of 
\textsc{$\alpha$-EnhancedVC3*} is $\tO(\exp(M_{\alpha}+\eps)^k)$ when $M_{\alpha}$ is the {\em second} largest
number of $M_{\alpha,2},\ldots, M_{\alpha, \Delta-1}$ (or $M_{\alpha, \Delta}$ if $\delta'=\Delta$). The value of $\exp(M_{\alpha})$ as a function of $\alpha$ is shown in Figure~\ref{fig:warmup}.
For $\alpha=1.5$ the run time of the algorithm is $\tO(1.01657^k)$.  This is the best running time for the specific approximation ratio presented in this paper.  The following table compares the running times of {\sc $\alpha$-EnhanvedVC3*} and {\sc $\alpha$-VC3*} for several values of $\alpha$.

\begin{center}
	\begin{tabular}{ c|| c |c |c | c | c | c | c | c|  c|} 
		$\alpha$ & $1.2$ &  $1.3$ &  $1.4$ &  $1.5$ &  $1.6$ & $1.7$  \\
		\hline
		{\sc $\alpha$-VC3*} &
		$1.12548^k$ &
		$1.06804^k$ &
		$1.03501^k$ &
		$1.01713^k$ &
		$1.00754^k$ &
		$1.00280^k$ \\
		\hline 
		{\sc $\alpha$-EnhancedVC3*} &
		
		$1.12386^k$ &
		$1.06420^k$ & 
		$1.03320^k$ &
		$1.01657^k$ &
		$1.00751^k$ &
		$1.00277^k$  
	\end{tabular}
\end{center}

\section{Application to $3$-Hitting Set}
\label{sec:HS}

\begin{figure}
	\centering
	\begin{subfigure}{.5\textwidth}
		\centering

		\skipfig{
				\begin{tikzpicture}
\node (v1) at (0,2) {};
\node (v2) at (1.5,3) {};
\node (v3) at (4,2.5) {};
\node (v4) at (0,0) {};
\node (v5) at (3,0.5) {};

\begin{scope}[fill opacity=0.8]
\filldraw[fill=green!10] ($(v1)+(-1,0)$) 
to[out=90,in=180] ($(v2) + (0,0.5)$) 
to[out=0,in=90] ($(v3) + (1,0)$)
to[out=270,in=0] ($(v2) + (1,-0.8)$)
to[out=180,in=270] ($(v1)+(-1,0)$);
\filldraw[fill=red!10]($(v5) +(+1,0.2)$)
to[out=270,in=0] ($(v5) + (+0.0,-0.5)$) 
to[out=180,in=270] ($(v5) + (-0.5,-0)$) 
to[out=90,in=180] ($(v3) + (0,+0.5)$) 
to[out=0,in=90] ($(v3) + (+0.5,+0.2)$) 
to[out=270,in=90] ($(v5) + (1, 0.2)$) ;
\filldraw[fill=blue!10] ($(v1)+(0,+0.5)$) 
to[out=45,in=180] ($(v2) + (0,0.2)$) 
to[out=0,in=90] ($(v4) + (0.5,0)$) 
to[out=270,in=0] ($(v4) + (0,-0.5)$) 
to[out=180,in=225]($(v1) + (0,+0.5)$);
\end{scope}

\foreach \v in {1,2,...,5} {
	\fill (v\v) circle (0.1);
}

\fill (v1) circle (0.1) node [right] {$v_1$};
\fill (v2) circle (0.1) node [below left] {$v_2$};
\fill (v3) circle (0.1) node [left] {$v_3$};
\fill (v4) circle (0.1) node [below] {$v_4$};
\fill (v5) circle (0.1) node [below right] {$v_5$};

\end{tikzpicture}
		}
	\caption{}
	\label{fig:hg}
	\end{subfigure}%
	\begin{subfigure}{.5\textwidth}
		\centering
			\begin{tikzpicture}
\node (v1) at (0,2) {};
\node (v2) at (1.5,3) {};
\node (v3) at (4,2.5) {};
\node (v4) at (0,0) {};

\begin{scope}[fill opacity=0.8]
\filldraw[fill=green!10] ($(v2)+(-1,0)$) 
to[out=90,in=180] ($(v2) + (0,0.5)$) 
to[out=0,in=90] ($(v3) + (1,0)$)
to[out=270,in=0] ($(v2) + (1,-0.8)$)
to[out=180,in=270] ($(v2)+(-1,0)$);
\filldraw[fill=blue!10] ($(v2)+(0,+0.3)$) 
to[out=0,in=90] ($(v4) + (0.5,0)$) 
to[out=270,in=0] ($(v4) + (0,-0.5)$) 
to[out=180,in=180]($(v2) + (0,+0.3)$);
\end{scope}

\foreach \v in {2,3,4} {
	\fill (v\v) circle (0.1);
}

\fill (v2) circle (0.1) node [below left] {$v_2$};
\fill (v3) circle (0.1) node [left] {$v_3$};
\fill (v4) circle (0.1) node [below] {$v_4$};

\end{tikzpicture}
		\caption{}
		\label{fig:ng}
	\end{subfigure}
	\caption{An example of a neighbors graph. A hypergraph $H$ is illustrated in \ref{fig:hg}. The neighbors graph of $v_1$, $\NeGra(v_1)$, is given in \ref{fig:ng}. }
	\label{fig:ng_example}
\end{figure}
In this section we present a parameterized approximation algorithm for $3$-Hitting Set. 
The algorithm 
draws some ideas from 
VC3* (see Section \ref{sec:VC3star}), which relies on two basic
observations. The first is that for any vertex $v$ of a graph $G$ and a vertex cover $S$, either  $v\in S$ or $N(v)\subseteq S$. The second observation
 is that, even if $v$  is in a minimum vertex cover, removing $N(v)$ from the graph decreases the size of a minimum cover at least by one. 

Consider the following analog of the above statement for $3$-Hitting Set. Given a $3$-hypergraph $H=(V,E)$, 
for any $v\in V$ define the {\em neighbors graph} of $v$ as the hypergraph $\NeGra(v)= (V_v, E_v)$ with $V_v=\left\{u\in V\setminus \{v\} \mid|~ \exists e\in E: u,v \in e\right\}$ and $E_v =\{ e\setminus \{v\}|~ e \in E, v\in e \}$ (see an example in Figure \ref{fig:ng_example}). Clearly, for every $e\in \NeGra(v)$ it holds that $|e|\leq 2$ (the neighbors graph is essentially a standard undirected graph with the addition of single node edges). Similar to the case of Vertex Cover, for any $v\in V$ and a hitting set $S$ of $H$, either $v\in S$ or there is a  minimal hitting set $T$ of $\NeGra(v)$ such that $T \subseteq S$.\footnote{A set $T$ is a minimal hitting set of a hypergraph $H$ if $T$ is a hitting set and no strict subset $T'\subsetneq T$ is also a hitting set of $H$.} Also, if $v$ belongs to a minimum hitting set of $H$ then removing a minimal hitting set of $\NeGra(v)$  from $H$  decreases the minimum hitting set size at least by $1$. 

An {\em isomorphism} between two hypergraphs $H_1=(V_1, E_1)$ and $H_2=(V_2,E_2)$ is a bijection $\varphi:V_1\rightarrow V_2$ such that $E_2= \{ \varphi(e)~|~e\in E_1\}$; we use the notation $\varphi(S) = \{\varphi(s)\,|\,s\in S\}$.
That is, $\varphi$ maps an edge in $E_1$ to an edge in $E_2$, and $\varphi^{-1}$, the inverse of $\varphi$, maps an edge in $E_2$ to an edge in $E_1$. 
 We say $H_1$ and $H_2$ are {\em isomorphic} if there is an isomorphism between $H_1$ and $H_2$. 

Given a vertex $v\in V$ we define $\deg_H(v)=\left|\{e\in E\,|\, v\in E\}\right|$ to be the number of hyperedges of $H$ which contain $v$. If $H$ is known by context then we use $\deg(v)=\deg_H(v)$. 
Let $v\in V$ such that $\{v\}\notin E$, then
the neighbors graph of $v$ admits a specific structure. It has up to $2\cdot\deg(v)$ vertices, exactly $\deg(v)$ edges (there may be edges with a single vertex) and no isolated vertices. Therefore, the number of possible graphs $\NeGra(v)$ for vertices of bounded degree is finite up to isomorphism.

For some fixed $\Delta \in \mathbb{N}$, we 
construct a set $\cG_\Delta$ of hypergraphs, such that $\NeGra(v)$ is isomorphic to a hypergraph in $\cG_\Delta$ for any $v$ with $\deg(v)\leq \Delta$. 
Let $\cG'_\Delta$ be the set of  hypergraphs $(V,E)$ with no isolated vertices, such that $V\subseteq \{1, 2, \ldots, 2 \Delta \}$, $|E| \leq \Delta$, and $|e| \leq 2$ $\forall e\in E$.
Let $\cG_\Delta \subseteq \cG'_\Delta$ be a minimal set of hypergraphs such that for any $G' \in \cG'_\Delta$ there is $G\in \cG_\Delta$ that is isomorphic to $G'$. Thus,
$\cG_\Delta$ can be derived from $\cG'_\Delta$ by removing isomorphic hypergraphs.  It is easy to see that the set $\cG_\Delta$ is finite.
Also, for every $G\in \cG_{\Delta}$ let $C^{G}_{1}, \ldots, C^{G}_{m^G}$ be all the minimal hitting sets of $G$.
Clearly, the set $\{ C^G_{i} |~ G\in \cG_{\Delta}, 1\leq i \leq m^G \}$ has a finite cardinality.

 We need one more technical definition before introducing our algorithm. Given a $3$-hypergraph $H=(V,E)$, a vertex $v\in V$ and $F\subseteq E$ such that $v\in e$ for any $e\in F$, define the {\em induced graph}  of $v$ and $F$ as the hypergraph $\Ind(v,F) = (V_{v,F},E_{v,F})$ with 
 $V_{v,F} = \left\{ u |~ \exists e\in F: u\in e \setminus\{ v \} \right\}$ and 
$E_{v,F} = \{e\setminus \{v\}|~e \in F\}$. By definition, it also holds that the cardinality of edges in $\Ind(v,F)$ is at most $2$ and $\Ind(v,F)$ has no isolated vertices (a vertex $u$ may only be contained in the hyperedge $\{u \}$). It follows that 
$\NeGra(v)= \Ind(v, \{e\in E|~ v\in e\})$. 
Our algorithm uses
induced graphs to handle vertices of degree larger than $\Delta$. Similar to the neighbors graph, the induced graph $\Ind(v,F)$ satisfies the following. Let $S$ be a hitting set of the hypergraph $H$, then
either $v\in S$ or there is a hitting  set $T$ of $\Ind(v,F)$ such that $T \subseteq S$.

{\tiny 
	\begin{algorithm}
		\caption{\sc 3HS} 
		\label{alg:3HS}
		\noindent \textbf{Input:} A $3$-hypergraph $H=(V,E)$\\
		\noindent \textbf{Configuration Parameters:} $\bgam^G \in \nonneg^{m^G+1}$ with $\sum_{i=1}^{m^G+1} \bgam^G_i = 1$ for any $G\in \cG_{\Delta}$. \\
		\noindent \textbf{Notation:} Define $H\setminus U=(V',E')$ with $V'=V\setminus U$ and $E' = \{e\in E|~ e\cap U = \emptyset\}$
		
		\begin{algorithmic}[1]
			\State \label{3hs:empyset}If the empty set is a hitting set of $H$ return $\emptyset$. 
			\State \label{3hs:empyset}If there is $\{v\}\in E$ then return 
			{\sc 3HS}$(H\setminus \{v\})\cup\{v\}$.
			\State \label{3hs:pickv}Pick an arbitrary vertex $v$. If $\deg(v)\leq \Delta$ set $N=\NeGra(v)$. Otherwise, set $N=\Ind(v,F)$ with an arbitrary set $F\subseteq E$ of $\Delta$ edges such that $\forall e\in F: v\in e$.
			\State \label{3HS:iso}Find a hypergraph $G\in \cG_\Delta$ such that $N$ and $G$ are isomorphic. Let $\varphi$ be a vertex isomorphism function from $G$ to $N$.
			\State  
			\label{3HS:rand}Select $S=\{v\}$ with 
			probability $\bgam^G_{m^G+1}$ and 
			$S=\varphi(C^G_i)$ with probability $\bgam^G_i$ for $1\leq i \leq m^G$. 
			Return {\sc 3HS}$(H\setminus S)\cup S$.
		\end{algorithmic}
\end{algorithm}}

The above observations are used to derive Algorithm \ref{alg:3HS}. It is easy to see that the algorithm  always returns a hitting set of the input hypergraph $H$. Also, the size of $H$ strictly decreases between recursive calls, and the processing time of each recursive call is polynomial. Therefore, the algorithm has polynomial running time (note that since $\Delta$ is a fixed constant, finding a graph $G$ isomorphic to $N$ takes  constant time). It is also easy to verify the algorithm indeed always finds a hypergraph
 $G\in \cG_{\Delta}$ isomorphic to $N$ in Line~\ref{3HS:iso}. 

Consider the following recurrence relation:
\begin{equation}
\label{eq:3HS}
\begin{aligned}
&p(b,k) = &\min
\begin{cases}
\bgam^G_{m^G+1} \cdot p\left(b-1, k\right) + \\
~~~~~+ \sum_{i=1}^{m^G} \bgam^G_i \cdot p\left(b-|C^G_i|, k - |C^G_i \cap C^G_j| \right) & \forall G\in \cG_\Delta, 1\leq j \leq m^G\\
\bgam^G_{m^G+1} \cdot  p\left( b-1, k-1\right) +\\   ~~~~~+\sum_{i=1}^{m^G} \bgam^G_i \cdot p\left(b- |C^G_i| , k -\one_{\|G\|<\Delta}\right) 
 &\forall G\in \cG_\Delta, 1\leq j \leq m^G \\
 p(b-1, k-1) &  
\end{cases} 
\end{aligned}
\end{equation}
 Also, $p(b,k) = 0$ for $b<0$, and $p(b,k) =1$ for $b\geq 0$ and $k\leq 0$. 
 Let  $\|G\|$ be the number of edges in $G$. We set $\one_{\|G\|<\Delta} =1 $ if $\|G\|<\Delta$ and $\one_{\|G\|<\Delta}=0$ otherwise.
 Let $P(b,H)$ be the probability that Algorithm \ref{alg:3HS} returns a hitting set of size $b$ or less, given the $3$-hypergraph $H$. With a slight abuse of notation, let $P(b,k)$ the minimal (infimum) value of $P(b,H)$ for a $3$-hypergraph $H$ which has a hitting set of size $k$ or less. The next lemma follows easily
 from the above discussion. We give a formal proof for completeness.
 
 \begin{lemma}
 \label{lem:3HS_P_to_p}
 For every $b\in \mathbb{Z}$ and $k\in \mathbb{N}$, $P(b,k) \geq p(b,k)$.
 \end{lemma}
\begin{proof}	
We prove the claim by induction on $b$.
For $b<0$ we have $P(b,k)=0=p(b,k)$, therefore the claim holds. 
For $b\in \mathbb{N}$, assume the claim holds for any smaller value of $b$. Let $k\in \mathbb{N}$, and $H$ a $3$-hypergraph with a hitting set $T$, $|T|\leq k$. 
If the algorithm returns $\emptyset$ (Line~\ref{3hs:empyset} of the algorithm)  then $P(b,H) =1\geq p(b,k)$.   Also, if there is an edge $\{v\} \in E$ then $v\in T$ (otherwise it is not an hitting set), and therefore $T\setminus \{v\}$ is a hitting set of $H\setminus\{v\}$. Thus, 
$$P(b,H)\geq P(b-1, H\setminus\{v\}) \geq 
P(b-1,k-1)\geq p(b-1, k-1) \geq p(b,k).$$

Otherwise, let $v$ be the vertex selected in Line \ref{3hs:pickv} of the algorithm, let $N$ be the selected hypergraph, $G\in \cG_{\Delta}$ the hypergraph isomorphic to $N$, $\varphi$ the vertex isomorphism from $G$ to $N$, and $S$ the randomly selected set in Line~\ref{3HS:rand}.

If $v\in T$, note that the set $T\setminus \{v\}$ is a hitting set of  $H\setminus \{v\}$; thus, $H\setminus \{v\}$ has a hitting set of size
 $ k-1$ (or less). Also, if it further holds that  $\|G\|<\Delta$ then $N = \NeGra(v)$. In this case, we have 
 that $T\setminus \{v\}$ is a hitting set of $H\setminus \varphi(C^G_i)$ for all $1\leq i \leq m^G$.
 Let $e$ be an edge in $H\setminus \varphi(C^G_i)$. If $v\in e$ then $e\setminus \{v\}$ is an edge in $N$. As $\varphi(C^G_i)$ is a hitting set of $N$  we have $e\cap \varphi(C^G_i)\neq \emptyset$; thus, $e$ cannot be an edge in $H\setminus \varphi(C^G_i)$. If $v \notin e$ then since $e\cap T \neq \emptyset$, we also have $e\cap (T\setminus\{v\}) \neq \emptyset$. 
 It follows that the probability Algorithm~\ref{alg:3HS} returns a hitting set of size $b$ or less given $H$ is at least
 \begin{equation*}
 	\begin{aligned}
 	P(b,H) \geq 
 	 &\bgam^G_{m^G+1} \cdot  P(b-1, H\setminus\{v\}) + \sum_{i=1}^{m^G } \bgam^G_i  \cdot P\left(b-|C^G_i|, 
 	 H\setminus \varphi(C^G_i) \right)\\
 \geq&\bgam^G_{m^G+1}\cdot  P(b-1, k-1) + \sum_{i=1}^{m^G } \bgam^G_i \cdot  P\left(b-|C^G_i|, k-\one_{\|G\|<\Delta}\right)\\
 \geq  &\bgam^G_{m^G+1} \cdot  p(b-1, k-1) + \sum_{i=1}^{m^G } \bgam^G_i \cdot  p\left(b-|C^G_i|, k-\one_{\|G\|<\Delta}\right)
\geq p(b,k).
\end{aligned}
 \end{equation*}

It remains to handle the case where $v\notin T$. Let
$F$ be the set of edges selected in Line \ref{3hs:pickv} of the algorithm if $\deg(v)>\Delta$, and 
 $F= \{e\in E| v\in e \}$ if $\deg(v)\leq \Delta$ . Then $N=\Ind(v,F)=(V_{v,F}, E_{v,F})$.  For any $e\in E_{v,F}$ it holds that $e\cup \{v\} \in F$; therefore, $e\cap T = (e\cup \{v\})\cap T\neq \emptyset$. Thus, $T$ contains a set $T_v\subseteq T$ such that $T_v$ is a hitting set of $N$. W.l.o.g., we may assume that $T_v$ is a minimal hitting set. Then $\varphi^{-1}(T_v)$ is a minimal vertex cover of $G$. Hence, there is $1\leq j \leq m^G$ such that $\varphi^{-1}(T_v)= C^G_j$, and equivalently $T_v = \varphi(C^G_j)$. 

The hypergraph $H\setminus S$ has a hitting set of size $|T \setminus (T\cap S)| \leq k - |T \cap S|$. 
For $S=\{v\}$ we have $|T\cap S|= |\emptyset|=0$, and 
for $S=\varphi(C^G_i)$,
$$|T\cap S| \geq |T_v \cap S| = | \varphi(C^G_j) \cap \varphi(C^G_i)| = |C^G_j \cap C^G_i|.$$
Therefore,
\begin{equation*}
\begin{aligned}
P(b,H)\geq 
&\bgam^G_{m^G+1} \cdot P(b-1, H\setminus\{v\}) + \sum_{i=1}^{m^G } \bgam^G_i\cdot  P\left(b-|C^G_i|, 
H\setminus \varphi(C^G_i) \right)\\
\geq&\bgam^G_{m^G+1} \cdot P\left(b-1, k\right) 
+ \sum_{i=1}^{m^G} \bgam^G_i \cdot P\left(b-|C^G_i|, k - |C^G_i \cap C^G_j| \right) \\
\geq& 
\bgam^G_{m^G+1} \cdot p\left(b-1, k\right) 
+ \sum_{i=1}^{m^G} \bgam^G_i \cdot p\left(b-|C^G_i|, k - |C^G_i \cap C^G_j| \right) \geq p(b,k)
\end{aligned}
\end{equation*}

Hence, $P(b,H) \geq p(b,k)$ for any $3$-hypergraph  $H$ with a hitting set of size $k$ or less. We conclude that $P(b,k)\geq p(b,k)$. 

\end{proof}

Following the above analysis, an $\alpha$-approximation algorithm for $3$-Hitting Set can be derived by the same approach used for Vertex Cover. This leads to Algorithm~\ref{alg:alpha_3HS}. 

\begin{algorithm}
	\caption{\sc $\alpha$-HS} 
	\label{alg:alpha_3HS}
	\hspace*{\algorithmicindent} \noindent \textbf{Input:} A $3$-hypergraph $H$, a parameter $k$
	
	\begin{algorithmic}[1]	
		\State Evaluate $r=p(\floor{\alpha k} , k)$ using dynamic programming ($p$ is defined in \eqref{eq:3HS}).
		\Loop { $\ceil{r^{-1}}$ times}
		\State Execute $\text{3HS}(H)$.
		\EndLoop 
		\State Return the minimal hitting set found.
	\end{algorithmic}
\end{algorithm} 

It follows from Lemma~\ref{lem:3HS_P_to_p} that Algorithm~\ref{alg:alpha_3HS} yields an $\alpha$-approximation for $3$-Hitting Set with running time of $\frac{1}{p(\alpha k , k )}$. For any value of $\alpha$, it is possible to optimize the value of $\bgam^G$ for each $G\in \cG_\Delta$ and evaluate the asymptotic behavior of $p(\alpha k , k)$ as $k$ goes to infinity using Theorem \ref{thm:rec}.

However,  the size of $\cG_\Delta$ grows rapidly as $\Delta$ increases, rendering the above computation less and less practical. With a little technical sophistication we were able to  evaluate the running time of the algorithm with $\Delta=7$ for various approximation ratios. Figure \ref{fig:HS} shows the running times of the algorithm with $\Delta=7$ as 
function of $\alpha$. A list of running times for several approximation ratios is given in the table below. For $\alpha = 2$ the running time is $\tO(1.0659^k)$, yielding a significant improvement over
the previous best result of $\tO(1.29^k)$ due to \cite{BF12}.  
\begin{center}
	\begin{tabular}{ c|| c |c |c | c | c | c | c | c|  c| c} 
		$\alpha$ %
		& $1.2$ & $1.4$ &  $1.6$ &  $1.8$ &  $2.0$ &  $2.2$ & $2.4$ & $2.6$ & $2.8$ \\
		\hline
		$\alpha-HS$ &
		$1.59^k$ &
		$1.29^k$ &
		$1.18^k$&
		$1.11^k$&
		$1.0659^k$ & 
		$1.039^k$& 
		$1.021^k$ &
		
		$1.0085^k$ &
		$1.0026^k$ 	
	\end{tabular}
\end{center}

\section {Advanced  Randomized Branching  for Vertex Cover
}
\label{sec:VC}

In this section we give a parameterized approximation algorithm for Vertex Cover building on the exact $\tO(1.33^k)$ algorithm 
presented in~\cite{Nied06}.
That is, we analyze below a variant of the algorithm 
in which branching is replaced by selection of
one of the branches {\em randomly}. The analysis
shows that randomized branching in conjunction with faster parameterized
algorithms can lead to faster parameterized approximation algorithms. 
We use below ideas presented in Section \ref{sec:warmup} and give the technical details for their implementation in a more advanced settings.

{\sc BetterVC}, depicted in Algorithm~\ref{alg:betterVC}, involves  several branching rules. For the restricted cases of a regular connected graphs of degree $2$, $3$ or $4$ the algorithm resorts to standard deterministic branching.\footnote{A graph is {\em $d$-regular} if the degree of all vertices is $d$. The graph is {\em regular} if there is $d\in\mathbb{N}$ for which it is $d$-regular. } As we show in the analysis, despite this use of deterministic branching the algorithm remains polynomial.  This is a consequence of the fact that a connected regular graph of degree $d$  cannot have a connected regular graph of degree $d$ as a strict vertex induced subgraph.  The algorithm preserves a simple invariant: all of its recursive calls only replace the input graph $G$ with a  vertex induced subgraph of $G$. 

The algorithm first handles simple cases. If the input graph has no edges, it returns the empty set as a cover. If the graph is not connected, then the algorithm makes two recursive calls: one with a single connected component and another with the remainder of the graph. If there is a vertex of degree 
$1$, its neighbor is added to the solution. 

Following these simple cases, the algorithm checks if there is a vertex $v$ of degree $5$ or more. If there is such a vertex, the algorithm randomly picks either $v$ or $N(v)$ for the solution. Similar to Algorithm~\ref{alg:VC3star}, the probability by which the algorithm selects $v$ or $N(v)$ depends on the degree of $v$. Furthermore, if the degree of $v$ is higher than a given fixed threshold $\Delta$, the algorithm only  selects~$\Delta$ vertices from $N(v)$ for the solution. 

If none of the above can be applied and the graph is regular, it is in particular a connected $d$-regular graph where $d\in \{2,3,4\}$. In such cases the algorithm applies {\em deterministic} branching. It picks an arbitrary edge $(v_1,v_2)\in E$ and initiates two recursive calls. In one call the vertex $v_1$ is removed from the graph and forced into the solution, and in the other call $v_2$ is removed from the graph and added to the solution.  The algorithm eventually returns the smaller solution output by these recursive calls. As a minimum vertex cover must contain either $v_1$ or $v_2$, one of the recursive calls adds a vertex from a minimum vertex cover to the solution. 

If the graph does not satisfy any of the above conditions 
then the graph is connected, not regular, and with vertices of degrees $2$, $3$ and $4$. In particular, this implies that the graph either has vertex of degree $2$, or has a vertex of degree $3$ with a neighbor of degree $4$. The algorithm then searches for  one of several reduction rules and randomized branching rules which may be applicable to the graph. 
It first attempts to apply reduction rules and branching for  vertices of degree $2$.
If there is no such vertex,  then the algorithm checks the applicability of  randomized branching rules for  degree $3$ vertices whose  neighborhoods satisfy some additional properties.
If neither is applicable the algorithm finds a vertex of degree $3$ which has a neighbor of degree $4$ and applies randomized branching. 
The algorithm uses different probabilities for every case in which randomized branching is applied.

\begin{algorithm}
	\renewcommand\algorithmicthen{}
	\caption{\sc BetterVC}
	\label{alg:betterVC}
	\textbf{Input:} An undirected graph $G = (V, E)$
	
	\textbf{Parameters:} The configuration parameters are:
	\begin{itemize}[nosep]
		\item $\Delta \in \mathbb{N}$
		\item $\gamma_5, \gamma_6, \ldots, \gamma_\Delta \in (0, 1)$
		\item $\lambda_{1, r} \in (0, 1)$ for every $3 \leq r \leq 7$
		\item $\lambda_{2, r} \in (0, 1)$ for $3 \leq r \leq 4$
		\item $\lambda_3 \in (0, 1)$
		\item $\delta_{r, 1}, \delta_{r, 2}, \delta_{r, 3} \in [0, 1]$ with $\delta_{r, 1} + \delta_{r, 2} + \delta_{r, 3} = 1$ for $r \in \{5, 6, 7\}$
	\end{itemize}
	
	\textbf{Notation:} 
	\begin{itemize}[nosep]
		\item  We use the term {\bf branch over} $U_1, \ldots, U_r$ with probabilities $p_1, \ldots, p_r$ to denote the operation of returning  $\text{\textsc{BetterVC}}(G \setminus U_i) \cup U_i$ with probability $p_i$. 
		\item The term {\bf select} $U$ denotes the operation of returning $\text{\textsc{BetterVC}}(G \setminus U) \cup U$.
	\end{itemize}
	
	\begin{algorithmic}[1]
		\State \textbf{Trivial Case:} if the empty set is a vertex cover of $G$, return $\emptyset$.
		
		\State \textbf{Disconnected Graph:} if $G$ is not connected, let $G'$ be a connected component of $G$ and $G'' = G - G'$. Return $\text{\textsc{BetterVC}}(G') \cup \text{\textsc{BetterVC}}(G'')$. \label{bVC:split}
		
		\State \textbf{Degree $1$ Vertex:} if $G$ has a vertex $v$ of degree $1$, let $u$ be its neighbor. Select $u$ to the cover. \label{bVC:deg1}
		
		\State \textbf{High-Degree Vertex:} if $G$ has a vertex $v$ of degree $d \geq 5$:
		\begin{itemize}
			\item Let $U = N(v)$ if $d < \Delta$, otherwise let $U \subseteq N(v)$ with $|U| = \Delta$.
			\item Branch over $\{v\}$ and  $U$ with probabilities $\gamma_d$ and $ 1 - \gamma_d$ (or $\gamma_\Delta$ and  $1 - \gamma_\Delta$ if $d \geq \Delta$). \label{bVC:deg5p}
		\end{itemize}
		
		\State \textbf{Regular Graph:} if $G$ is a regular graph, select an arbitrary edge $(v_1, v_2) \in E$. 
		\begin{itemize}
			\item Evaluate $S_1 = \text{\textsc{BetterVC}}(G \setminus \{v_1\}) \cup \{v_1\}$ and $S_2 = \text{\textsc{BetterVC}}(G \setminus \{v_2\}) \cup \{v_2\}$.
			\item Return the smaller set between $S_1$ and $S_2$. \label{bVC:regular}
		\end{itemize}
		
		  \If {$G$ has a vertex $v$ of degree $2$, $N(v) = \{x, y\}$:}
		\State \textbf{Degree $2$ Triangle:} if $(x, y) \in E$ {\bf then} select $\{x, y\}$ to the cover. \label{bVC:deg2case1}
		\State \textbf{Degree $2$ Diamond :} if $\deg(x) = \deg(y) = 2$ and $N(x) = N(y) = \{z, v\}$ {\bf then} 
		\begin{itemize}[nosep,left=2em] \item select $\{z, v\}$ to the cover. \label{bVC:deg2case2}
			\end{itemize}
		\State \textbf{Degree $2$ Branching:} if none of the above holds, {\bf then } \begin{itemize}[nosep, left=2em]\item let $r = |N(x) \cup N(y)|$ and branch over $N(v)$ and  $N(x) \cup N(y)$ with probabilities $\lambda_{1, r}$ and  $1 - \lambda_{1, r}$. \label{bVC:deg2case3}
			\end{itemize}
		\EndIf

		\If{$G$ has a vertex $v$ of degree $3$ such that  $N(v) = \{x, y, z\}$ and $(x,y)\in E$:} 
		\State {\bf Degree $3$ Triangle:} branch over $N(v)$ and  $N(z)$ with probabilities $\lambda_{2, r}$ and $ 1 - \lambda_{2, r}$ where $r = |N(z)|$. \label{bVC:deg3case1}
		\EndIf
		
			\If{$G$ has a vertex $v$ of degree $3$ such that  $N(v) = \{x, y, z\}$ and  there is $w \notin N(v) \cup \{v\}$ with $x, y \in N(w)$:} 
		\State {\bf Degree $3$ Diamond:} branch over $N(v)$ and  $\{v, w\}$ with probabilities $\lambda_3$  and  $1 - \lambda_3$. \label{bVC:deg3case2}
		\EndIf

		\State \textbf{Degree $4$ Branching:} find a vertex $v$ of degree $3$ with  $N(v) = \{x, y, z\}$ and $\deg(x) = 4$:
		\begin{itemize}
			\item Let $r = |N(y) \cup N(z)|$ and branch over $N(v)$, $N(x)$ and $\{x\} \cup N(y) \cup N(z)$ with probabilities $\delta_{r, 1}$,  $\delta_{r, 2}$ and  $\delta_{r, 3}$. \label{bVC:deg3case3}
		\end{itemize}
	\end{algorithmic}
\end{algorithm}

It is easy to see that Algorithm~\ref{alg:betterVC} always returns a cover of the input graph~$G$. Furthermore,
\begin{lemma}
	\label{lem:bvc_poly}
Algorithm \ref{alg:betterVC} has a polynomial running time. 
\end{lemma}
The proof of Lemma \ref{lem:bvc_poly}  is given at the end of this section, along with the proof of the next lemma.
\begin{lemma}
	\label{lem:bvc_prob}
Let $G\in \mathcal{G}_k$ ($\mathcal{G}_k$ is the set of graphs with vertex cover of size $k$ or less), then the probability that Algorithm \ref{alg:betterVC} returns a 
cover of size $b$ or less is greater or equal to $p(b,k)$, where 
\small
\begin{equation}
\label{eq:bvc_rec}
\begin{aligned}
&p(b,k) = \min %
&\begin{cases}
p(b-1, k-1) \\
p(b-2, k-2) & k\geq 2 \\
\gamma_d\cdot p(b-1, k-1) ~+~ (1-\gamma_d)\cdot  p(b-d, k-1) &  5\leq  d <\Delta \\ 
\gamma_d \cdot p(b-1, k) ~+~ (1-\gamma_d)\cdot  p(b-d, k-d) &  5\leq  d \leq \Delta \\ 
\gamma_\Delta \cdot  p(b-1, k-1) ~+~ (1-\gamma_\Delta) \cdot p(b-\Delta, k) &  \\ 
\lambda_{1,r} \cdot p(b-2, k-2) + (1-\lambda_{1,r}) \cdot p( b-r, k-2) & 3\leq r \leq 7\\
\lambda_{1,r} \cdot p(b-2, k-1) + (1-\lambda_{1,r}) \cdot p( b-r, k-r) & 3\leq r \leq 7,\\
\lambda_{2,r} \cdot p (b-3, k-3) + (1-\lambda_{2,r}) \cdot p \left(
b-r, k- 1
\right) & 3 \leq r \leq 4 \\
\lambda_{2,r} \cdot p (b-3, k-1) + (1-\lambda_{2,r}) \cdot p \left(
b-r, k- r
\right) & 3\leq r \leq 4\\
\lambda_{3} \cdot p (b-3, k-3) + (1-\lambda_{3}) \cdot p(b-2, k)  & \\
\lambda_{3} \cdot p (b-3, k-1) + (1-\lambda_{3}) \cdot p (
b-2, k-2) & \\
\delta_{r,1} \cdot
 p (b-3, k-3)  + \delta_{r,2} \cdot p(b-4, k-1 ) + \delta_{r,3}  \cdot p(b-r-1, k-3) &  5\leq r \leq 7 \\
\delta_{r,1}\cdot p (b-3, k-1)  + \delta_{r,2} \cdot p(b-4, k-4 ) + \delta_{r,3} \cdot p(b-r-1, k-r) &  5\leq r \leq 7\\
\delta_{r,1}\cdot p (b-3, k-2)  + \delta_{r,2} \cdot p(b-4, k-4 ) + \delta_{r,3} \cdot p\left(b-r-1, k-1-\ceil{\frac{r}{2}}\right) & 5 \leq r \leq 7\\
\delta_{r,1}\cdot p (b-3, k-2)  + \delta_{r,2} \cdot p(b-4, k-2 ) + \delta_{r,3} \cdot p(b-r-1, k-r-1)  & 5\leq r \leq 7, 
\end{cases}
\end{aligned}
\end{equation}
\normalsize

and $p(b,k)=0$ for $b<0$, and $p(b,k)=1$ for $b\geq 0$ and $k\leq 0$. 
\end{lemma}

 The proof of Lemma \ref{lem:bvc_prob}
is %
a case by
case analysis similar to the one done in \cite{Nied06}. The main difference between the analysis presented here and the analysis in \cite{Nied06} is that here we also count the reduction in the minimal cover size in 
a non-optimal branching step.

Let $\alpha$-{\sf BetterVC} be the algorithm which executes Algorithm~\ref{alg:alpha_approx} with Algorithm \ref{alg:betterVC} as $\mathcal{A}$, and 
with $p$ as the recurrence in Lemma \ref{lem:bvc_prob}.
It follows from Lemma \ref{lem:bvc_prob} that $\alpha$-{\sf BetterVC} is a random parameterized $\alpha$-approximation algorithm for Vertex Cover,
with running time $\tO\left( \frac{1}{p(\alpha k, k)}\right)$. 
As before, we arbitrarily select $\Delta=100$. 
For every $1<\alpha<2$ and a set of configuration parameters, by Theorem \ref{thm:rec} we can numerically evaluate (see Section \ref{sec:numerical} for the details) a 
value $M_{\alpha}$
 such that $p(\alpha k , k )  > \exp(-M_\alpha-\eps)$ for any $\eps>0$ and large enough $k$.
 Similarly, for every $1<\alpha<2$ we can optimize the configuration parameters so this value 
 is minimized.
 Therefore, the running 
time of Algorithm $\alpha$-{\sc BetterVC} is
 $\tO(\exp(M_\alpha+\eps)^k)$ for any $\eps>0$. 
 Figure \ref{fig:bettervc} shows $\exp(M_\alpha)$ as a function of $\alpha$. 

\begin{figure}
\centering
\begin{tikzpicture}[scale = 1.0]
\begin{axis}[
xmin = 1, xmax=2, ymin =0.99 , ymax=1.5 , xlabel= 	{approximation ratio}, 
ylabel={exponent base}, samples=50]
\addplot[blue, ultra thick] (2-x, 1.2738^x ); 

\addplot[black, ultra thick] coordinates {
( 1.0  ,  1.2738 )
( 1.01  ,  1.2697966946339028 )
( 1.02  ,  1.2658059708770486 )
( 1.03  ,  1.2618277891878902 )
( 1.04  ,  1.2578621101491503 )
( 1.05  ,  1.2539088944674337 )
( 1.06  ,  1.249968102972836 )
( 1.07  ,  1.2460396966185563 )
( 1.08  ,  1.2421236364805097 )
( 1.09  ,  1.2382198837569434 )
( 1.1  ,  1.2343283997680499 )
( 1.11  ,  1.2304491459555849 )
( 1.12  ,  1.2265820838824855 )
( 1.13  ,  1.222727175232489 )
( 1.1400000000000001  ,  1.218884381809753 )
( 1.15  ,  1.215053665538477 )
( 1.16  ,  1.211234988462526 )
( 1.17  ,  1.2074283127450531 )
( 1.18  ,  1.2036336006681256 )
( 1.19  ,  1.199850814632351 )
( 1.2  ,  1.196079917156504 )
( 1.21  ,  1.1923208708771553 )
( 1.22  ,  1.188573638548303 )
( 1.23  ,  1.1848381830410002 )
( 1.24  ,  1.1811144673429903 )
( 1.25  ,  1.1774024545583384 )
( 1.26  ,  1.1737021079070669 )
( 1.27  ,  1.1700133907247903 )
( 1.28  ,  1.1663362664623513 )
( 1.29  ,  1.1626706986854614 )
( 1.3  ,  1.1590166510743358 )
( 1.31  ,  1.1553740874233371 )
( 1.32  ,  1.1517429716406147 )
( 1.33  ,  1.148123267747748 )
( 1.34  ,  1.1445149398793886 )
( 1.35  ,  1.1409179522829078 )
( 1.3599999999999999  ,  1.137332269318038 )
( 1.37  ,  1.1337578554565237 )
( 1.38  ,  1.130194675281768 )
( 1.3900000000000001  ,  1.1266426934884801 )
( 1.4  ,  1.1231018748823276 )
( 1.4100000000000001  ,  1.1195721843795874 )
( 1.42  ,  1.1160535870067976 )
( 1.43  ,  1.1125460479004097 )
( 1.44  ,  1.1090495323064467 )
( 1.45  ,  1.105564005580155 )
( 1.46  ,  1.1020894331856639 )
( 1.47  ,  1.0986257806956408 )
( 1.48  ,  1.0951730137909526 )
( 1.49  ,  1.091731098260324 )
( 1.5  ,  1.0883 )
};

\addplot[mark=*, mark options={scale=0.3}, only marks, black, forget plot] coordinates {
	(1.5,1.0883) 
	(1.666666, 1.04)
	(1.75, 1.024)
	(1.8, 1.017)
	(1.8333 , 1.01208191)
	( 1.875 ,  1.00734187 )
	( 1.88888888889 ,  1.0060641 )
	( 1.9 ,  1.00503861 )
	( 1.90909090909 , 1.00425981 )
	( 1.91666666667 , 1.00365343 )
	( 1.92307692308 , 1.00317141 )
	( 1.92857142857 , 1.00278148 )
	( 1.93333333333 , 1.00246127 )
	( 1.9375 ,  1.00221275)
	( 1.94117647059 , 1.00198584 )
	( 1.94444444444 ,  1.0017931)
	( 1.94736842105 , 1.0016279 )
	( 1.95 ,   1.00148516 )
};  %

\addplot[purple, ultra thick] coordinates {
( 1.000000001 , 1.4655712079951821 )
( 1.01 , 1.4142992556541911 )
( 1.02 , 1.3796905224309892 )
( 1.03 , 1.3513093010768376 )
( 1.04 , 1.3268738804015106 )
( 1.05 , 1.3052923325717478 )
( 1.06 , 1.2859207093324299 )
( 1.07 , 1.2683346312304615 )
( 1.08 , 1.2522340831375591 )
( 1.09 , 1.2373963835851272 )
( 1.1 , 1.2236502096756936 )
( 1.11 , 1.2108600605657813 )
( 1.12 , 1.1989163896586514 )
( 1.13 , 1.1877290359162285 )
( 1.1400000000000001 , 1.177222682475088 )
( 1.15 , 1.1673336177645628 )
( 1.16 , 1.1580073642831836 )
( 1.17 , 1.149196904071708 )
( 1.18 , 1.1408613245057746 )
( 1.19 , 1.1329647677869674 )
( 1.2 , 1.1238583247052474 )
( 1.21 , 1.1142350046542577 )
( 1.22 , 1.1051737061918148 )
( 1.23 , 1.0966406533062092 )
( 1.24 , 1.0890273363189389 )
( 1.25 , 1.0844032277463296 )
( 1.26 , 1.079987928111209 )
( 1.27 , 1.0757697035887084 )
( 1.28 , 1.0717378631574397 )
( 1.29 , 1.0678826368917937 )
( 1.3 , 1.0641950714823845 )
( 1.31 , 1.0606669408458678 )
( 1.32 , 1.0572906686129075 )
( 1.33 , 1.0540592609802601 )
( 1.34 , 1.0509662480292394 )
( 1.35 , 1.0480056324101277 )
( 1.3599999999999999 , 1.0447786392100875 )
( 1.37 , 1.0415217390754372 )
( 1.38 , 1.0384236422881525 )
( 1.3900000000000001 , 1.0354785607592514 )
( 1.4 , 1.0331956250110093 )
( 1.4100000000000001 , 1.0312594344906947 )
( 1.42 , 1.0294032711038863 )
( 1.43 , 1.027624274373305 )
( 1.44 , 1.0259197567144593 )
( 1.45 , 1.0242871892896974 )
( 1.46 , 1.022724190237944 )
( 1.47 , 1.0211061148691052 )
( 1.48 , 1.0194724241333957 )
( 1.49 , 1.0179187691522171 )
( 1.5 , 1.0165674569904897 )
( 1.51 , 1.0154641391259127 )
( 1.52 , 1.0144087091873695 )
( 1.53 , 1.0133997923158904 )
( 1.54 , 1.012436083058512 )
( 1.55 , 1.0114547596039911 )
( 1.56 , 1.0104819722592377 )
( 1.57 , 1.009576591546044 )
( 1.58 , 1.0088566788288071 )
( 1.5899999999999999 , 1.0081704955747772 )
( 1.6 , 1.007517193989135 )
( 1.6099999999999999 , 1.0068589214444632 )
( 1.62 , 1.0062168673941576 )
( 1.63 , 1.0056706102535804 )
( 1.6400000000000001 , 1.0051821501633222 )
( 1.65 , 1.004716811708091 )
( 1.6600000000000001 , 1.0042450021290232 )
( 1.67 , 1.003826644950794 )
( 1.6800000000000002 , 1.0034626045761648 )
( 1.69 , 1.0031096093998462 )
( 1.7000000000000002 , 1.0027685566578624 )
( 1.71 , 1.002483044291485 )
( 1.72 , 1.002207736466775 )
( 1.73 , 1.001951261875834 )
( 1.74 , 1.001728912143514 )
( 1.75 , 1.0015114858320462 )
( 1.76 , 1.0013276079407811 )
( 1.77 , 1.001150938002623 )
( 1.78 , 1.0009985138018003 )
( 1.79 , 1.0008579215400477 )
( 1.8 , 1.0007314241337721 )
( 1.81 , 1.0006210015967711 )
( 1.82 , 1.0005224850122387 )
( 1.83 , 1.0004353015474952 )
( 1.8399999999999999 , 1.000358773954474 )
( 1.85 , 1.0002921261368678 )
( 1.8599999999999999 , 1.0002352464447102 )
( 1.87 , 1.00018616407132 )
( 1.88 , 1.0001448402740132 )
( 1.8900000000000001 , 1.0001104438422828 )
( 1.9 , 1.0000820672258894 )
( 1.9100000000000001 , 1.0000592265304418 )
( 1.92 , 1.0000411663331896 )
( 1.9300000000000002 , 1.0000272851986407 )
( 1.94 , 1.0000170095656638 )
( 1.9500000000000002 , 1.0000097419997498 )
( 1.96 , 1.0000049368308817 )
( 1.97 , 1.000002061518123 )
( 1.98 , 1.0000005075929264 )
( 1.99 , 1.0000000000237768 )
( 1.9999999 , 1.000000000000005 )
};

\addplot[brown, ultra thick] coordinates {
( 1.000000001 , 1.324717943439684 )
( 1.01 , 1.2930971872855843 )
( 1.02 , 1.2709553562950033 )
( 1.03 , 1.2524114847946912 )
( 1.04 , 1.236184666676916 )
( 1.05 , 1.221659939266156 )
( 1.06 , 1.208472688759161 )
( 1.07 , 1.196381051359742 )
( 1.08 , 1.1852128031155162 )
( 1.09 , 1.1748391359811732 )
( 1.1 , 1.1651601693734248 )
( 1.11 , 1.1560962764757634 )
( 1.12 , 1.1475825675011355 )
( 1.13 , 1.1395652100002955 )
( 1.1400000000000001 , 1.1319988798097702 )
( 1.15 , 1.1248449379905316 )
( 1.16 , 1.118070092453454 )
( 1.17 , 1.1116453932805612 )
( 1.18 , 1.1056611400607466 )
( 1.19 , 1.1005318066137195 )
( 1.2 , 1.0956590093845138 )
( 1.21 , 1.0910239370064931 )
( 1.22 , 1.0866099008693249 )
( 1.23 , 1.082402020072795 )
( 1.24 , 1.0783869639245827 )
( 1.25 , 1.0745527396574905 )
( 1.26 , 1.0708885156125345 )
( 1.27 , 1.0673844735995928 )
( 1.28 , 1.0640316839126145 )
( 1.29 , 1.0608219997059647 )
( 1.3 , 1.0577479666683236 )
( 1.31 , 1.0548027455227311 )
( 1.32 , 1.0519800455276038 )
( 1.33 , 1.0492740664240956 )
( 1.34 , 1.0466794482170065 )
( 1.35 , 1.0441912273818021 )
( 1.3599999999999999 , 1.0418047983557137 )
( 1.37 , 1.0395158794903296 )
( 1.38 , 1.037320483486005 )
( 1.3900000000000001 , 1.0352148905347995 )
( 1.4 , 1.0331956250110093 )
( 1.4100000000000001 , 1.0312594344906947 )
( 1.42 , 1.0294032711038863 )
( 1.43 , 1.027624274373305 )
( 1.44 , 1.0259197567144593 )
( 1.45 , 1.0242871892896974 )
( 1.46 , 1.022724190237944 )
( 1.47 , 1.021228513342151 )
( 1.48 , 1.0197980379695943 )
( 1.49 , 1.0184307600921705 )
( 1.5 , 1.0171247840549282 )
( 1.51 , 1.0158783147937023 )
( 1.52 , 1.0146896510437236 )
( 1.53 , 1.013557179208429 )
( 1.54 , 1.012479367231861 )
( 1.55 , 1.011516340844781 )
( 1.56 , 1.010639386051077 )
( 1.57 , 1.0098040957328096 )
( 1.58 , 1.0090094003753804 )
( 1.5899999999999999 , 1.0082542805250032 )
( 1.6 , 1.0075377635891791 )
( 1.6099999999999999 , 1.006895964943023 )
( 1.62 , 1.0063060352372275 )
( 1.63 , 1.0057466660550909 )
( 1.6400000000000001 , 1.0052171503605938 )
( 1.65 , 1.0047190832683586 )
( 1.6600000000000001 , 1.0042808350358887 )
( 1.67 , 1.0038668550771435 )
( 1.6800000000000002 , 1.0034766153742065 )
( 1.69 , 1.0031191217343285 )
( 1.7000000000000002 , 1.0027957597677062 )
( 1.71 , 1.0024920990839112 )
( 1.72 , 1.0022148330287521 )
( 1.73 , 1.0019635708773775 )
( 1.74 , 1.0017289193650065 )
( 1.75 , 1.0015213162213954 )
( 1.76 , 1.0013281867976958 )
( 1.77 , 1.001156699295766 )
( 1.78 , 1.0009996039199218 )
( 1.79 , 1.000859765656144 )
( 1.8 , 1.000734329926821 )
( 1.81 , 1.0006222605709354 )
( 1.82 , 1.0005229789328225 )
( 1.83 , 1.000435616673092 )
( 1.8399999999999999 , 1.0003592006158748 )
( 1.85 , 1.0002927499540164 )
( 1.8599999999999999 , 1.0002353263160844 )
( 1.87 , 1.000186415978615 )
( 1.88 , 1.0001450330109274 )
( 1.8900000000000001 , 1.0001105029208743 )
( 1.9 , 1.0000821417456502 )
( 1.9100000000000001 , 1.0000592420239571 )
( 1.92 , 1.0000411668839475 )
( 1.9300000000000002 , 1.0000272937223253 )
( 1.94 , 1.000018379361036 )
( 1.9500000000000002 , 1.0000126966435747 )
( 1.96 , 1.000008083665902 )
( 1.97 , 1.0000045236998467 )
( 1.98 , 1.000002000313399 )
( 1.99 , 1.0000004975606898 )
( 1.9999999 , 1.0000000000001712 )
};
\addlegendentry[no markers, blue]{FKRS \cite{FKRS18}}
\addlegendentry[no markers, black]{BF \cite{BF13}}
\addlegendentry[no markers, purple]{\textsc{EnhancedVC3}*}
\addlegendentry[no markers, brown]{\textsc{BetterVC}}

\end{axis}
\end{tikzpicture}
\caption{The performance of \textsc{BetterVC}. A dot at $(\alpha,c)$ means that the respective algorithm yields $\alpha$-approximation with running time $\tO(c^k)$ or $\tO\left( (c+\eps)^k\right)$ for any $\eps>0$. }
\label{fig:bettervc}
\end{figure}

Note that the algorithm in \cite{NR99} can be used along with our framework of 
randomized branching. However, due to its technical complexity, we
preferred to use the algorithm in \cite{Nied06}, which can be viewed as a simplified 
version of the same algorithm. In the discussion we describe
the obstacles we encountered while attempting to obtain  randomized branching variants of faster algorithms.  

\subsection{Proofs}
\label{sec:vc_proofs}

\begin{proof}[Proof of Lemma \ref{lem:bvc_poly}]
To prove the algorithm has  polynomial running time,  it suffices to show that the number of recursive calls is polynomial. 
We note that the only non-trivial part of the proof is the handling of regular graphs in Line \ref{bVC:regular}. We use a simple  potential function to handle this case. 
For $i=2,3,4$, define $\Phi_i(G) =1$ if $G$ has a  non-empty $i$-regular vertex induced subgraph and $\Phi_i(G)=0$ otherwise. Also, define $\Phi(G) = \Phi_2(G)+\Phi_3(G)+\Phi_4(G)$.

Let $R(G)$ be the maximum number of recursive calls initiated in the execution of $\text{\textsc{BetterVC}}(G)$.
We now prove by induction (on $|V|$) that the number of recursive calls 
initiated by the algorithm is at most $R(G)\leq \max\left\{2(|V|-1)\cdot 2^{\Phi(G)},\,0\right\}$.
The idea behind the potential function $\max\left\{2(|V|-1)\cdot 2^{\Phi(G)},\,0\right\}$ is to  bound the incurred cost of the branching on regular graphs  in Step~\ref{bVC:regular}. When such branching occurs, the value of $\Phi(G)$ in the generated sub-instances must decrease by one, and the multiplicative factor of  $2^{\Phi(G)}$ in the potential function  can be ``charged''  to the recursive calls. The $2\cdot (|V|-1)$ factor in the potential function simply measures the size of the graph, and captures the idea that the graph size decreases between recursive calls.

If $|V|\leq 1$ the algorithm does not initiate recursive calls, and the claim holds. Each time a \textbf{Branch} or \textbf{Select} is used
the size of $|V|$ decreases by at least one, $\Phi(G)$ does not increase, and only one recursive call is initiated, therefore the claim 
holds in these cases.

If $G$ is not connected (Line \ref{bVC:split}) and is split into $G' = (V', E')$ and $G'' = (V'', E'')$ we 
note that $\Phi(G)\geq \Phi(G'), \Phi(G'')$ and $|V'|,|V''|\geq 1$; therefore,
\[R(G)=2 + R(G')+R(G'')\leq 2+2(|V'|-1)\cdot  2^{\Phi(G')} + 2(|V''| -1 ) \cdot 2 ^{\Phi(G'')} 
\leq 2(|V| -1 ) \cdot 2^{\Phi(G)}.\]

Finally, we need to handle the case in which $G$ is a $d$-regular graph (Line \ref{bVC:regular}). By the code structure, $d\in \{2,3,4\}$
and $G$ is connected.
In this case, two recursive calls are initiated, with $G_1=(V_1, E_1)$ and $G_2=(V_2, E_2)$ 
which are strict subgraphs of $G$.  
Since  $G$ is a connected $d$-regular graph, no vertex induced subgraph of $G$ is also $d$-regular, thus $\Phi_d(G_1)=\Phi_d(G_2)=0$
while $\Phi_d(G)=1$. Thus, $\Phi(G_1), \Phi(G_2)\leq \Phi(G)-1$. Since $|V_1|,|V_2|\geq 1 $ it follows that 
\begin{equation*}
\begin{aligned}
R(G)  &= 2+ R(G_1 )+R(G_2) \\
&= 2 + 2(|V_1|-1)\cdot  2^{\Phi(G_1)} + 2(|V_2| -1 ) \cdot 2 ^{\Phi(G_2)} \\ &\leq 2+
2(|V|-2)\cdot  2^{\Phi(G)-1} + 2(|V| -2 ) \cdot 2 ^{\Phi(G)-1} 
\leq  2(|V|-1)\cdot  2^{\Phi(G)}. 
\end{aligned}
\end{equation*}
\end{proof}

\begin{proof}[Proof of Lemma \ref{lem:bvc_prob}]
	To prove the lemma we show by induction a slightly stronger claim. 
	Given a collection of graphs $G_1, \ldots, G_t$,
	let $P(b,(G_1, \ldots, G_t))$ denote the probability that \\ %
	$\sum_{i=1}^t |\textsc{BetterVC}(G_i)| \leq b$. Now, we claim that 
	if the total size of minimal vertex covers of the graphs is 
	at most $k$ (formally, there are $S_1, \ldots, S_t$ where $S_i$ is a vertex cover of $G_i$ and $\sum_{i=1}^t |S_i| \leq k$)  then
	$P(b,(G_1, \ldots, G_t)) \geq p(b,k)$. 
	We prove the claim by induction over the lexicographical order of $(b,M,\ell)$, where $M$ 
	is the maximal number of vertices of a graph in $G_1, \ldots, G_t$, 
	and $\ell$ is the number of graphs of maximal size. 
	
	\noindent {\bf Base Case 1:} If $b<0$ then clearly $P(b, (G_1, \ldots, G_t)) = 0 = p(b,k)$.
	
	\noindent {\bf Base Case 2:} For any $b\in \mathbb{N}$, if $M\leq 1$, then
	clearly $P(b, (G_1, \ldots, G_t)) = 1 \geq p(b,k)$. 
	
	\noindent {\bf Induction Step:} Let $b\in \mathbb{N}$ and  $G_1, \ldots, G_t$ with $\ell$ graphs of maximal size $M$ and assume the claim holds for every $(b',M',\ell')$ lexicographically smaller than $(b,M,\ell)$. W.l.o.g assume that $G_1$ has $M$ vertices.
	We consider the execution of $\textsc{BetterVC}(G_1)$ and divide the analysis into cases depending on its execution path.  We use two simple properties along the proof. If
	$\textsc{BetterVC}(G_1)$ uses {\bf branch} over $U_1, \ldots, U_r$ with probabilities $\mu_1, \ldots, \mu_r$ then 
	\[P(b, (G_1,\ldots, G_t)) = \sum_{j=1}^r \mu_j P(b-|U_j|, (G_1\setminus U_j, G_2,\ldots, G_t)).\] 
	And if the algorithm selects $U$ into the cover then 
	\[P(b, (G_1, \ldots, G_t)) = P(b-|U| , (G_1\setminus U, G_2, \ldots, G_t)).\] 
	\noindent {\bf Case 1 (Trivial Case): } The empty set is a cover of $G_1$. Therefore 
	$|\textsc{BetterVC}(G_1)| = 0 $ and thus
	\[
	\Pr \left[ \sum_{i=1}^t |\textsc{BetterVC}(G_i)| \leq b \right] 
	=\Pr \left[ \sum_{i=2}^t |\textsc{BetterVC}(G_i)| \leq b \right]
	= P(b, (G_2, \ldots, G_t)) \geq p(b,k) ,
	\]
	where the last inequality follows from the induction hypothesis, as either the maximal graph size in $G_2, \ldots, G_t$ is smaller than $M$, or the number of graphs of maximal size is less than $\ell$. 
	
	\noindent {\bf Case 2 (Disconnected Graph):} $G_1$ is not connected, then let $G_1'$ and $G''_1$ be the two graphs considered in Line \ref{bVC:split}. Then,
	\begin{equation*}
	\begin{aligned}
	&\Pr \left[ \sum_{i=1}^t |\textsc{BetterVC}(G_i)| \leq b \right] \\
	=&\Pr \left[ |\textsc{BetterVC}(G'_1)| + |\textsc{BetterVC}(G_1'')|+ \sum_{i=2}^t |\textsc{BetterVC}(G_i)| \leq b \right]\\
	=& P(b, (G_1',G_1'',G_2, \ldots, G_t)) \geq p(b,k).
	\end{aligned}
	\end{equation*}
	Note that since the number of vertices in both $G_1'$ and $G_1''$ is 
	strictly smaller than $M$, the induction claim holds for $b$ and 
	$(G_1',G_1'',G_2, \ldots, G_t)$ from which the last inequality follows.
	
	\noindent {\bf Case 3 (Degree 1 Vertex):} The selection in Line \ref{bVC:deg1} is executed. Then, $G_1$ has a vertex $v$  of degree $1$, and $N(v) = \{u\}$, and $u$ is selected into the cover. Clearly, if $G_1$ has a vertex cover of size $k_1$ then $G_1\setminus \{u\}$ has a vertex cover of size $k_1-1$. Therefore, 
	\begin{equation*}
	P(b,(G_1,\ldots, G_t))= P(b-1, (G_1\setminus \{u\},G_2, \ldots, G_t)) \geq p(b-1,k-1)
	\geq p(b,k).
	\end{equation*}
	The first inequality holds by the induction hypothesis, and the second inequality follows from~\eqref{eq:bvc_rec}. 
	
	\noindent {\bf Case 4 (High Degree Vertex):}
	The algorithm uses the branching in Line \ref{bVC:deg5p}.
	Let $S_1$ be a minimal cover of $G_1$. 
	Denote $d^*=\min\{d, \Delta\}$.
	If  $v \in S_1$,
	then $S_1 \setminus \{v\}$ is a vertex cover of $G_1 \setminus \{v\}$. Also, if $d<\Delta$, then $S_1 \setminus \{v\}$ is also a vertex cover of $G_1 \setminus U$. Therefore, 
	\begin{equation*}
	\begin{aligned}
	&P(b, (G_1, \ldots, G_t))\\
	=& \gamma_{d^*}  \cdot P(b-1, (G_1 \setminus \{u\}, G_2 ,\ldots, G_t) )
	+ (1-\gamma_{d^*})  \cdot P(b-{d^*}, (G_1 \setminus U, G_2 ,\ldots, G_t) )\\
	\geq& \gamma_{d^*} \cdot p (b-1, k-1 ) + (1-\gamma_{d^*}) \cdot p \left(
	b-{d^*}, k- \begin{cases} 1 & \text{if~} d^* < \Delta \\ 0 &  \text{if~} d^*= \Delta\end{cases}
	\right) \geq p(b,k).
	\end{aligned}
	\end{equation*}
	The first inequality follows from the induction hypothesis, the second is due to \eqref{eq:bvc_rec}.
	
		Otherwise, if $v\notin S_1$, then $U\subseteq S_1$. Clearly, $S_1 \setminus U$ is a vertex cover of $G_1\setminus U$. Thus, we get
	\begin{equation*}
	\begin{aligned}
	&P(b, (G_1, \ldots, G_t))\\
	=& \gamma_{d^*}  \cdot P(b-1, (G_1 \setminus \{u\}, G_2 ,\ldots, G_t) )
	+ (1-\gamma_{d^*})  \cdot P(b-d^*, (G_1 \setminus U, G_2 ,\ldots, G_t) )\\
	\geq& \gamma_{d^*} \cdot p (b-1, k) + (1-\gamma_{d^*}) \cdot p \left(
	b-{d^*}, k- d^*
	\right) \geq p(b,k).
	\end{aligned}
	\end{equation*}
	As before, the first inequality is by the induction hypothesis, and the second is due to \eqref{eq:bvc_rec}.
	
	As the claim holds whether $v\in S_1$ or $v\notin S_1$ we get that the induction hypothesis holds for this case.

	\noindent {\bf Case 5 (Regular Graphs):}
	Line \ref{bVC:regular} takes place. Let $S_1$ be a minimal vertex cover 
	of $G_1$. As $S_1$ is a cover we have $v_1\in S_1$ or $v_2\in S_1$. W.l.o.g we may assume $v_1\in S_1$. Clearly, $S_1\setminus \{v_1\}$ 
	is a cover of $G_1 \setminus \{v_1\}$. Now, 
	
	\begin{equation*}
	\begin{aligned}
	&\Pr \left[ \sum_{i=1}^t |\textsc{BetterVC}(G_i)| \leq b \right] \\
	=& \Pr \left[ 1+ \min_{j=1,2 }|\textsc{BetterVC}(G_1 \setminus \{v_j\})| 
	+
	\sum_{i=2}^t |\textsc{BetterVC}(G_i)| \leq b 
	\right]  \\
	\geq& \Pr \left[ 1+|\textsc{BetterVC}(G_1 \setminus \{v_1\})|+  \sum_{i=2}^t |\textsc{BetterVC}(G_i)| \leq b \right] \\
	=& P(b-1 ,  (G_1\setminus \{v_1\}, G_2, \ldots, G_t)) \geq p(b-1,k-1 )  \geq p(b,k).
	\end{aligned}
	\end{equation*}
	The first inequality is since the event set  in the third term 
	is a subset of the event set of the second term. The second 
	inequality follows from the induction claim, and the last inequality 
	is due to \eqref{eq:bvc_rec}.

	\noindent {\bf Case 6 (Degree 2 Triangle):} The algorithm executes Line \ref{bVC:deg2case1}.
	Let $S_1$ be a minimal vertex cover of $G_1$. Note that $|S_1\cap \{x,y,v\}|\geq 2$
	and $S_1 \setminus \{x,y,v\}$ is a vertex cover of $G_1 \setminus \{x,y\}$. 
	Therefore,
	\begin{equation*}
	\begin{aligned}
	P(b, (G_1, \ldots, G_t))
	= P(b-2 ,  (G_1\setminus \{x,y\}, G_2, \ldots, G_t)) \geq p(b-2,k-2 ) \geq p(b,k).
	\end{aligned}
	\end{equation*}
	
	The first inequality follows from the induction hypothesis, and the second from \eqref{eq:bvc_rec}.

	\noindent {\bf Case 7 (Degree 2 Diamond):} Line \ref{bVC:deg2case2} is executed. 
	Let $S_1$ be a minimal vertex cover of $G_1$. Clearly, 
	$|S_1 \cap \{v,x,y,z\}| \geq 2$ and $S_1 \setminus \{v,x,y,z\}$ is 
	a vertex cover of $G_1 \setminus \{z,v\}$. Therefore, as in the previous case,
	\begin{equation*}
	\begin{aligned}
	P(b, (G_1, \ldots, G_t))
	= P(b-2 ,  (G_1\setminus \{z,v\}, G_2, \ldots, G_t)) \geq p(b-2,k-2 ) \geq p(b,k) .
	\end{aligned}
	\end{equation*}
	
	\noindent {\bf Case 8 (Degree 2 Branching):} Line \ref{bVC:deg2case3} is executed. Denote $G_1= (V_1,E_1)$. 
	Since the conditions are not met for {\bf Degree 2 Triangle} (Line \ref{bVC:deg2case1}) and {\bf Degree 2 Diamond} (Line \ref{bVC:deg2case2}), it holds that
	 $(x,y)\notin E_1$ and $|N(x)\cup N(y)| \geq 3$. As the graph does not have vertices of degree $5$ or more, we also have  $\deg(x),\deg(y) \leq 4$.
	We can  conclude that $3 \leq r \leq 7$ (recall that $r= |N(x)\cup N(y)|$).
	
	If there is a minimal vertex cover $S_1$ of $G_1$ such that $v\notin S_1$, then
	$x,y \in S_1$. Clearly, $S_1\setminus \{x,y\}$ is a vertex cover of $G_1\setminus \{x,y\}=G_1\setminus N(v)$.
	Also, it is easy to see that $S_1 \setminus \{x,y\}$ is also a vertex cover of $G_1\setminus (N(x)\cup N(y))$ (we 
	remove vertices which do not belong to the graph). Therefore,
	\begin{equation*}
	\begin{aligned}
	&P(b, (G_1, \ldots, G_t))\\
	=& \lambda_{1,r}  \cdot P(b-2, (G_1 \setminus N(v), G_2 ,\ldots, G_t) )
	+ (1-\lambda_{1,r})  \cdot P(b-r, (G_1 \setminus (N(x)\cup N(y)), G_2 ,\ldots, G_t) )\\
	\geq& \lambda_{1,r} \cdot p (b-2, k-2) + (1-\lambda_{1,r}) \cdot p \left(
	b-r, k- 2
	\right) \geq p(b,k).
	\end{aligned}
	\end{equation*}
	The first inequality follows from the induction hypothesis. The second  is due to \eqref{eq:bvc_rec}.
	
	Otherwise, every minimal vertex cover of $G_1$ includes $v$. Let $S_1$ be a minimal vertex cover of $G_1$. Clearly, $v\in S_1$. We note that $x,y \notin S_1$, since otherwise
	$S_1\setminus \{v\} \cup \{x,y\}$ is a vertex cover of $G_1$ of the same size as $S_1$, in contradiction to our case. Therefore, $N(x) \cup N(y)\subseteq S_1$. 
	Obviously, $S_1 \setminus (N(x) \cup N(y))$ is a vertex cover of $G_1 \setminus (N(x) \cup N(y))$. We also note that $S_1 \setminus \{v\}$  is a cover of $G_1 \setminus N(v)$.
	Therefore, 
	\begin{equation*}
	\begin{aligned}
	&P(b, (G_1, \ldots, G_t))\\
	=& \lambda_{1,r}  \cdot P(b-2, (G_1 \setminus N(v), G_2 ,\ldots, G_t) )
	+ (1-\lambda_{1,r})  \cdot P(b-r, (G_1 \setminus (N(x)\cup N(y)), G_2 ,\ldots, G_t) )\\
	\geq& \lambda_{1,r} \cdot p (b-2, k-1) + (1-\lambda_{1,r}) \cdot p \left(
	b-r, k- r
	\right) \geq p(b,k).
	\end{aligned}
	\end{equation*}
	The first inequality follows from the induction hypothesis. The second  is due to \eqref{eq:bvc_rec}.
	
	\noindent {\bf Case 9 (Degree $3$ Triangle):} Line \ref{bVC:deg3case1} is executed. Since this line of code has been reached, it holds that  $G_1$ has only vertices of  degree $3$ and $4$. Therefore $r=|N(z)| \in \{3,4\}$.
	
	If there is a minimal vertex cover $S_1$ of $G_1$ such that $v\notin S_1$, then $N(v)\subseteq S_1$, and $S_1\setminus N(v)$ is a vertex cover of $G\setminus N(v)$. Also, it is easy to see that $S_1\setminus \{z\}$ is a vertex cover of $G\setminus N(z)$ (after removing vertices which no longer belong to the graph).  Therefore,
	\begin{equation*}
	\begin{aligned}
	&P(b, (G_1, \ldots, G_t))\\
	=& \lambda_{2,r}  \cdot P(b-3, (G_1 \setminus N(v), G_2 ,\ldots, G_t) )
	+ (1-\lambda_{2,r})  \cdot P(b-r, (G_1 \setminus N(z), G_2 ,\ldots, G_t) )\\
	\geq& \lambda_{2,r} \cdot p (b-3, k-3) + (1-\lambda_{2,r}) \cdot p \left(
	b-r, k- 1
	\right) \geq p(b,k).
	\end{aligned}
	\end{equation*}
	The first inequality follows from the induction hypothesis. The second  is due to \eqref{eq:bvc_rec}.
	
	Otherwise, every minimal vertex cover $S_1$ of $G_1$ has $v$ in it. 
	Let $S_1$ be a minimal vertex cover of $G_1$. Clearly, $v\in S_1$. 
	If $|S_1\cap \{x,y,z\}| \geq 2$ then  $S_1\cup \{x,y,z\} \setminus \{v\}$ is a vertex cover of $G_1$ of the same size, contradicting our assumption. Therefore, $|S_1\cap \{x,y,z\}| \leq 1$. Since $x\in S_1$ or $y\in S_1$ (since $(x,y)$ is an edge of $G_1$) we have $z\notin S_1$, and $N(z) \subseteq S_1$. Also, note that $S_1\setminus \{v\}$ is a cover of $G\setminus N(v)$ and $|S_1\setminus \{v\}| \leq |S_1| -1$. 
	Therefore,  
	\begin{equation*}
	\begin{aligned}
	&P(b, (G_1, \ldots, G_t))\\
	=& \lambda_{2,r}  \cdot P(b-3, (G_1 \setminus N(v), G_2 ,\ldots, G_t) )
	+ (1-\lambda_{2,r})  \cdot P(b-r, (G_1 \setminus N(z), G_2 ,\ldots, G_t) )\\
	\geq& \lambda_{2,r} \cdot p (b-3, k-1) + (1-\lambda_{2,r}) \cdot p \left(
	b-r, k- r
	\right) \geq p(b,k).
	\end{aligned}
	\end{equation*}
	The first inequality follows from the induction hypothesis. The second one is due to \eqref{eq:bvc_rec}.

	\noindent {\bf Case 10 (Degree 3 Diamond):} Line \ref{bVC:deg3case2} is executed. 
	
	If there is a minimal vertex cover $S_1$ such that $v\notin S_1$, then,
	\begin{equation*}
	\begin{aligned}
	&P(b, (G_1, \ldots, G_t))\\
	=& \lambda_{3}  \cdot P(b-3, (G_1 \setminus N(v), G_2 ,\ldots, G_t) )
	+ (1-\lambda_{3})  \cdot P(b-2, (G_1 \setminus \{v,w\}, G_2 ,\ldots, G_t) )\\
	\geq& \lambda_{3} \cdot p (b-3, k-3) + (1-\lambda_{3}) \cdot p (
	b-2, k) \geq p(b,k).
	\end{aligned}
	\end{equation*}
	The first inequality follows from the induction hypothesis. The second  is due to \eqref{eq:bvc_rec}.
	
	Otherwise, every minimal vertex cover $S_1$ has $v$ in it. Let $S_1$ be a
	minimal vertex cover of $G_1$. If $|S_1 \cap \{x,y,z\}|\geq 2$ we get 
	get a contradiction to the assumption by removing $v$ from $S_1$ 
	and adding a vertex from $x,y,z$ into it. Therefore $|S_1 \cap \{x,y,z\}|\leq 1$
	and surely $w \in S_1$ (if $w \notin S_1$ then $x,y\in S_1$). 
	We also note that $S_1\setminus \{v\}$ is a vertex 
	cover of $G\setminus N(v)$.
	Therefore,
	\begin{equation*}
	\begin{aligned}
	&P(b, (G_1, \ldots, G_t))\\
	=& \lambda_{3}  \cdot P(b-3, (G_1 \setminus N(v), G_2 ,\ldots, G_t) )
	+ (1-\lambda_{3})  \cdot P(b-2, (G_1 \setminus \{v,w\}, G_2 ,\ldots, G_t) )\\
	\geq& \lambda_{3} \cdot p (b-3, k-1) + (1-\lambda_{3}) \cdot p (
	b-2, k-2) \geq p(b,k).
	\end{aligned}
	\end{equation*}
	The first inequality follows from the induction hypothesis. The second  is due to \eqref{eq:bvc_rec}.

	\noindent {\bf Case 11 (Degree 4 Branching):}  Line \ref{bVC:deg3case3} is executed. Since there are no edges between $x,y,z$ and the vertices
	have no common neighbor beside $v$, we have $r\in \{5,6,7\}$. 
	We  further distinguish between the following  sub-cases.
	
	\begin{enumerate}
		\item \label{bVC:cases11_xnotin}
		If there is a minimal vertex cover $S_1$ of $G_1$ such that $v\notin S_1$, then $N(v) \in S_1$. Clearly, $S_1 \setminus N(v)$ is a vertex cover of $G_1 \setminus N(v)$. Also, $S_1 \setminus \{x\}$ is a vertex cover of $G_1 \setminus N(x)$, 
		and $S_1\setminus N(v)$ is a vertex cover of $G\setminus (\{x\}\cup N(y)\cup N(z))$. Therefore,
		\begin{equation*}
		\begin{aligned}
		&P(b, (G_1, \ldots, G_t))\\
		=& \delta_{r,1}  \cdot P(b-3, (G_1 \setminus N(v), G_2 ,\ldots, G_t) ) + \\
		&\delta_{r,2}  \cdot P(b-4, (G_1 \setminus N(x), G_2 ,\ldots, G_t) ) + \\
		&\delta_{r,3}  \cdot P(b-r-1, (G_1 \setminus (\{x\} \cup N(y)\cup N(z)), G_2 ,\ldots, G_t) )
		\\ 
		\geq &\delta_{r,1}\cdot p (b-3, k-3)  + \delta_{r,2} \cdot p(b-4, k-1 ) + \delta_{r,3} \cdot p(b-r-1, k-3) \geq p(b,k).
		\end{aligned}
		\end{equation*}
		Thus, we may assume  that $v$ is in every minimal cover.
		\item
	There is a minimal cover $S_1$ of $G_1$ such that $x,y,z\notin S_1$. Then $N(x),N(y), N(z) \subseteq S_1$. 
		Now, $S_1\setminus N(x)$ is a vertex cover of $G_1 \setminus N(x)$, $S_1\setminus \{v\}$ is a vertex cover of $G_1 \setminus N(v)$ and $S_1\setminus (N(y)\cup N(z))$ is a vertex cover 
		of $G_1 \setminus (\{x\} \cup N(y) \cup N(z)) $. Therefore,
		\begin{equation*}
		\begin{aligned}
		&P(b, (G_1, \ldots, G_t))\\
		=& \delta_{r,1}  \cdot P(b-3, (G_1 \setminus N(v), G_2 ,\ldots, G_t) ) + \\
		&\delta_{r,2}  \cdot P(b-4, (G_1 \setminus N(x), G_2 ,\ldots, G_t) ) + \\
		&\delta_{r,3}  \cdot P(b-r-1, (G_1 \setminus (\{x\} \cup N(y)\cup N(z)), G_2 ,\ldots, G_t) )
		\\ 
		\geq &\delta_{r,1}\cdot p (b-3, k-1)  + \delta_{r,2} \cdot p(b-4, k-4 ) + \delta_{r,3} \cdot p(b-r-1, k-r) \geq p(b,k).
		\end{aligned}
		\end{equation*}
		
		\item 
		There is a minimal cover $S_1$ of $G_1$ such that $x\notin S_1$, but one of $y,z $ is in $S_1$, w.l.o.g $y\in S_1$. 
		Therefore $N(x), N(z) \subseteq S_1$, and we have that
		$S_1\setminus N(x)$ is a vertex cover of $G_1 \setminus N(x)$, 
		$S_1\setminus \{v,y\}$ is a vertex cover of $G_1\setminus N(v)$ 
		and $S_1 \setminus \left(N(z)\cup \{y\}\right)$ is a vertex cover
		of $G_1 \setminus (\{x\} \cup N(y) \cup N(z))$. Note that $N(z)\geq \ceil{\frac{r}{2}}$.  Therefore,
		\begin{equation*}
		\begin{aligned}
		&P(b, (G_1, \ldots, G_t))\\
		=& \delta_{r,1}  \cdot P(b-3, (G_1 \setminus N(v), G_2 ,\ldots, G_t) ) + \\
		&\delta_{r,2}  \cdot P(b-4, (G_1 \setminus N(x), G_2 ,\ldots, G_t) ) + \\
		&\delta_{r,3}  \cdot P(b-r-1, (G_1 \setminus (\{x\} \cup N(y)\cup N(z)), G_2 ,\ldots, G_t) )
		\\ 
		\geq &\delta_{r,1}\cdot p (b-3, k-2)  + \delta_{r,2} \cdot p(b-4, k-4 ) + \delta_{r,3} \cdot p\left(b-r-1, k-1-\ceil{\frac{r}{2}}\right) \geq p(b,k).
		\end{aligned}
		\end{equation*}
		\item There is a minimal cover $S_1$ such that $x\notin S_1$
		and $y,z \in S_1$. Then $S_1 \cup \{x\}\setminus \{v\}$ is a minimal cover without $v$, and therefore the claim holds due to sub-case \ref{bVC:cases11_xnotin}. 
		\item
		There is a minimal vertex cover $S_1$ such that $x,v\in S_1$.
		If $y\in S_1$ or $z\in S_1$, w.l.o.g $y\in S_1$, then
		$S_1 \cup \{z\} \setminus \{v\}$ is a minimal vertex 
		cover of $G_1$ which does not include $v$. As this situation is 
		already handled in sub-case \ref{bVC:cases11_xnotin}, we 
		may assume $y,z \notin S_1$ and therefore $N(y), N(z) \subseteq S_1$. Now, note that $S_1\setminus \{x,v\}$ is a vertex cover of both $G_1\setminus N(v)$ and $G_1 \setminus N(x)$. Therefore,
		\begin{equation*}
		\begin{aligned}
		&P(b, (G_1, \ldots, G_t))\\
		=& \delta_{r,1}  \cdot P(b-3, (G_1 \setminus N(v), G_2 ,\ldots, G_t) ) + \\
		&\delta_{r,2}  \cdot P(b-4, (G_1 \setminus N(x), G_2 ,\ldots, G_t) ) + \\
		&\delta_{r,3}  \cdot P(b-r-1, (G_1 \setminus (\{x\} \cup N(y)\cup N(z)), G_2 ,\ldots, G_t) )
		\\ 
		\geq &\delta_{r,1}\cdot p (b-3, k-2)  + \delta_{r,2} \cdot p(b-4, k-2 ) + \delta_{r,3} \cdot p(b-r-1, k-r-1) \geq p(b,k).
		\end{aligned}
		\end{equation*}
	\end{enumerate}
\end{proof}

\section{Numerical Methods}
\label{sec:numerical}
\newcommand{\ruleopt}{\textnormal{\texttt{r\_opt}}}
While our main contributions are theoretical, optimizing the parameter values and evaluating the running times of our algorithms required some numerical analysis.
In this section we overview the methods and tools used for obtaining the numerical results. 
We include a Python implementation of these methods as part of the supplementary material. 

We use  Algorithm~\ref{alg:VC3star} as our running example. 
In each recursive call the algorithm finds a vertex $v$ of degree at least $3$. 
The algorithm then
either selects a vertex $v$ to the solution with probability $\gamma_{d}$, or up to $\Delta$ of the neighbors of $v$ with probability $1-\gamma_{d}$, where $d = \min\{\deg(v), \Delta\}$ and $\gamma_3,\ldots, \gamma_{\Delta}$ as well as $\Delta$ are configuration parameters of the algorithm.  As explained in Section~\ref{sec:VC}, given a graph which has a vertex cover of size $k$, the algorithm finds vertex cover of size at most $\alpha \cdot k$ with probability at least $p(\floor{\alpha \cdot k},k)$, where $p$ is the composite recurrence  defined in \eqref{eq:vcstar_rec}.
Technically, $p$ is the composite recurrence of $\terms$ as defined in \eqref{eq:terms_vcstar}. Observe that the set $\terms$ depends on the configuration parameters $\gamma_3,\ldots,\gamma_{\Delta}$. 
 This leads to a parameterized $\alpha$-approximation for vertex cover  in time $\approx \frac{1}{p(\floor{\alpha,k},k)}$. 

Theorem~\ref{thm:rec} shows that $\frac{1}{p(\floor{\alpha k} k)}\approx \left( \exp(M)\right)^k$, where $M$ is the maximal $\alpha$-branching  number of a term in $\terms$. For a fixed $\alpha$ (say $\alpha =1.5$), this leads to two related numerical problems. First, given values for  $\gamma_3,\ldots, \gamma_{\Delta}$, say $\gamma_d= \frac{1}{2}$ for all $3\leq d \leq \Delta$, what is the $\alpha$-branching number of each of the terms in $\terms$. And second, how do we find values for $\gamma_3,\ldots, \gamma_{\Delta}$  for which the maximum $\alpha$-branching number of a term in $\terms$ is minimized. That is, how do we find $\gamma_3,\ldots, \gamma_{\Delta}$ for which the overall running time is minimized. We describe numerical tools which solve the two problems simultaneously. 

While the task of finding the $\alpha$-branching number of a given term is well defined (Definition~\ref{def:alpha_branching}), the second task requires us to provide an abstract viewpoint regarding the structure of our algorithms, and the configuration parameters we aim to optimize. In this abstract viewpoint, 
each of our algorithms consists of $R\in \mathbb{N}_+$ branching rules. For example, Algorithm~\ref{alg:VC3star} involves $\Delta-2$ branching rules: $\Delta-3$ rules for vertices of degree $3\leq d<\Delta$, in which either $v$ or all of its neighbors are selected for the solution, and a single rule for vertices of degree $\Delta$ or more, in which either $v$ or $\Delta$ of its neighbors are selected for the solution.

In our abstract viewpoint, the  $\ell$-th rule, $1\leq \ell \leq R$, has $r_{\ell}$ branching {\em options} and $h_{\ell}$ branching {\em states}. The branching {\em options} reflect the potential actions the algorithm may take, and determines the length of the terms associated with the rule. For example, in Algorithm~\ref{alg:VC3star} each rule has $r_{\ell}=2$ branching options: either  select $v$ or select all (some) of its neighbors for the solution.  
The  branching {\em states} reflect the various cases used in the analysis of the algorithm, and define the number of terms added to the composite recurrence due to the rule. 
In Algorithm~\ref{alg:VC3star}, the number of states for each rule is $h_{\ell}=2$: either $v$ is in some  minimal vertex cover, or not. It is often the case that  the number of branching options and states of each rule is  the same, but this is not always the case. For example, the Degree $4$ Branching rule of Algorithm~\ref{alg:betterVC} (Line~\ref{bVC:deg3case3}) involves $3$ branching options (select $N(v)$, $N(x)$ or $\{x\}\cup N(y)\cup N(z)$ to the solution) while there are $4$ different branching states associated with the same rule in the analysis. 

For each rule $1\leq \ell \leq R$, the algorithm uses a distribution $\bgam^\ell\in \mathbb{R}^{r_{\ell}}$ to randomly select one of the  branching options.  In the case of Algorithm~\ref{alg:VC3star} the distributions are $(\gamma_d, 1-\gamma_d)$ for $3\leq d\leq \Delta$. 
The vector $\bb^{\ell}\in \mathbb{N}_+^{r_{\ell}}$ is the budget decrease incurred by selecting each option. 
The $i$-th entry of $\bb^{\ell}$ is essentially the number of vertices added to the solution in the $i$-th options. In Algorithm~\ref{alg:VC3star} the budget vectors are $(1,d)$ for $3\leq d\leq \Delta$ as either a set of size $1$ ($\{v\}$) or of size $d$ (the neighbors of $v$) is added to the solution.

Each rule  $1\leq \ell \leq R$ is also associated with $h_{\ell}$ vectors $\bk^{\ell,1}, \ldots, \bk^{\ell, h_{\ell}} \in \mathbb{N}^{r_{\ell}}$, where the value $\bk^{\ell, j}_i$ is the decrease in the parameter (coverage) when selecting the $i$-th option of rule $\ell$  while in state $j$. For example, the rule for degree $3$ vertices of Algorithm~\ref{alg:VC3star} has two states: there is a minimum vertex cover with $v$, and there is no minimum vertex cover with $v$. In the former case, the parameter decreases by $1$ if $v$ is selected, and also by $1$ if the neighbors of $v$ are selected. Therefore, the vector associated with this particular state and branching rule is $(1,1)$.

 Using the above notation, we can lower bound the success probability of the algorithm by writing the composite recurrence $p_{\bgam^1, \ldots,\bgam^R}(b,k)$, given by
\begin{equation}
\label{eq:rules_rec}
p_{\bgam^1, \ldots, \bgam^R}(b,k) =
\min_{1\leq \ell \leq R} \min_{~~1\leq j \leq h_\ell ~~} \sum_{i=1}^{r_{\ell}} \bgam^{\ell}_i \cdot p_{\bgam^1, \ldots, \bgam^R} (b-\bb^{\ell}_i , k -\bk^{\ell,j}_i ),
\end{equation}
with the same initial conditions as in \eqref{eq:rec_relation}. 
For example, we can write in (\ref{eq:rules_rec}) $R=\Delta -2$, $h_{\ell} = r_{\ell} = 2$, $\bb^{\ell} = (1, \ell +2)$ for all $1\leq \ell \leq \Delta$, $\bk^{\ell,1} = (1,1)$  and $\bk^{\ell,2} =(0,\ell+2)$ for all $1\leq \ell <\Delta-2$; also,  
$\bk^{\Delta-2,1} = (1,0)$  and $\bk^{\Delta-2,2} =(0,\Delta)$. Then, we have that \eqref{eq:rules_rec} is the same as  \eqref{eq:vcstar_rec}.

For each of our algorithms and a given approximation ratio $\alpha$, to obtain an optimal running time we seek distributions $\bgam^1, \ldots, \bgam^{R}$ that maximize $\lim_{k\rightarrow \infty} \frac{1}{k}\log  
p_{\bgam^1, \ldots, \bgam^R}(\floor{\alpha k} ,k)$. 
By Theorem~\ref{thm:rec},
\begin{align}
\label{eq:global_optimization}
&\max \left\{ 
	\lim_{k\rightarrow \infty} \frac{1}{k}\log  
	p_{\bgam^1, \ldots, \bgam^R}(\floor{\alpha k} ,k)
	\middle|
	\begin{array}{c}
	\forall 1\leq \ell \leq R: \bgam^{\ell}\in \mathbb{R}^{r_{\ell}}\\
	\bgam^1, \ldots, \bgam^R \text{ are distributions}
	\end{array}
\right\}\\
\nonumber
=&
\max \left\{ 
-\max_{1\leq \ell \leq R} \max_{~1\leq j \leq h_{\ell}~} M^{\ell,j}
\middle|
\begin{array}{c}
\forall 1\leq \ell \leq R: \bgam^{\ell}\in \mathbb{R}^{r_{\ell}}\\
\bgam^1, \ldots, \bgam^R \text{ are distributions}\\
M^{\ell,j} \text{ is the $\alpha$-branching number of $(\bb^{\ell}, \bk^{\ell,j}, \bgam^{\ell})$
}
\end{array}
\right\}\\
=&
-
\max_{1\leq \ell \leq R}
\min
 \left\{
 \max_{~1\leq j \leq h_{\ell}~} \frac{ \D{\bq^{j}}{\bgam}} {\bk^{\ell,j} \cdot \bq^j}
\middle|
\begin{array}{c}
\bgam, \bq^1, \ldots, \bq^{h_{\ell}}\in \mathbb{R}^{r_{\ell}} \text{ and are all distributions}\\
\forall 1\leq j \leq h_{\ell}:
\alpha \bq^j \cdot \bk^{\ell,j} \geq \bq^j \cdot \bb^{\ell}
\end{array}
 \right\}
 \label{eq:rule_optimization}
\end{align}

Define the {\em rule opimization problem} as follows. The input is $\alpha\in \mathbb{R}_+$,  $r,h \in \mathbb{N}$, $\bb\in \mathbb{N}_+^{r}$ and $h$ vectors $\bk^1, \ldots, \bk^{h}\in \mathbb{N}^{r}$. The objective is to find distributions 
$\bgam, \bq^1, \ldots, \bq^h \in \nonneg^{r}$ such that, for any $1\leq j \leq h$, it holds that $\alpha \bq^j \cdot  \bk^j \geq \bq^j\cdot \bb$, and 
$\max_{1\leq j \leq h} \frac{\D{\bq^j}{\bgam}}{\bq^j \cdot \bk^j}$ is minimized. That is, the rule optimization problem is 
\begin{equation}
	\label{eq:rule_opt}
\ruleopt(\alpha, r,h,\bb,\bk^1,\ldots, \bk^h) =
\min \left\{
\max_{~1\leq j \leq h~} \frac{ \D{\bq^{j}}{\bgam}} {\bk^{j} \cdot \bq^j}
\middle|
\begin{array}{c}
	\bgam, \bq^1, \ldots, \bq^{h}\in \mathbb{R}^{r} \text{ and are all distributions}\\
	\forall 1\leq j \leq h:
	\alpha \bq^j \cdot \bk^{j} \geq \bq^j \cdot \bb
\end{array}
\right\}
\end{equation}
We can rewrite the problem in \eqref{eq:rule_optimization} as $R$ separate rule optimization problems
\begin{align*}
	&\max \left\{ 
	\lim_{k\rightarrow \infty} \frac{1}{k}\log  
	p_{\bgam^1, \ldots, \bgam^R}(\floor{\alpha k} ,k)
	\middle|
	\begin{array}{c}
		\forall 1\leq \ell \leq R: \bgam^{\ell}\in \mathbb{R}^{r_{\ell}}\\
		\bgam^1, \ldots, \bgam^R \text{ are distributions}
	\end{array}
	\right\}\\
	=& -\max_{1\leq \ell \leq R} \ruleopt(\alpha, r_{\ell}, h_{\ell}, \bb^{\ell}, \bk^{\ell,1},\ldots, \bk^{\ell, h_{\ell} }).
\end{align*}
That is, to evaluate the running of the algorithm it suffices to solve the rule optimization problem separately for each of the rules used by the algorithm. 

In the following we show how the rule optimization problem can be numerically solved. 
We first
show that the rule optimization problem is quasiconvex and discuss the methods used to solve the problems as such.
We then consider a common special case  which has a nearly closed form solution.

\subsection{Quasiconvex Programming} 
\label{sec:quasi}

A set $D\subseteq \mathbb{R}^d$ is {\em convex} if  the line connecting any two points in $D$ is also in $D$. That is, for every $\bar{x}, \bar{y}\in D$ and $\xi \in (0,1)$ it holds that $\xi \cdot \bar{x} + (1-\xi )\cdot \bar{y}\in D$ . 
A function $f:C \rightarrow \mathbb{R}$ is {\em quasiconvex} if $C$ is convex and, for any $\beta \in \mathbb{R}$, the {\em level-set} $\{x\in C~ | f(x)\leq \beta\}$ is convex. A {\em quasiconvex program} is the problem of finding the minimum of a quasiconvex  function $f$ over a convex set $D$ (that is, $\min_{\bar{x}\in D} f(\bar{x})$). 
We remind the reader that a function $f:C \rightarrow \mathbb{R}$ is {\em convex} if $C$ is a convex set, and for every $\bar{x},\bar{y}\in C$ and $\xi\in [0,1]$, it holds that 
$$
f\left( \xi \cdot \bar{x} +(1-\xi )\cdot \bar{y}\right) \leq \xi \cdot f(\bar{x}) +(1-\xi)\cdot f(\bar{y}).
$$
It is easy to show that every convex function is quasiconvex. In contrast, there  are quasiconvex functions, such as $x^3 $, which are not convex. 
 Quasiconvex programming was first defined by Amenta et. al. \cite{ABE99}, and was already used in the context of  multivariate recurrences in \cite{Epp06}.

We use two well-known constructions of quasiconvex functions. The next lemma is a special case of Theorem 1 in~\cite{AB19}.
\begin{lemma}
	\label{lem:quasi_div}
	Let $f:C\rightarrow \mathbb{R}$ be a convex function  where $C\subseteq \mathbb{R}^d$,  and let  $\bar{c}\in \mathbb{R}^d$ be an arbitrary vector. Then 
	$\frac{f(\bar{x})}{\bar{c}\cdot \bar{x}}$ is a quasiconvex function. 
\end{lemma}
We also use the following construction.
\begin{lemma}
	\label{lem:quasi_max}
	Let $f_1,\ldots, f_k:C\rightarrow \mathbb{R}$ be $k$ quasiconvex functions. Then $f(\bar{x})=\max\{f_1(\bar{x}),\ldots,f_k(\bar{x})\}$ is quasiconvex.
\end{lemma}
\begin{proof}
	Let $ \beta\in \mathbb{R}$. The level set of  $f$ corresponding to $\beta$ is 
	$$\{ \bar{x} \in C ~|~f(\bar{x}) \leq \beta \}
	 = \bigcap_{1\leq i \leq k} \{\bar{x}~|~f_i(\bar{x}) \leq \beta\}.
	$$
	Each of the sets $\{\bar{x}~|~f_i(\bar{x}) \leq \beta\}$ is convex as $f_i$ is quasiconvex, therefore their intersection is convex as well. That is, $f$ is quasiconvex.
\end{proof}

We now show that the rule optimization problem is a a  quasicovex program. Fix an instance $\alpha, r, h, \bb, \bk^1,\ldots ,\bk^h$ of the problem. It is well known that Kullback-Leibler divergence is convex (Theorem 2.7.2 cf. \cite{Co06}); therefore, by Lemma~\ref{lem:quasi_div},
the functions  $f_j(\bgam, \bq^1 ,\ldots, \bq^h)=\frac{\D{\bq^j}{\gamma}}{\bk^j \cdot \bq^j}$, $\forall 1\leq j \leq h$, 
are  quasiconvex. Thus, by Lemma~\ref{lem:quasi_max}, the function  $f(\bgam, \bq^1 ,\ldots, \bq^h)= \max_{1\leq j \leq h} f_j(\bgam, \bq^1 ,\ldots, \bq^h)$ is also quasiconvex. Furthermore, the constraints over $\bgam, \bq^1 ,\ldots, \bq^h$ defining the rule optimization problem are all linear; thus, the feasible region is convex.

We used the disciplined quasiconvex programming module  of \textsc{cvxpy} \cite{AB19}, an open source python optimization package, to solve the rule  optimization problems which did not fall into the category of simple rules (see Section~\ref{sec:simple_rules}). Specifically, the results for $3$-Hitting Set (Section \ref{sec:HS}) were evaluated using this method. We encountered numerical accuracy issues when using \textsc{cvxpy}. In such cases, the returned solution was modified to make it a feasible solution. While such changes may harm the optimality of the solution, they can only increase the running times of our algorithms. 

\subsection{Simple Branching Rules}
\label{sec:simple_rules}

Many of the branching rules presented in this paper involve two branching options and two branching states; that is, $h=r=2$.  These include, e.g., all the branching rules of Algorithm~\ref{alg:VC3star}. We refer to such rules as {\em simple}. For such rules, the rule optimization problem has a nearly closed-form solution. 

Let $\bb\in \mathbb{N}_{+}^2$ and $\bk^1,\bk^2 \in \mathbb{N}_{\geq 0}^2$, where $\bk^1$ and $\bk^2$ are not all zeros.
We consider the rule optimization problem of $r=h=2$, $\bb$, the vectors $\bk^1$ and $\bk^2$ and an arbitrary $\alpha> \max_{j\in \{1,2\}} \min_{i\in \{1,2\}:~\bk^j_i \neq 0} \frac{\bb_i}{\bk^j_i}$ (this ensures that we do not consider approximation ratios below the critical ratio, see Definition~\ref{def:critical_ratio}). By~\eqref{eq:rule_opt}, we have
\begin{equation}
	\label{eq:simple_ropt}
	\begin{aligned}
	\ruleopt&(\alpha, 2,2,\bb,\bk^1,\bk^2) =
	\min \left\{
	\max_{~1\leq j \leq 2~} \frac{ \D{\bq^{j}}{\bgam}} {\bk^{j} \cdot \bq^j}
	\middle|
	\begin{array}{c}
		\bgam, \bq^1, \bq^{2}\in \mathbb{R}^{2} \text{ and are all distributions}\\
		\forall 1\leq j \leq 2:
		\alpha \bq^j \cdot \bk^{j} \geq \bq^j \cdot \bb
	\end{array}
	\right\}\\
	=& \min_{~\bgam\in \mathbb{R}^2 \textnormal{~is a distribution}
~	}
	\max_{j\in \{1,2\}} \min\left\{ \frac{ \D{\bq}{\bgam}} {\bk^{j} \cdot \bq}
	\,		\middle|\, 
		\begin{array}{c}
		\bq \in \mathbb{R}^{2} \text{~is a distributions}\\
		\alpha \bq \cdot \bk^{j} \geq \bq \cdot \bb
	\end{array}
	\right\}\\
	=& \min_{~\gamma \in [0,1]
		~	}
	\max_{j\in \{1,2\}} \min\left\{ \frac{ \D{(q,1-q)}{(\gamma,1-\gamma)}} {\bk^{j} \cdot (q,1-q)}
	\,		\middle|\, 
	\begin{array}{c}
		q\in [0,1]\\
		\alpha \cdot (q,1-q) \cdot \bk^{j} \geq (q,1-q) \cdot \bb
	\end{array}
	\right\}.
	\end{aligned}
\end{equation}
For every $j\in \{1,2\}$, define 
\begin{equation}
	\label{eq:Cj_def} C_j = \{ q\in [0,1]\,|\,  \alpha \cdot (q,1-q) \cdot \bk^j \geq (q,1-q) \cdot \bb\}
	\end{equation}
to be the feasible region of the optimization problem in the last term of \eqref{eq:simple_ropt}. We use the shorthand $\D{a}{b}  = \D{(a,1-a)}{(b,1-b)}$ for two scalars $a,b\in [0,1]$ and define
\begin{equation}
	\label{eq:fj_def}
f_j(\gamma) = \min\left\{ \frac{\D{q}{\gamma}}{\bk^j \cdot (q,1-q)}\,\middle|\, q\in C_j \right\}.
\end{equation}
Using the above notation in \eqref{eq:simple_ropt}, we obtain
\begin{equation}
	\label{eq:simple_ropt_by_f}
\ruleopt(\alpha, 2,2,\bb,\bk^1,\bk^2) = \min_{\gamma\in [0,1]} \max \{  f_1(\gamma),f_2(\gamma)\}. 
\end{equation}
Next we show that $f_j$ is monotone and has a closed form.  This means that either \eqref{eq:simple_ropt_by_f} is trivial, or can be solved using a simple binary search. We first show that the set $C_j$ has a simple structure.
\begin{lemma}
For every $j\in [0,1]$ there is $c_j\in [0,1]$ such that $C_j = [c_j,1]$ or $C_j=[0,c_j]$. 
\end{lemma}
\begin{proof}
The set $C_j \subseteq \mathbb{R}$ is defined by three linear inequalities (two of those are $q\geq 0$ and $q\leq 1$, see \eqref{eq:Cj_def}), and is therefore convex.  A convex set in $\mathbb{R}$  is an interval, therefore $C_j=[a,b]$ for some $0\leq a$ and $b\leq 1$. 

Let $i^*\in \{1,2\}$ such that $i^* =  \argmin_{i\in \{1,2\}:~\bk^j_i \neq 0} \frac{\bb_i}{\bk^j_i}$. By our requirements, $\alpha >\frac{\bb_{i^*}}{\bk^j_{i^*}}$; therefore, $\alpha \cdot \bk^j_{i^*} > \bb_{i^*}$. Consider the following two cases:
\begin{itemize}
	\item $i^*=1$.  Then if we select $q=1$, we have that 
	$$\alpha \cdot (q,1-q)\cdot \bk^j = \alpha \cdot (1,0)\cdot \bk^j = \alpha \cdot \bk^j_{i^*} >  \bb_{i^*} =   (1,0)\cdot \bb =   (q,1-q) \cdot \bb,$$
	thus $1\in C_j$. 
	\item $i^*=2$. Then if we select $q=0$ it holds that 
	that 
	$$\alpha \cdot (q,1-q)\cdot \bk^j = \alpha \cdot (0,1)\cdot \bk^j = \alpha \cdot \bk^j_{i^*} >  \bb_{i^*} =   (0,1)\cdot \bb =   (q,1-q) \cdot \bb,$$
	thus $0\in C_j$.
	\end{itemize}
	By the above, either $0\in C_j$ or $1\in C_j$, and therefore $a=0$ or $b=1$. 
\end{proof}
Since $C_j\subseteq \mathbb{R}$ is defined by three inequalities , the interval itself, as well as the value of $c_j$, can be easily found using linear programming.  It is also fairly straightforward to  find a closed formula for $c_j$. 
The next lemma presents the closed formula for $f_j$ and uses it to show monotonicity. 
\begin{lemma}
	For  every $j\in \{1,2\}$ and $\gamma\in[0,1]$, it holds that
	$$f_j(\gamma) = \begin{cases} \frac{ \D{c_j}{\gamma}}{ \bk^j\cdot (c_j,1-c_j)}&\gamma\notin C_j\\
	0& \gamma\in C_j \end{cases} .
	$$
	Furthermore, if $C_j=[c_j,1]$ then $f_j$ is weakly  decreasing, and if $C_j=[0,c_j]$ then $f_j$ is weakly increasing. 
\end{lemma}
\begin{proof}
	Fix $j\in \{1,2\}$.
Define $g(q,\gamma) = \frac{\D{q}{\gamma}}{\bk^j\cdot (q,1-q)}$. Then $f_j(\gamma) = \min_{q\in C_j} g(q,\gamma)$ for all $\gamma\in[0,1]$. Since $\D{}{}$ is non-negative, it follows that $g$ is non-negative, and therefore $f$ is non-negative.  Furthermore, $g(\gamma,\gamma )= \frac{\D{\gamma}{\gamma}}{\bk^j\cdot(\gamma,1-\gamma)}=0$ for every $\gamma$. 

Let $\gamma\in[0,1]$. 
If $\gamma\in C_j$ then $0\leq f_j(\gamma) \leq g(\gamma,\gamma)=0$, and therefore $f_j(\gamma) =0$. 
We now need to consider the case where $\gamma \notin C_j$. 
Since $\D{\cdot}{\gamma}$ 
is convex, it follows from Lemma~\ref{lem:quasi_div} that $g(x,\gamma)$ is quasiconvex as a function of $x$ and has a minimum at $x=\gamma$.
 Therefore, the function is decreasing in $[0,\gamma]$ and increasing in $[\gamma,1]$. 
 As $\gamma\notin C_j$, it holds that $C_j\subseteq [0,\gamma]$ or $C_j\subseteq [\gamma,1]$. In both cases this implies that 
$$
f_j(\gamma) = \min_{q\in C_j} g(q,\gamma) = g(c_j,\gamma) = \frac{\D{c_j}{\gamma}}{\bk^j\cdot(c_j,1-c_j)}. 
$$

It remains to show that $f_j$ is monotone. Since $\D{\cdot}{\cdot}$ is convex, it follows that $g(c_j,x)=\frac{\D{c_j}{x}}{\bk^j\cdot (c_j,1-c_j)}$ is a convex function of $x$. As $g(c_j,c_j)=0$ and $g$ is non-negative this implies that $g(c_j,x)$ has a minimum as $x=c_j$. It therefore holds that $g(c_j,x)$ is decreasing in $[0,c_j]$ and increasing in $[c_j,1]$. Consider the two possible cases:
\begin{itemize}
	\item In case $C_j=[c_j,1]$, we have that $f(\gamma) = g(c_j,\gamma)$ in the interval $[0,c_j)$ and is therefore decreasing in the interval $[0,c_j)$. Furthermore, $g(\gamma)=0$  for every $\gamma\geq c_j$. Thus, the function is weakly decreasing in the interval $[0,1]$. 
	\item In case $C_j=[0,c_j]$, for $\gamma>c_j$ it holds that $f(\gamma) = g(c_j,\gamma)$, and therefore the function is increasing in the interval $(c_j,1]$. Additionally, $f(\gamma) = 0$ for $\gamma \in [0,c_j]$. Thus, the function is weakly increasing in the interval $[0,1]$. 
	\end{itemize}
\end{proof} 

In order to solve \eqref{eq:simple_ropt} we need to consider two cases.
\begin{itemize}
	\item 
	Consider the case in which $C_1\cap C_2 =\emptyset$. In this case, one of the following must hold: ($C_1= [0,c_1]$ and $C_2 = [c_2,1]$) or ($C_1= [c_1,1]$ and $C_2=[0,c_2]$). Therefore, either $f_1$ is increasing and $f_2$ is decreasing, or vice versa. In both cases ,
	$\max \{  f_1(\gamma),f_2(\gamma)\} $ is either monotone, and thus its minimum can be easily found, or is decreasing up to some $\gamma^*$ which satisfies $f_1(\gamma^*) -f_2(\gamma^*)=0$, and then increasing. In the latter case, the minimum is at $\gamma^*$, which  can be easily found using binary search, as $f_1(\gamma)- f_2(\gamma)$ is monotone. 
	\item
	In case $C_1\cap C_2 \neq \emptyset $, let $\gamma^* \in C_1\cap C_2$. Then, it holds that 
	$$
	0\leq \ruleopt(\alpha, 2,2,\bb,\bk^1,\bk^2) = \min_{\gamma\in [0,1]}\max \{  f_1(\gamma),f_2(\gamma)\} \leq \max\{ f_1(\gamma^*) ,f_2(\gamma^*) \} = \max\{0,0\} =0.
	$$
	Therefore, in this case, $ \ruleopt(\alpha, 2,2,\bb,\bk^1,\bk^2)=0$ and an optimal value of $\bgam$ is $(\gamma^*,1-\gamma^*)$.  
\end{itemize}

\section{The Analysis of Two-variable Recurrence Relations}
\label{sec:rec}

In this section we give the proof of Theorem~\ref{thm:rec} which formalizes the asymptotic behavior of composite recurrences. Recall that a composite recurrence is defined by a set $\terms= \{(\bb^j,\bk^j,\bgam^j)~|~1\leq j\leq N\}$ of $N$ {\em terms}. The $j$-th term  $(\bb^j,\bk^j,\bgam^j)$ consists of three vectors, each of $r_j$ dimensions.
We require $\bb^j \in \mathbb{N}_{>0}^{r_j}$, $\bk^j\in \mathbb{N}_{\geq 0}^{r_j}$ and $\bgam^j\in \mathbb{R}_{>0}$. We further require that $\bk^j$ is not all zeros, and $\bgam^j$ is a distribution. The composite recurrence of $\terms$ is the function $p:\mathbb{Z}\times \mathbb{Z}\rightarrow [0,1]$ defined by
\begin{equation*}
	\begin{aligned}
		p(b,k) &= \min_{1\leq j \leq N}  \sum_{i=1}^{r_j} \bgam^j_i \cdot p(b-\bb^j_i, k-\bk^j_i) &\\
		p(b,k) &= 0 & \forall b<0, k \in \mathbb{Z} \\
		p(b,k) &= 1 & \forall  b\geq 0, k \leq 0
	\end{aligned}
\end{equation*}

Theorem~\ref{thm:rec} deals with the asymptotic behavior of composite recurrences. It relies on two technical notions: {\em critical ratio} and {\em branching-numbers}. The {\em critical ratio} of the term $(\bb^j,\bk^j,\bgam^j)$ is $\min_{1\leq i\leq r_j: \bk^j_i\neq 0} \frac{\bb^j_i}{\bk^j_i}$. In particular, it ensures that the following definition is sound. 
\alphabranching*
Theorem~\ref{thm:rec} states that $p(\floor{\alpha k},k)$ is dominated by the highest $\alpha$-branching number of a term in~$\terms$.
\main*

The proof of Theorem~\ref{thm:rec} consists of several stages which
together yield the statement of the theorem.

 The analysis relies on an equivalence between  a composite recurrence and the  probability of a rare event in a specifically designed random walk with an adversary. We describe the random walk and prove this equivalence in Section~\ref{sec:random_walk}. 

We used an adaptation of the {\em method of types} (see, e.g., \cite{Co06,C98}) to analyze the random walks associated with the recurrences. We elaborate on types and prove the basic properties associated with them in Section~\ref{sec:types}. 

Once the random walk and required tools are defined we can proceed to the core of the proof of Theorem~\ref{thm:rec}. 
The main idea in the proof is to consider the composite recurrence $p_{\gamma}$ of the terms  $\{(\bb^j,\bk^j,\bgam^j)\,|\,j\in N\}$, and a second composite recurrence $p_{\delta}$ which is defined by the terms  $\{(\bb^j,\bk^j,\bdel^j)\,|\,j\in N\}$. That is, in $p_{\delta}$ we change the probability vectors from $\bgam^j$ to $\bdel^j$. The vectors $\bdel^j$ are carefully selected vectors strictly inside the feasibility region of \eqref{eq:alpha_num}.   A main idea of our proof is to consider the probability of the  {\em same} event in two different random walks: one associated with $p_{\gamma}$ and the terms $\{(\bb^j,\bk^j,\bgam^j)\,|\,j\in N\}$, and another associated with $p_{\delta}$ and the terms $\{(\bb^j,\bk^j,\bdel^j)\,|\,j\in N\}$.

Towards this end, we first show the following lemma.
\begin{restatable}{lemma}{trivial}
	\label{lem:trivial_rec}
	Let $p$ be the composite recurrence of $\{(\bb^j, \bk^j, \delta^j ) |~ 1\leq j \leq N\}$, and  $\alpha>0$. Also, assume that $\bb^j \cdot \bdel^j < \alpha \cdot \bk^j\cdot \bdel^j$ for all $j\in [N]$. Then  $\lim_{k\rightarrow \infty} \frac{1}{k} \cdot \ln p(\floor{\alpha\cdot k},k)= 0$. 
\end{restatable} 
Lemma~\ref{lem:trivial_rec} gives a condition over the terms of the recurrence. If the condition holds then $p(\floor{\alpha \cdot k},k )$ is {\em high}.  The condition on $\bdel^j$ in Lemma~\ref{lem:trivial_rec} is  equivalent to ``$\bdel^j$ is in the interior of the feasiblity region of the optimization problem in \eqref{eq:alpha_num} which defines the branching number". Thus, intuitively the lemma states that, if the  probability vectors of the terms  are strictly inside  the feasibility region of \eqref{eq:alpha_num},  then $p(\floor{\alpha \cdot k} , k)$ is expected to be high. Furthermore, in such cases the $\alpha$-branching number of $(\bb^j,\bk^j,\bdel^j)$ is equal to zero, hence, Lemma~\ref{lem:trivial_rec} is a special case of Theorem~\ref{thm:rec}. 
The proof of Lemma~\ref{lem:trivial_rec}, given in Section~\ref{sec:trivial}, relies on the equivalence between recurrences and random walks, and utilizes the method of types. We precede the proof with an intuitive interpretation of the lemma through the lens of the random walk. 

The next stage in the proof of Theorem~\ref{thm:rec} is a lemma which shows a connection between the asymptotic behavior of two  different recurrences.
\begin{restatable}{lemma}{translation}
	\label{lem:translation}
	Let $p_{\gamma}$ be the composite recurrence of $\{(\bb^j, \bk^j, \bgam^j ) \,|\, 1\leq j \leq N\}$,  $p_{\delta}$ the composite recurrence of $\{(\bb^j, \bk^j, \bdel^j )\,|\, 1\leq j \leq N\}$, and $\alpha> 0$. 
	If $\lim_{k\rightarrow \infty} \frac{1}{k} \ln p_{\delta}(\floor{\alpha \cdot k},k)=0$ then
	$$\liminf_{k\rightarrow \infty} \frac{1}{k} \cdot \ln p_{\gamma}(\floor{\alpha\cdot k},k) \,\geq\, -\max_{j\in[N] } \frac{\D{\bdel^j}{\bgam^j}}{\bdel^j\cdot \bk^j}.$$
\end{restatable}
Observe that $p_{\delta}$ and $p_{\gamma}$, the recurrences in Lemma~\ref{lem:translation}, only differ in the distribution vectors $\bgam^j$ and $\bdel^j$ (for $j\in [N]$) which define their terms. The proof of the lemma, given in Section~\ref{sec:translation}, considers two random walks, one that is associated with $p_{\gamma}$ and another that is associated with $p_{\delta}$. The key idea is to consider an event which has a high probability in the random walk associated with $p_{\delta}$, and evaluate its probability in the random walk associated with $p_{\gamma}$. That is, the proof considers the probability of the same event in two different probability spaces.  The translation of probabilities between spaces follows from the method of types.

Together, Lemmas~\ref{lem:trivial_rec} and~\ref{lem:translation} also give some insight into the formula of   the branching numbers given in \eqref{eq:alpha_num}.
To use Lemma~\ref{lem:translation}, one may wish to choose values for $\delta^j$ which satisfy the conditions of Lemma~\ref{lem:trivial_rec} and yield the best possible bound over $p_{\gamma}$ according to Lemma~\ref{lem:translation}. Up to minor technicalities, this leads to Definition~\ref{def:alpha_branching}. The next lemma  follows this logic.
\begin{restatable}{lemma}{reclower}
	\label{lem:rec_lower}
	Let $p$ be the composite recurrence of $\{(\bb^j, \bk^j, \bgam^j ) |~ 1\leq j \leq N\}$, and  $\alpha>0$ such that $\alpha > \critical(\bb^j,\bk^j,\bgam^j)$ for $1\leq j\leq N$. Denote by $M_j$ the
	$\alpha$-branching number of $(\bb^j, \bk^j, \bgam^j)$, 
	and let $M=\max\{M_j | 1\leq j \leq N\}$.  
	Then,
	\[\liminf_{k\rightarrow \infty} \frac{1}{k}\cdot \log p\left(\floor{\alpha \cdot k} , k\right)  \geq  -M. \]
\end{restatable}
The proof of Lemma~\ref{lem:rec_lower} is given in Section~\ref{sec:proof_lower}. We note that Lemma~\ref{lem:rec_lower} suffices for all the algorithmic applications in this paper. However,  it is unclear if the lower bound given in  Lemma~\ref{lem:rec_lower}  is tight. To show that the bound is tight, thereby completing the proof of Theorem~\ref{thm:rec}, we show the following. 
\begin{restatable}{lemma}{recupper}
	\label{lem:rec_upper}
	Let $p$ be the composite recurrence of $\{(\bb^j, \bk^j, \bgam^j ) |~ 1\leq j \leq N\}$, and  $\alpha>0$ such that $\alpha > \critical(\bb^j,\bk^j,\bgam^j)$ for $1\leq j\leq N$. Denote by $M_j$ the
	$\alpha$-branching number of $(\bb^j, \bk^j, \bgam^j)$, 
	and let $M=\max\{M_j | 1\leq j \leq N\}$.  
	Then,
	\[\limsup_{k\rightarrow \infty} \frac{1}{k}\cdot \log p\left(\floor{\alpha k} , k\right) \leq  -M. \]
\end{restatable}
The proof of Lemma~\ref{lem:rec_upper} is given in Section~\ref{sec:proof_upper}. The proof of the lemma follows from a fairly simple application of the method of types. Theorem~\ref{thm:rec} follows immediately from Lemmas~\ref{lem:rec_lower}  and~\ref{lem:rec_upper}.

\subsection{Random Walk with an Adversary} 
\label{sec:random_walk}
In this section we describe the random walk and show its equivalence to composite recurrences. We start with intuitive description of the random walk followed by some formal definitions. We then show the equivalence between the recurrence and the random walk. To this end, we present an alternative formula for composite recurrences and then show the equivalence to the random walk using this formula.

\subsubsection{An Informal Description of the Random Walk} 
\label{sec:inf_walk}
\begin{figure}
	\centering
	\begin{tikzpicture}[scale=0.75]
		\draw[step=1cm,gray,very thin,dashed] (0,-1) grid (15,8);
		\draw[thick,->] (0,-1) -- (15,-1) ;
		\draw[thick,->] (0,-1) -- (0,8);
		
		\pgfsetlinewidth{0.0pt}

		\foreach \x in {0,1,2,3,4,5,6,7,8,9,10,11,12,13,14,15}
		\draw (\x cm,-1cm-2pt) -- (\x cm,-1cm-2pt)  node[anchor=north] {$\x$};
		
		\foreach \y in { -1,0,1,2,3,4,5,6,7,8}
		\draw (2pt,\y cm) -- (-2pt,\y cm) node[anchor=east] {$\y$};
		
		\filldraw[blue] (0,0) circle (3pt);
		
		\draw[blue,->,thin,dashed] (0,0) -- (2.9,0);
		\filldraw[blue] (3,0) circle (3pt) node[anchor=north ] {$(X_1,Y_1)=(3,0)$};
		
		\draw[blue,->,thin,dashed] (3.1,0.1) -- (6-0.1,3-0.1);
		\filldraw[blue] (6,3) circle (3pt) node[anchor=north west ] {$(X_2,Y_2)=(6,3)$};
		
		\draw[blue,->,thin,dashed] (6.1,3.1) -- (9-0.1,6-0.1);
		\filldraw[blue] (9,6) circle (3pt) node[anchor=north west] {$(X_3,Y_3)=(9,6)$};
		
		\draw[blue,->,thin,dashed] (9.1,6.1) -- (10-0.1,7-0.1);
		\filldraw[blue] (10,7) circle (3pt) node[anchor=south] {$(X_4,Y_4)=(10,7)$};
		
		\draw[blue,->,thin,dashed] (10.1,7) -- (11-0.1,7);
		\filldraw[blue] (11,7) circle (3pt) node[anchor=north west ] {$(X_5,Y_5)=(11,7)$};
		
	\end{tikzpicture}
	\caption{An instance of the first five steps  of a random walk associate with the composite recurrence in \eqref{eq:vc3intro}, i.e., the composite recurrence of $\left\{( \bb^j,\bk^j,\bgam^j)~|~j\in \{1,2\}\right\}$ where $\bgam^1 =\bgam^2 = (\gamma, 1-\gamma)$, $\bb^1=\bb^2= (1,3)$, $\bk^1 = (1,0)$ and $\bk^2=(0,3)$.   The walk corresponds to the adversary selecting $j_1= 1$, $j_2=2$, $j_3 = 2$ ,$j_4= 1$ and $j_5=2$, and random samples $i_1=2$, $i_2=2$, $i_3=2$, $i_4=1$ and $i_5=1$.
	}
	\label{fig:sec6_walk}
\end{figure}
Let $\left\{(\bb^j,\bk^j,\bgam^j)\,\middle|\, j\in[N]\right\}$ be a set of $N$ terms, and assume the term $(\bb^j,\bk^j,\bgam^j)$ is of length $r_j$. 
We consider a random walk which starts at $(X_0,Y_0)=(0,0)$ and at the $n$-th step is positioned at $(X_n,Y_n) \in \mathbb{N}\times \mathbb{N}$.  The walk is a generalization of the walk presented in Section~\ref{sec:int_walk} for a specific and simple composite recurrence.
At the $n$-th step  of the walk an {\em adversary} selects a value $j_n \in [N]$, and then a value $1\leq i_{n}\leq r_j$ is sampled according to $\bgam^{j_n}$. That is , $\Pr(i_n = i) = \bgam^{j_n}_i$ for all $1\leq i\leq r_{j_n}$. The next position of the walk is set to $$(X_n,Y_n) \, = \, (X_{n-1}, Y_{n-1}) +(\bb^{j}_{i}, \bk^{j}_{i}) \,=\, \left(X_{n-1}+\bb^{j}_{i}, Y_{n-1} +\bk^{j}_{i}\right),$$
where $i=i_n$ and $j=j_n$. That is, the position at the $n$-th step moves by $(\bb^{j_n}_{i_n},\bk^{j_n}_{i_n})$. In particular, the adversary selects the {\em term} $(\bb^j,\bk^j,\bgam^j)$ that will be used for the $n$-th step, but the step itself is still random.

We allow the adversary to make  its decision based  on the path made so far by the random walk. We define the {\em trace} of the walk up to step $n$ by $A_1= (j_1,i_1),A_2=(j_2,i_2),\ldots,A_n=(j_n,i_n)$. Observe that  that path of the walk up to the $n$-th step, $(X_0,Y_0),\ldots, (X_n,Y_n)$, is determined by the  trace of the walk up to the $n$-th step. We consider adversaries which set the value of $j_n$ based on the trace of the walk up to step $n-1$.  An illustration of the random walk is given in Figure~\ref{fig:sec6_walk}. Observe that in case $N=1$ the adversary always selects $j=1$, and the random walk is identical to the one defined in Section~\ref{sec:int_walk}.

Let  $p$ be the composite recurrence of $\left\{(\bb^j,\bk^j,\bgam^j)\,\middle|\, j\in[N]\right\}$. Our main claim is that for any $b,k\in \mathbb{Z}\times \mathbb{Z}$ it holds that $p(b,k)$ is the probability that in the above random walk there is a step $n\in \mathbb{N}$  in which $X_n\leq b$ and $Y_n\geq k$, when the adversary selects an optimal strategy that {\em minimizes} the probability of such an event.

\subsubsection{The Random Walk} 
\label{sec:random_walk_formal}
Following the motivation from the previous section, we  formally define the random walk.  
Let $\terms = \left\{(\bb^j,\bk^j,\bgam^j)\,\middle|\, j\in[N]\right\}$ be a set of $N$ terms, and assume the term $(\bb^j,\bk^j,\bgam^j)$ is of length $r_j$.  
In this formal definition, we first define the trace of the walk, and subsequently use the trace to define the walk itself. 
We define the alphabet associated with  the $j$-th term by $\chi_j = \{(j,i)\,|\,i\in [r_j]\}$ and the alphabet of the process by \begin{equation}
	\label{eq:chidef}
	\chi\,=\,\bigcup_{j=1}^{N} \chi_j\,=\, \left\{ (j,i) \,|\, j\in [N], \,i\in [r_j]\right\}.
\end{equation}

In particular, in the terminology of Section~\ref{sec:inf_walk}, the trace up to step $n$ is a vector in $\chi^n$.  We define a sequence  $(A_n)_{n=1}^{\infty}$ of random variables, where $A_n \in \chi$ for every $n\in \mathbb{N}_{>0}$. We associate the step $\left( \bb^j_i,\bk^j_i\right) $ with every $(j,i)\in \chi$.
Define $\kappa(j,i) = \bk^j_i$ and $\beta(j,i)=\bb^j_i$ for every $(j,i)\in \chi$. 
With a slight abuse of notation we define $\kappa( a_1,\ldots, a_n) = \sum_{\ell=1}^{n} \kappa(a_\ell)$ and  $\beta( a_1,\ldots, a_n) = \sum_{\ell=1}^{n} \beta(a_\ell)$ for every $n\in \mathbb{N}$ and $(a_1,\ldots, a_n)\in \chi^n$ .  Furthermore, define $\kappa(\emptyvec) =\beta(\emptyvec)= 0$  where $\emptyvec$ is the vector of dimension $0$. The position of the random walk after $n$ steps is $$\big(\beta(A_1,\ldots, A_n),\, \kappa(A_1\,\ldots ,A_n)\bigg)=\bigg(\beta(A_1,\ldots,A_{n-1}), \kappa(A_1,\ldots,A_{n-1}) \bigg) + \bigg( \beta(A_n),\, \kappa(A_n)\bigg).$$

Recall that $\chi^* = \bigcup_{n=0}^{\infty} \chi^n$ is the set of all vectors of finite dimension with entries in $\chi$.  A {\em strategy} is a function $S:\chi^*\rightarrow [N]$. We define $\mS$ to be the set of all strategies. The strategy reflects the choices made by the adversary. 

We define a random process which depends both on the terms $\terms=\left\{(\bb^j,\bk^j,\bgam^j)\,\middle|\, j\in[N]\right\}$  and a strategy $S\in \mS$.  
We use $(\Omega,\mF,\Pr_{S})$ to denote the probability space associated with the walk defined by strategy $S$.  Without loss of generality we assume the sample space $\Omega$ and event space $\mF$ do not depend on the strategy $S$. 

For every $j\in [N]$ define $(I^j_n)_{n\geq 1}$ as an infinite series of i.i.d where $\Pr_{S}(I^j_n = i )=\bgam^j_i$ for every $i\in [r_j]$. 
We define 
$A_n = (j, I^j_n )\in \chi$ where   $j=S(A_1,\ldots, A_{n-1})$. Observe that  the probability distribution of $I^j_n$  is determined  by $\bgam^j$ for every $j\in [N]$, and that the strategy $S$ determines the value of $j$ used for $A_n$ according to $A_1,\ldots, A_{n-1}$.

We define a random walk $\bigg((X_n,Y_n) \bigg)_{n=0}^{\infty}$, by $X_n = \beta(A_1,\ldots, A_n)$ and $Y_n = \kappa(A_1,\ldots, A_n)$ for every $n\geq 0$. 
We observe this random walk indeed matches the intuitive description where the strategy $S$ reflects the choices made by the adversary.
Before the $n$-th step the walk is positioned at $(X_{n-1},Y_{n-1})$, the adversary selects a value $j=S(A_1,\ldots, A_{n-1})$ which depends on the path taken by the process so far. Then $I^j_n$ is randomly selected and is distributed according to $\bgam^j$. Finally, the next position of the walk is $$(X_n,Y_n) \,=\, (X_{n-1},Y_{n-1}) + \left(\beta(j, I^j_n ),\, \kappa(j, I^j_n )\right) \,= \, (X_{n-1},Y_{n-1}) + \left(\bb^j_{I^j_n}, \bk^j_{I^j_n}\right).$$

For every $b,k\in \mathbb{Z}$  define the event \begin{equation}
	\label{eq:Gdef}
	G^{b,k}  = \{ \exists n\geq 0:~X_n\leq b \textnormal{ and } Y_n\geq k\}.
	\	\end{equation}
That is, $G^{b,k}$ is the event in which the walk crosses  $k$ on its $y$-axis before it crosses $b$ on the $x$-axis.  
We show the following connection between the random walk and $p$, the composite recurrence of $\terms$. 
\begin{lemma}
	\label{lem:comp_to_walk}
	Let $p$ be the composite recurrence of $\terms$, then $p(b,k) = \min_{S\in \mS} \Pr_{S} (G^{b,k})$ for all $b,k \in \mathbb{Z}$.
\end{lemma}
In words, $p(b,k)$ is the probability of the event $G^{b,k}$, when the adversary selects the strategy $S\in \mS$ which minimizes this probability.  We note that the lemma also implies  that $\min_{S\in \mS} \Pr_S(G^{b,k})$ is defined.
The remainder of Section~\ref{sec:random_walk} is dedicated to the proof of Lemma~\ref{lem:comp_to_walk}, which is simple yet slightly involved technically.
Before we prove Lemma~\ref{lem:comp_to_walk} we give an alternative formula for composite recurrences which replaces the $\min$ operation in \eqref{eq:rec_relation} with a strategy $S$. We note the alternative formula is only used as a stepping stone towards the proof of Lemma~\ref{lem:comp_to_walk}.  

\subsubsection{Strategic Composite Recurrences}

As in the previous sections, let  $\left\{(\bb^j,\bk^j,\bgam^j)\,\middle|\, j\in[N]\right\}$ be a set of $N$ terms and assume the term $(\bb^j,\bk^j,\bgam^j)$ is of length $r_j$.  
 Similarly, we use $\mS$ to denote the set of all strategies $S:\chi^*\rightarrow [N]$. 
 We define a new type of recurrences, {\em strategic composite recurrences}, which serve as a bridge between the composite recurrence and the random walk. Like composite recurrences they do not involve any probability space and similarly the random walk they involve a strategy. The idea is that instead of using the $\min$ operation in \eqref{eq:rec_relation} which picks a value of $j$, the  strategic composite recurrence would choose $j$ by a given strategy $S$. More specifically, the value of $j$ would be selected to be $S(\epsilon)$, the adversary selection for $j$  in the first step of the random walk. In the recurrence formula, the strategy~$S$ is replaced with a new strategy which emulates the adversary after a step. 
 
  Given a strategy $S$ and $(a_1,\ldots, a_n) \in \chi^n$ we define a new strategy 
$S_{(a_1,\ldots, a_n)}: \chi^*\rightarrow [N]$ by 
\begin{equation}
	\label{eq:strategy_sub}
	S_{(a_1,\ldots, a_n)}(c_1,\ldots, c_m ) = S( a_1,\ldots, a_n,c_1,\ldots,c_m)
\end{equation}
for every $(c_1,\ldots, c_m)\in \chi^*$.  That is, $S_{(a_1,\ldots, a_n)}$ is the strategy which $S$ uses in step $(n+1)$ and onward,
assuming the trace of the  first $n$ steps of the random walk is  $a_1,\ldots, a_n$.  In case $n=1$ we use $S_{a_1}=S_{(a_1)}$.  Intuitively, $S_{a_1}$ is the strategy the adversary uses as of step $1$ of the random walk, if $A_1=a_1$. 

The {\em strategic composite recurrence} of $\left\{(\bb^j,\bk^j,\bgam^j)\,\middle|\, j\in[N]\right\}$  is the function $\tp:\mS\times \mathbb{Z}\times \mathbb{Z}\rightarrow [0,1]$ defined by
\begin{equation}
	\label{eq:alternative_rec}
	\begin{aligned}
		\tp(S,b,k) &= \sum_{i=1}^{r_j} \bgam^j_i \cdot \tp\left(S_{(j,i)}, b-\bb^j_i, k-\bk^j_i\right) & \textnormal{ where $j=S(\emptyvec)$}\\
		\tp(S,b,k) &= 0 & \forall b<0, k \in \mathbb{Z} \\
		\tp(S,b,k) &= 1 & \forall  b\geq 0, k \leq 0
	\end{aligned}
\end{equation}
The formula of the strategic composite recurrence resembles the formula for composite recurrence \eqref{eq:rec_relation}.  Observe that instead of the minimum operation in \eqref{eq:rec_relation}, in \eqref{eq:alternative_rec} the value of $j$ in the recurrence is determined by $S(\emptyvec)$ (recall that $\emptyvec$ is the vector of dimension $0$).  Furthermore, note that in the recurrence in \eqref{eq:alternative_rec} the strategy $S$ is replaced with $S_{(j,i)}$. 

The following lemma states that composite recurrences are equivalent to strategic composite recurrences, if the chosen strategy is the one which minimizes the value of the function. 
\begin{lemma}
	\label{lem:alternative_equiv}
	Let $p$ and $\tp$  be the composite recurrence and strategic composite recurrence of $\left\{(\bb^j,\bk^j,\bgam^j)\,\middle|\, j\in[N]\right\}$ (resp.). Then $p(b,k) = \min_{S\in \mS} \tp(S,b,k)$ for every $b,k\in \mathbb{Z}$.
\end{lemma}
\begin{proof}
	We prove the lemma by showing inequalities in both directions. 
	\begin{claim}
		\label{claim:equiv_dir1}
		For every $b,k\in \mathbb{Z}$ and strategy $S\in \mS$ it holds that $p(b,k) \leq \tp(S,b,k)$. 
	\end{claim}
	\begin{claimproof}
		We prove the claim using a simple   induction on $b$. \\
		\noindent {\bf base case.} if $b<0$ then $\tp(S,b,k) = 0=p(b,k)$ by definition, and the claim holds. \\
		\noindent {\bf induction step.} let $b\geq 0$ and assume $p(b',k')\leq \tp(S',b',k')$ for every $b'<b$, $k'\in \mathbb{Z}$ and strategy $S'\in \mS$. 
		Let 
		$k\in \mathbb{Z}$ and $S\in \mS$. If $k\leq 0$ then 
		$\tp(S,b,k) =1= p (b,k)$ and the claim holds. It remains to handle the case where $k>0$. Let $j^*=S(\emptyvec)$. Then,
		$$
		\begin{aligned}
			\tp(S,b,k) \, &= \,  \sum_{i=1}^{r_{j^*}} \bgam^{j^*}_i \cdot \tp\left(S_{(j^*,i)}, b-\bb^{j^*}_i, k-\bk^{j^*}_i\right) \\
			&\geq \, \sum_{i=1}^{r_{j^*} }  \bgam^{j^*}_i \cdot p \left(b-\bb^{j^*}_i, k-\bk^{j^*}_i\right)\\
			&\geq \, \min_{1\leq j\leq N}  \sum_{i=1}^{r_{j} }  \bgam^{j}_i \cdot p \left(b-\bb^{j}_i, k-\bk^{j}_i\right)\\
			&=\, p(b,k)
		\end{aligned}
		$$
		The first equality holds by  the definition of strategic composite recurrences \eqref{eq:alternative_rec}. The first inequality follows from the induction hypothesis (recall that $\bb^j_i>0$). The last equality follows from the definition of composite recurrences \eqref{eq:rec_relation}. 
	\end{claimproof}
	
	The next claim essentially shows the inequality in the opposite direction to Claim~\ref{claim:equiv_dir1}.
	\begin{claim}
		\label{claim:equiv_dir2}
		For every $b,k\in \mathbb{Z}$ there exists a strategy $S^*\in \mS$ such that $p(b,k) = p(S^*,b,k)$. 
	\end{claim}
	\begin{claimproof}
		If $b<0$ then $p(b,k) = 0=\tp(S,b,k)$ for every strategy $S\in \mS$. Furthermore, if $b\geq 0$ and $k\leq 0$ then $p(b,k) = 1=\tp(S,b,k)$ for every strategy $S\in \mS$. Thus, it remains to handle the case where $k>0$ and $b\geq 0$.

		We define a strategy which mimics the outcome of the $\min$ operation in \eqref{eq:rec_relation}. 
		Recall we defined $\beta(j,i) = \bb^j_i$ and $\kappa(j,i)= \bk^j_i$ for every $(j,i)\in \chi$ , and $\beta(a_1,\ldots,a_n)  = \sum_{\ell=1}^{n}\beta(a_\ell)$ as well as $\kappa(a_1,\ldots,a_n)  = \sum_{\ell=1}^{n}\kappa(a_\ell)$   for all $a_1,\ldots,a_n\in \chi$. 
		Define a strategy $S^*:\chi^*\rightarrow [N]$ by 
		\begin{equation}
			\label{eq:Sstar_def}
			S^*(a_1,\ldots, a_n) \, = \,\argmin_{1\leq j \leq N} \sum_{i=1}^{r_j} \bgam^j_i \cdot p(b'-\bb^j_i, k'-\bk^j_i) \textnormal{ where } 
			b'=b-\beta(a_1,\ldots,a_n) \textnormal{ and } k'=k-\kappa(a_1,\ldots,a_n)
		\end{equation}
		
		We prove by reverse induction on $n$ that for every $n\in \mathbb{N}_{\geq 0}$ and $(a_1,\ldots, a_n) \in \chi^n$ it holds that 
		$p(b',k') = \tp(S^*_{(a_1,\ldots, a_n)}, b',k' )$ where 	$b'=b-\beta(a_1,\ldots,a_n)$  and  $k'=k-\kappa(a_1,\ldots,a_n)$\\
		\noindent{\bf base case.} Let $n>b$ and $(a_1,\ldots, a_n)\in \chi^n$.  Then $$b' =b-\beta(a_1,\ldots,a_n) = b-\sum_{\ell=1}^{n} \beta(a_{\ell})   \leq  b-\sum_{\ell=1}^{n}1 <0,$$ 
		where the first inequality holds as $\beta((j,i))=\bb^j_i \geq1$ for all $(j,i)\in \chi$.   Therefore, $p(b',k')= 0 =\tp(S^*_{(a_1,\ldots,a_n)} ,b',k')$. \\
		\noindent{\bf induction step.} Let $n\leq b$ and $(a_1,\ldots , a_n)\in \chi^n$. As in the induction hypothesis we use $b'=b-\beta(a_1,\ldots,a_n)$  and  $k'=k-\kappa(a_1,\ldots,a_n)$. If $b'<0$ then $p(b',k')=0=\tp(S_{(a_1,\ldots,a_n)},b',k')$, and similarly if $b'\geq 0$ and $k'\leq 0$ it holds that $p(b',k')=1=\tp(S_{(a_1,\ldots,a_n)},b',k')$. Thus, we only need to show the induction hypothesis holds in case $b'\geq 0$ and $k'>0$. Let $j^* = S^*_{(a_1,\ldots, a_n)}(\emptyvec )=S^*(a_1,\ldots,a_n)$. Then, 
		$$
		\begin{aligned}
			\tp(S^*_{(a_1,\ldots, a_n)}, b',k')\,&=\,  \sum_{i=1}^{r_{j^*}} \bgam^{j^*}_i \cdot \tp\left(\left(S^*_{ (a_1,\ldots, a_n)}\right)_{(j^*,i)}, b'-\bb^{j^*}_i, k'-\bk^{j^*}_i\right) \\
			&=\, \sum_{i=1}^{r_{j^*}} \bgam^{j^*}_i \cdot \tp\left(S^*_{ \left(a_1,\ldots, a_n,(j^*,i)\right)}, b'-\bb^{j^*}_i, k'-\bk^{j^*}_i \right) \\
			&=\, \sum_{i=1}^{r_{j^*}} \bgam^{j^*}_i \cdot p\left( b'-\bb^{j^*}_i, k'-\bk^{j^*}_i \right) \\
			&=\, \min_{1\leq j\leq N}\sum_{i=1}^{r_{j}} \bgam^{j}_i \cdot p\left( b'-\bb^{j}_i, k'-\bk^{j}_i \right)\,=\, p(b',k').
		\end{aligned}
		$$
		The first equality is due to \eqref{eq:alternative_rec}, and the second equality holds as $(S^*_{(a_1,\ldots,a_n)})_{a} = S^*_{(a_1,\ldots, a_n,a)}$. The third equality holds by the induction hypothesis: observe that
		$b'-\bb^{j^*}_i  = b-\beta(a_1,\ldots, a_n, (j^*,i) )$ and $k'-\bk^{j^*_i} = k-\kappa(a_1,\ldots,a_n,(j^*,i))$.  The fourth holds as $j^* = S^*(a_1,\ldots,a_n)$ and by the definition of $S^*$ \eqref{eq:Sstar_def}. The last equality follows from the definition composite recurrences \eqref{eq:rec_relation}.  Thus, we completed the proof of the induction step.
		
		Hence, we have that $p(b',k') = \tp(S^*_{\emptyvec},b',k')$ where $b'=b-\beta(\emptyvec) = b$ and $k'=\kappa(\emptyvec) = k$, which completes the proof of the claim. 
	\end{claimproof}
	
	By Claims~\ref{claim:equiv_dir1} and~\ref{claim:equiv_dir2}, it follows that $p(b,k)= \min_{S\in \mS} \tp(S,b,k)$ for all $b,k\in\mathbb{Z}$, which completes the proof of the lemma. 
\end{proof}

\subsubsection{From Strategic Recurrences to Random Walks}
\label{sec:strat_to_walk}

We still need to prove Lemma~\ref{lem:comp_to_walk}. To do so, we show an auxiliary lemma (Lemma~\ref{lem:alternative_to_walk}) that establishes equivalence between {\em strategic} composite recurrences and the random walk. Together with Lemma~\ref{lem:alternative_equiv}, this auxiliary lemma leads to the proof of Lemma~\ref{lem:comp_to_walk}.  As before, let  $\terms = \left\{(\bb^j,\bk^j,\bgam^j)\,\middle|\, j\in[N]\right\}$ be a set of $N$ terms, and assume the term $(\bb^j,\bk^j,\bgam^j)$ is of length $r_j$. Furthermore, we use the random variables as defined in Section~\ref{sec:random_walk_formal}. 

We use some additional properties of the random walk. 
\begin{definition}
	\label{def:consistent}
	Let $S\in \mS$ be a strategy and $(a_1,\ldots, a_n)\in \chi^n$. Also, let $a_{\ell} = (j_{\ell} ,i_{\ell})$ for every $1\leq \ell\leq n$. We say that $(a_1,\ldots, a_n)$ is {\em consistent} with $S$ if $j_{\ell} = S(a_1,\ldots, a_{\ell-1})$ for every $1\leq \ell\leq n$.
\end{definition} 

Recall that $A_n = (j, I^j_n)$ where $j=S(A_1,\ldots, A_{n-1})$). Thus, $(a_1,\ldots, a_n)$ is consistent with $S$ if the random variables $A_1,\ldots, A_n$ can potentially take the value $(a_1,\ldots, a_n)$, given that the adversary uses $S$ as a strategy.

\begin{lemma}
	\label{lem:inconsistent_prob}
	Let $S\in \mS$ be a strategy and $(a_1,\ldots, a_n)\in \chi^n$ such that $(a_1,\ldots, a_n)$ is {\bf not} consistent with $S$. Then $\Pr_S( (A_1,\ldots,A_n)=(a_1,\ldots ,a_n))=0$. 
\end{lemma}
\begin{proof}
We note that as $(a_1,\ldots, a_n)$ is not consistent with $S$, there is 
	an index  $1\leq \ell^*\leq n$ such that  $j_{\ell^*}\neq S(a_1,\ldots ,a_{\ell^*-1})$. Then,
	$$
	\begin{aligned}
		\Pr_S( (A_1,\ldots, A_n ) =(a_1,\ldots, a_n))\,& =\, \prod_{\ell=1}^n \Pr_S\left(A_{\ell} = a_{\ell}\, \middle| (A_1,\ldots, A_{\ell-1})=(a_1,\ldots,a_{\ell-1})\right)\\
		&\leq\, \Pr_{S}\left(A_{\ell^*} = a_{\ell^*}\, \middle| (A_1,\ldots, A_{\ell^*-1})=(a_1,\ldots,a_{\ell^*-1})\right)\,\\
		&= \frac{ \Pr_{S}\left((A_1,\ldots, A_{\ell^*-1})=(a_1,\ldots,a_{\ell^*-1}) \textnormal{ and } A_{\ell^*} =a_{\ell^*} \right)}{ \Pr_{S} \left((A_1,\ldots, A_{\ell^*-1})=(a_1,\ldots,a_{\ell^*-1})\right)}  \,=\,0
		,  	\end{aligned}
	$$
	where the last equality holds since $A_{\ell^*} = a_{\ell^*}$ and $(A_1,\ldots, A_{\ell^*-1})= (a_1,\ldots,a_{\ell^*-1})$  implies that $S(a_1,\ldots,a_{\ell^*-1}) = S(A_1,\ldots, A_{\ell^*-1}) = j_{\ell^*}$,  which does not  hold by the selection of $\ell^*$. 
\end{proof}

Recall that $S_{a_1}$ is defined in \eqref{eq:strategy_sub}. We use below the next observation. 
\begin{obs}
	\label{obs:consistency}
	For any $n\geq 1$, $(a_1,\ldots, a_n) \in \chi^n$ and strategy $S\in \mS$ is holds that $(a_1,\ldots,a_n)$ is consistent with $S$ if and only if $a_1$ is consistent with $S$ and $(a_2,\ldots, a_n)$ is consistent with $S_{a_1}$. 
\end{obs}	

We also use the following technical lemma.
\begin{lemma}
	\label{lem:one_let_rec}
	For every strategy $S\in \mS$, $n\in \mathbb{N}_{>0}$ and $(a_1,\ldots, a_n) \in \chi^n$ 
	it holds that 
	$$\Pr_{S}\left( (A_1,\ldots A_n)=(a_1,\ldots, a_n ) \right) 
	=\Pr_{S}(A_1 = a_1) \cdot \Pr_{S_{a_1}} \left( (A_1, \ldots,A _{n-1})  = (a_2,\ldots ,a_n) \right)$$
\end{lemma}
Observe the right-hand term involves two different probability distributions: $\Pr_S$ and $\Pr_{S_{a_1}}$, the first uses $S$ as a strategy, and the second uses $S_{a_1}$. 
Intuitively, the probability space $\Pr_{S_{a_1}}$ can be viewed as ``restart'' of the random walk after a single step, assuming $A_1=a_1$. The lemma formally reflects this intuition.

\begin{proof}[Proof of Lemma~\ref{lem:one_let_rec}]
	Consider the following cases.
	\begin{itemize}
		\item
		If $(a_1,\ldots, a_n)$ is not consistent with $S$ then, by Lemma~\ref{lem:inconsistent_prob}, $\Pr_S( (a_1,\ldots ,a_n))=0$. Furthermore, by Observation~\ref{obs:consistency}, $(a_1)$ is not consistent with $S$ or $(a_2,\ldots,a_n)$ is not consistent with $S_{a_1}$. Thus, by Lemma~\ref{lem:inconsistent_prob}, $\Pr_S(A_1=a_1) =0$, or $\Pr_{S_{a_1}}( (A_1,\ldots, A_{n-1}) = (a_2,\ldots ,a_n))=0$. We therefore have that 
		$$\Pr_{S}\left( (A_1,\ldots A_n)=(a_1,\ldots, a_n ) \right) 
		\,=\,0\,=\,\Pr_{S}(A_1 = a_1) \cdot \Pr_{S_{a_1}} \left( (A_1, \ldots,A _{n-1})  = (a_2,\ldots ,a_n) \right).$$
		\item 
		If $(a_1,\ldots, a_n)$ is  consistent with $S$, 
		let $a_{\ell} = (j_{\ell},i_{\ell})$ for every $1\leq \ell \leq n$.   Then,
		\begin{equation*}
			\begin{aligned}
				\Pr_{S}&\left( (A_1,\ldots A_n)=(a_1,\ldots, a_n ) \right) =
				\Pr_{S} \left(\forall 1\leq \ell \leq n: 
				I^{S(a_1,\ldots ,a_{\ell-1})}_\ell = i_{\ell}
				\right)\\
				&=
				\Pr_{S} \left( 
				I^{j_1}_1 = i_{1}
				\right)\cdot
				\Pr_{S} \left(\forall 2\leq \ell \leq n: 
				I^{j_{\ell}}_\ell = i_{\ell}
				\right)\\
				&= 	\Pr_{S} \left( 
				A_1 = a_{1}
				\right)\cdot \prod_{\ell =2}^n \bgam^{j_{\ell}}_{i_{\ell}}\\
				&=
				\Pr_{S} \left( 
				A_1 = a_{1}
				\right)\cdot
				\Pr_{S_{a_1}} \left(\forall 1\leq \ell \leq n-1: 
				I^{j_{\ell+1}}_\ell = i_{\ell+1}
				\right)\\
				&=
				\Pr_{S} \left( 
				A_1 = a_{1}
				\right)\cdot
				\Pr_{S_{a_1}} \left(\forall 1\leq \ell \leq n-1: 
				I^{S_{a_1}(a_2,\ldots,a_{\ell+1})}_\ell = i_{\ell+1}
				\right)\\
				&=
				\Pr_{S} \left( 
				A_1 = a_{1}
				\right)\cdot
				\Pr_{S_{a_1}} \left(
				(A_1,\ldots, A_{n-1}) = (a_2, \ldots , a_n)
				\right).
			\end{aligned}
		\end{equation*}
		The first and fifth equalities  hold as $(a_1,\ldots,a_n)$ and $(a_2,\ldots,a_n)$ are consistent with $S$ and $S_{a_1}$ (resp.). The second equality holds as the variables $I^j_n$ are independent. The third and fourth equalities  follow from $\Pr_S(I^j_\ell = i)=\Pr_{S_{a_1}}(I^j_\ell = i) = \bgam^j_i$. 
	\end{itemize}
	As the statement of the lemma holds in both cases, this completes the proof. 
\end{proof}

We use the above to show equivalence between a strategic recurrence and the random walk which uses the same strategy.
\begin{lemma}
	\label{lem:alternative_to_walk}
	Let $\tp$ be the stragetic composite recurrence of $\terms$, then $\tp(S,b,k) = \Pr_{S} (G^{b,k})$ for every $b,k\in \mathbb{Z}$ and strategy $S\in \mS$.
\end{lemma}
\begin{proof}
	
	For every $b\geq 0$ and $k>0$ define 
	\begin{equation}
		\label{eq:Z_def}
		Z^{b,k}= \left\{ (a_1, \ldots ,a_n )\in \Chi^* \middle| \begin{array}{l}
			\beta(a_1,\ldots, a_n)\leq b\\
			\kappa(a_1,\ldots, a_n)\geq k \\
			\kappa(a_1,\ldots, a_{n-1})< k
		\end{array}\right\}.
	\end{equation}
	That is, $Z^{b,k}$ is the collection of all vectors $(a_1,\ldots, a_n)\in \chi^*$ such that if $(A_1,\ldots, A_n) = (a_1,\ldots, a_n)$ then the $n$-th position of the walk, $(X_n,Y_n)$, is the {\em first} position in which $Y_n \geq k$, and furthermore, $X_n\leq b$ (recall that $X_n =\beta(A_1,\ldots, A_n)$ and $Y_n=\kappa(A_1,\ldots, A_n)$).  We note that the vectors in $Z^{b,k}$ are of varying dimensions. An important property of $Z^{b,k}$ is that it is {\em prefix free}; that is, if $(a_1,\ldots, a_n)\in Z^{b,k}$ then $(a_1,\ldots, a_{\ell}) \notin Z^{b,k}$ for every $\ell < n$. We also define $Z^{b,k} = \emptyset$ for $b<0$ and arbitrary $k$, and $Z^{b,k} =\{\emptyvec\}$ in case $b\geq 0$ and $k\leq 0$.  
	
	The set $Z^{b,k}$ has a recursive structure. It can be easily verified that
	\begin{equation}
		\label{eq:recZ}
		Z^{b,k}= \left\{ (a_1, \ldots ,a_n )\in \Chi^* \middle| 
		(a_2,\ldots,a_n) \in Z^{b-\beta(a_1), k-\kappa(a_1)} 
		\right\},
	\end{equation}
	for every $b\geq 0$ and $k>0$.
	Recall the event $G^{b,k}$ defined in~\eqref{eq:Gdef}. We can formulate $G^{b,k}$ using $Z^{b,k}$ by
	\begin{equation}
		\label{eq:GviaZ}
		\begin{aligned}
			G^{b,k}  \,&= \, \{ \exists n\geq 0:~X_n\leq b \textnormal{ and } Y_n\geq k\}\\
			\, &=\,  \{ \exists n\geq 0:~X_n\leq b,  Y_{n}\geq k \textnormal{ and }  Y_{n-1}<k\} \\
			&=\,  \{ \exists n\geq 0:~\beta(A_1,\ldots,A_n)\leq b \textnormal{ and } \kappa(A_1,\ldots,A_n)\geq k \textnormal{ and }  \kappa(A_1,\ldots,A_{n-1})<k\} \\
			\, &=\,  \{ \exists (a_1,\ldots,a_n)\in Z^{b,k}: ~(A_1,\ldots, A_n) = (a_1,\ldots, a_n)\} 
		\end{aligned}
	\end{equation}
	for all $b,k\in \mathbb{Z}$ such that $b\geq 0$ and $k>0$.   Since $Z^{b,k}$ is prefix free the events $(A_1,\ldots, A_n) =(a_1,\ldots,a_n)$ for different vectors $(a_1,\ldots,a_n) \in Z^{b,k}$ are disjoint. Therefore, for every strategy $S\in \mS$, $b\geq0$ and $k>0$, we have
	\begin{equation}
		\label{eq:strat_to_wakl1}
		\begin{aligned}
			\Pr_{S}& \left(G^{b,k} \right) \,=\, \sum_{(a_1,\ldots ,a_ n) \in Z^{b,k}} \Pr_{S}\left( (A_1,\ldots,A_n) =(a_1,\ldots, a_n) \right) \\
			&= \,\sum_{(a_1,\ldots ,a_ n) \in Z^{b,k}}  \Pr_{S}\left(A_1=a_1\right)\cdot \Pr_{S_{a_1}}\left( (A_1,\ldots,A_{n-1}) =(a_2,\ldots, a_n)\right)\\
			&= \, \sum_{(j,i)\in \chi} \sum_{~(a_2,\ldots ,a_ n) \in Z^{b-\beta(j,i),k-\kappa(j,i)}~}  \Pr_{S}\left(A_1= (j,i)\right)\cdot \Pr_{S_{(j,i)}}\left( (A_1,\ldots,A_{n-1}) =(a_2,\ldots, a_n)\right)\\
			& = \, \sum_{(j,i)\in \chi} \Pr_{S}\left(A_1= (j,i)\right)\cdot \sum_{~(a_2,\ldots ,a_ n) \in Z^{b-\bb^j_i,k-\bk^j_i}~}   \Pr_{S_{(j,i)}}\left( (A_1,\ldots,A_{n-1}) =(a_2,\ldots, a_n)\right)\\
			&= \, \sum_{(j,i)\in \chi} \Pr_{S}\left(A_1= (j,i)\right) \cdot \Pr_{S_{(j,i)}}\left(G^{b-\bb^j_i, k-\bk^j_i}\right).
		\end{aligned}
	\end{equation}
	The first and last equalities follow from the representation of $G^{b,k}$ via  $Z^{b,k}$, as given in~\eqref{eq:GviaZ}. The second equality uses Lemma~\ref{lem:one_let_rec}, the third equality holds due to the recursive structure of $Z^{b,k}$, as given in~\eqref{eq:recZ}, and the fourth equality uses $\beta(j,i ) = \bb^j_i$ and $\kappa(j,i) =\bk^j_i$. 
	
	By the definition of the random variable $A_1$, it holds that $A_1=\left(S(\emptyvec), I^{S(\emptyvec)}_1\right)$; therefore, $\Pr(A_1 = (j,i))=\bgam^{j}_i$ if $j=S(\emptyvec)$ and $\Pr(A_1 = (j,i))=0$ if $j\neq S(\emptyvec)$. Plugging this observation into~\eqref{eq:strat_to_wakl1}, we get that for every $b\geq 0$, $k>0$ and a strategy $S\in \mS$, it holds that
	\begin{equation}
		\label{eq:strat_to_walk2}
		\begin{aligned}
			\Pr_{S}& \left(G^{b,k} \right) 
			\,= \, \sum_{(j,i)\in \chi} \Pr_{S}\left(A_1= (j,i)\right) \cdot \Pr_{S_{(j,i)}}\left(G^{b-\bb^j_i, k-\bk^j_i}\right)\\
			&=\, \sum_{i=1}^{r_{j^*}} \bgam^{j^*}_i \cdot \Pr_{S_{(j^*,i)}}\left(G^{b-\bb^{j^*}_i, k-\bk^{j^*}_i}\right),
		\end{aligned}
	\end{equation}
	where $j^*=S(\emptyvec)$. 
	
	We use \eqref{eq:strat_to_walk2} to  show that $\tp(S,b,k) = \Pr_{S} (G^{b,k})$ for every $b,k\in \mathbb{Z}$ and strategy $S\in \mS$, by induction on $b$.\\
	\noindent{\bf base case.} If $b<0$ then $\tp(S,b,k) = 0 = \Pr_{S} (G^{b,k})$, due to \eqref{eq:Gdef} and \eqref{eq:alternative_rec}. \\
	\noindent{\bf induction step.}  Let $b\geq 0$, $k\in \mathbb{Z}$ and $S\in \mS$.
	In case $k\leq 0$ it holds that  $\tp(S,b,k) =1= \Pr_{S} (G^{b,k})$   due to \eqref{eq:Gdef} and \eqref{eq:alternative_rec}. If $k>0$ then by \eqref{eq:strat_to_walk2}, we have
	$$\Pr_S(G^{b,k}) \,=\,  \sum_{i=1}^{r_{j^*}} \bgam^{j^*}_i \cdot \Pr_{S_{(j^*,i)}}\left(G^{b-\bb^{j^*}_i, k-\bk^{j^*}_i}\right) \,=\,  \sum_{i=1}^{r_{j^*}} \bgam^{j^*}_i \cdot \tp(S_{(j^*,i)} ,b-\bb^{j^*}_i,k-\bk^{j^*}_i) \, =\, \tp(S,b,k),
	$$
	where $j^*=S(\emptyvec)$. The second equality follows from the induction hypothesis, and the last equality is due to the definition of strategic composite recurrences \eqref{eq:alternative_rec}. Thus completes the induction step.
	
	Overall, we showed that  $\tp(S,b,k) = \Pr_{S} (G^{b,k})$ for every $b,k\in \mathbb{Z}$ and strategy $S\in \mS$, as required. 
\end{proof} 

Lemma~\ref{lem:comp_to_walk} is a simple consequence of Lemmas~\ref{lem:alternative_equiv} and~\ref{lem:alternative_to_walk}.
\begin{proof}[Proof of Lemma~\ref{lem:comp_to_walk}]
	Let  $\tp$  be the composite recurrence and strategic composite recurrence of $\terms$. 
	Then for every $b,k\in \mathbb{Z}$, we have
	$$
	p(b,k) \,=\, \min_{S\in \mS} \tp(S,b,k) \,=\, \min_{S\in \mS} \Pr_{S}\left( G^{b,k}\right),
	$$
	where the first equality follows from Lemma~\ref{lem:alternative_equiv} and the second equality uses Lemma~\ref{lem:alternative_to_walk}. 
\end{proof}

\subsection{Types}
\label{sec:types}
Our analysis relies on the notion of {\em types}. As in the previous sections, we fix a set of $N$ terms  $\terms = \left\{(\bb^j,\bk^j,\bgam^j)\,\middle|\, j\in[N]\right\}$ where  $(\bb^j,\bk^j,\bgam^j)$ is of length $r_j$. We refer to the random process and the random variables associated with it (e.g., $A_n$, $I^j_n$, $X_n$ and $Y_n$), as defined in Section~\ref{sec:random_walk_formal}.  Recall that the set $\chi$ is defined in \eqref{eq:chidef}.

The {\em type}
of  $(a_1, \ldots, a_n) \in \Chi^n$, denoted  $\type(a_1, \ldots, a_n) =T\in \nonneg^{\chi}$,  is defined by $T_a=\frac{|\{\ell | a_{\ell }= a \}|}{n}$ for every $a\in \chi$. That is, $T_a$ is the frequency of each $a\in \Chi$  in $(a_1,\ldots, a_n)$. 
For example, the type $T$ of $( (1,1),(2,1),(1,1)) \in \chi^3$ is $T_{(1,1)} =\frac{2}{3}$, $T_{(2,1)}=\frac{1}{3}$ and $T_{(j,i)}=0$ for every other $(j,i)\in \chi$. 
Observe that $\sum_{a\in \chi } T_a = 1$, thus the type $T$ can be interpreted as a distribution over $\chi$. 

Surprisingly, this simple notion is highly powerful in proving various combinatorial and probabilistic properties, if those can be expressed in terms of types (see, e.g., \cite{C98,Co06}). 
In our case, as we show below, the position $(X_n,Y_n)$ of the random walk can be expressed in terms of the type of the random variables $A_1,\ldots, A_n$. Subsequently, the events $G^{b,k}$ (defined in~\eqref{eq:Gdef}) can be expressed in terms of types. As $p(b,k)$ can be expressed as the probability of the event $G^{b,k}$ (see Lemma~\ref{lem:comp_to_walk}), this allows us to obtain an estimate for~$p(b,k)$.

As defined above, $\kappa(j,i) =\bk^{j}_i$ ($\beta(j,i) =\bb^{j}_i$) for all $(j,i)\in \chi$, and $\kappa(a_1,\ldots,a_n) =\sum_{\ell=1}^{n }\kappa(a_{\ell})$ ($\beta(a_1,\ldots,a_n) =\sum_{\ell=1}^{n }\beta(a_{\ell})$) for every $n\in \mathbb{N}$ and $(a_1,\ldots, a_n)\in \chi^n$. With a slight abuse of notation, we extend the definition of $\kappa$, and $\beta$ to types. Given $T\in \mathbb{R}_{\geq 0}^{\chi}$, we define $\kappa(T) = \sum_{ a \in \chi } T_{a}\cdot \kappa(a)$ and $\beta(T) = \sum_{ a \in \chi } T_{a}\cdot \beta(a)$. That is, $\kappa(T)$ is the expected value of $\kappa(A)$, assuming $A\in\chi$ is a random variable such that $\Pr(A=a) = T_a$ for all $a\in \chi$. 

For every $n\in \mathbb{N}$ and $(a_1,\ldots, a_n) \in \chi^n$, it holds that
\begin{equation}
	\label{eq:kappa_by_type_explicit}
\kappa(a_1,\ldots, a_n) \,=\,  \sum_{\ell =1}^{n} \kappa(a_{\ell}) \,=\, \sum_{a\in \chi } \kappa(a) \cdot \left|\{\ell | a_{\ell }=a\} \right| \, = \, \sum_{a\in \chi } \kappa(a) \cdot n\cdot T_a \,=\, n\cdot \kappa(T),
\end{equation}
and 
\begin{equation}
	\label{eq:beta_by_type}
	\beta(a_1,\ldots, a_n) \,=\,  \sum_{\ell =1}^{n} \beta(a_{\ell}) \,=\, \sum_{a\in \chi } \beta(a) \cdot \left|\{\ell | a_{\ell }=a\} \right| \, = \, \sum_{a\in \chi } \beta(a) \cdot n\cdot T_a \,=\, n\cdot \beta(T),
\end{equation}
where $T= \type(a_1,\ldots, a_n)$.
Therefore, 
\begin{equation}
	\label{eq:kappa_by_type}
	(X_n, Y_n) \,=\, (\beta (A_1,\ldots ,A_n), \kappa(A_1,\ldots ,A_n)  ) \,=\, (n\cdot \beta(\type(A_1,\ldots, A_n)), n\cdot \kappa(\type(A_1,\ldots ,A_n))).
\end{equation}
That is, the position of the random walk after $n$ steps is a function of the type of $A_1,\ldots,A_n$.  Consequently, 
$$
G^{b,k } \,=\, \{ \exists n\geq 0:~X_n\leq b \textnormal{ and } Y_n\geq k\}  = \, \{ \exists n\geq 0: \type(A_1,\ldots, A_n) \in Q_{b,k,n}\},$$
where $Q_{b,k,n} = \left\{ T\in \mathbb{R}_{\geq 0}^{\chi} ~|~\beta(T)\leq \frac{b}{n} \textnormal { and }  \kappa(T)\geq \frac{k}{n} \right\}$. Therefore, 
$$
p(b,k) \, =\, \min_{S\in \mS} \Pr_S (G^{b,k}) \, =\, \min_{S\in \mS} \Pr_{S} \left(  \exists n\geq 0: \type(A_1,\ldots, A_n) \in Q_{b,k,n}\right).
$$

We note that the {\em method of types} exactly deals with estimation of probabilities of the form $\Pr(\type(A_1,\ldots ,A_n) \in Q)$, and thus promises to be useful in our case. In its standard form, the method is used for independent random variables $A_1,\ldots, A_n$ \cite{Co06}. However, this is not the case in our random process (due to the existence of multiple terms and the adversary). To overcome this hurdle, we show that some properties of types can be adjusted for our random process	. 

The first property of types is the somewhat trivial observation that the number of types of vectors of length $n$ is {\em polynomial} in $n$. 
This is in contrast to the number of vectors of length $n$, which is {\em exponential} in $n$.
Define 
\begin{equation}
	\label{eq:Kndef}
	\mK_n  = \left\{  \frac{1}{n} \cdot \bm ~\Big|~ \bm \in \{0,1,\ldots, n\}^{\chi}   \right\}.
\end{equation}
In the above definition, the set $\mK_n$ consists of  all the scalar by vector multiplications of the scalar $\frac{1}{n}$ and a vector $\bm$ of dimension $|\chi|= \sum_{j=1}^{N} r_j$ with integral entries in the range $0$ to $n$. 
\begin{obs}
	\label{obs:types}
	For every $n\in \mathbb{N}$ and $(a_1,\ldots ,a_n) \in \chi^n$ it holds that $\type(a_1,\ldots, a_n) \in \mK_n$. Furthermore,   $|\mK_n| = (n+1)^{|\chi|}$.
\end{obs}
We note that the above observation is independent of the random process, and is indeed proved in classic textbooks (e.g., Theorem 11.1.1 in \cite{Co06}).  Observation~\ref{obs:types} is commonly used to estimate the probability of an event via the probability of the most common type satisfying the event. That is, for every  $Q\subseteq \mathbb{R}^{\chi}_{\geq 0}$ it holds that  
$$\Pr_S(\type(A_1,\ldots, A_n)\in Q) \,\leq\, (n+1)^{|\chi|}\cdot \max_{T\in \mK_n \cap Q} \Pr_S(\type(A_1,\ldots, A_n)=T),$$
and 
$$ \Pr_S(\type(A_1,\ldots, A_n)\in Q)   \,\geq\, \max_{T\in \mK_n \cap Q} \Pr_S(\type(A_1,\ldots, A_n)=T).$$
That is, $$\Pr_S(\type(A_1,\ldots, A_n)\in Q)\, \approx\,  \max_{T\in \mK_n \cap Q}   \Pr_S(\type(A_1,\ldots, A_n)=T),$$ up to polynomial factors. 

The second property of types we use below is that the probability of the event $(A_1,\ldots, A_n) =(a_1,\ldots,a_n)$ only depends on the type of $(a_1,\ldots,a_n)$ and on whether $(a_1,\ldots, a_n)$ is consistent with the strategy $S$ (see Definition~\ref{def:consistent}). 
This property is analogous to a classic property of types implying the probability that a sequence of $n$ {\em i.i.d.} $A_1, \ldots , A_n$ taking values in $\chi$ satisfies $(A_1, \ldots , A_n) = (a_1,\ldots,a_n) \in \chi^n$ only depends on the type of $(a_1,\ldots, a_n)$ (see Theorem 11.1.2 in \cite{Co06}).

The {\em entropy} of a type $T\in \mK_n$ is given by
$$
\entropy(T )\, =\, \sum_{a\in \chi} T_a\cdot \ln \frac{1}{T_a}\,.
$$
We also (symbolically) extend the definition of Kullback-Leibler divergence  to measure the divergence between a type $T\in \mK_n$ and the vectors $\left(\bgam^{j}\right)_{j=1}^{N}$. With a slight abuse of notation, we use $\bgam$ to refer to the $N$ vectors $\left(\bgam^{j}\right)_{j=1}^{N}$.  We define 
\begin{equation}
	\label{eq:div_type}
	\D{T}{\bgam} \,= \, \sum_{(j,i)\in \chi } T_{(j,i)}\cdot \ln \frac{T_{(j,i)}}{\bgam^j_i}.
\end{equation}
\begin{lemma}
	\label{lem:type_prob}
	Let $(a_1, \ldots a_n) \in \Chi^n$,  $T=\type(a_1,\ldots, a_n)$ and $S\in \mS$. If $(a_1,\ldots,a_n)$ is consistent with the strategy $S$ then 
	$$\Pr_S\bigg( (A_1,\ldots, A_n ) = (a_1,\ldots, a_n)\bigg)  \, = \, \exp\bigg(  -n \left( \entropy(T) + \D{T}{\bgam}\right) \bigg).
	$$
\end{lemma}

\begin{proof}
	Let $a_{\ell }= (j_{\ell}, i_{\ell})$ for every $1\leq \ell \leq n$. Recall that the case in which $(a_1,\ldots,a_n)$ is not consistent with $S$ was considered in Lemma~\ref{lem:inconsistent_prob}. Assume then that $(a_1,\ldots ,a_n)$ is consistent with $S$. Thus, $j_{\ell } =S(a_1,\ldots,a_{\ell-1})$ for all $1\leq \ell \leq n$. It follows that
	\begin{equation*}
		\begin{aligned}
			\Pr_{S}&\left( (A_1,\ldots A_n)=(a_1,\ldots, a_n ) \right) \,=\,
			\Pr_{S} \left(\forall 1\leq \ell \leq n: 
			I^{j_{\ell}}_\ell = i_{\ell}
			\right)
			&= 	 \,\prod_{\ell =1}^n \bgam^{j_{\ell}}_{i_{\ell}}
			&= 	\, \prod_{(j,i)\in \chi}^n \left( \bgam^j_i\right)^{n\cdot T_{(j,i)}},
		\end{aligned}
	\end{equation*}
	where the last equality holds as $(j,i)$ appears $n\cdot T_{(j,i)}$ many times in $a_1,\ldots, a_n$. Using simple algebraic manipulations, we have
	\begin{equation*}
		\begin{aligned}
			\Pr_{S}&\left( (A_1,\ldots A_n)=(a_1,\ldots, a_n ) \right) \,=\,	\, \prod_{(j,i)\in \chi}^n \left( \bgam^j_i\right)^{n\cdot T_{(j,i)}} \\ 
			&= \,\exp \left(  -n \cdot \sum_{ (j,i)\in \chi } T_{(j,i)} \cdot \ln \frac{1}{\bgam^{j}_i} \right) \\
			&= \exp \left(  -n \left(  \sum_{ (j,i)\in \chi } T_{(j,i)} \cdot \ln \frac{T_{(j,i)}}{\bgam^{j}_i} + \sum_{(j,i)\in \chi} T_{(j,i)} \cdot \ln \frac{1}{T_{(j,i)}} \right) \right) \\
			&= \exp\left( -n \cdot \left( \D{T}{\bgam} +\entropy(T)\right)\right).
		\end{aligned}
	\end{equation*}
\end{proof}

Furthermore, we can upper bound the number of vectors $(a_1,\ldots, a_n)$ of type $T$ which are consistent with a strategy. To this end, we use the following classic bound from the method of types.
\begin{lemma}[Theorem 11.1.3 in \cite{Co06}]
	\label{lem:classic_type_bound}
	Let $T\in \mK_n$. The number of vectors $(a_1,\ldots, a_n) \in \chi^n$ of type $T$ is at most $\exp\left( n\cdot \entropy (T)\right)$. That is,
	$$\left| \left\{ (a_1,\ldots, a_n)\in\chi^n\,\middle| \, \type(a_1,\ldots, a_n) =T\right\}\right| \, \leq \, \exp\left( n\cdot \entropy(T)\right).$$
\end{lemma}

We define {\em term entropy} of a type $T$ by 
\begin{equation}
	\label{eq:terment}
	\terment(T) = \sum_{j=1}^{N}  \lambda_j \cdot \ln \frac{1}{\lambda j} \textnormal{ where } \lambda_j =\sum_{i=1}^{r_j} T_{(j,i)}.
\end{equation}
Let $(a_1,\ldots, a_n)\in \chi^n$ be a vector of type $T$. 
Then the term entropy is the entropy of the vector $(\lambda_1,\ldots,\lambda_N)$ where $\lambda_j$ is the frequency of letters from the alphabet $\chi_j = \{(j,i)~|~i\in [r_j]\}$ in the vector $(a_1,\ldots, a_n)$.  Informally, the vector $(\lambda_1,\ldots, \lambda_N)$ is the type of the vector $(j_1,\ldots, j_{n})\in [N]^n$, where $a_{\ell} =(j_{\ell},i_{\ell})$ for every $1\leq \ell \leq n$. 
\begin{lemma}
	\label{lem:const_type_bound}
	Let $T\in \mK_n$ and $S\in \mS$ a strategy. Then the number of vectors $(a_1,\ldots, a_n) \in \chi^n$ of type $T$  which are consistent with $S$ is at most
	$\exp\left(n\cdot  \left( \entropy(T)-\terment(T)\right) \right)$. 
	That is, 
	$$\begin{aligned} \left| \left\{ (a_1,\ldots, a_n)\in \chi^n\,\middle| \, \type(a_1,\ldots, a_n) =T\textnormal{ and $(a_1,\ldots,a_n)$ is consistent with $S$} \right\}\right| \, \\
		\leq \, \exp\left(n\cdot  \left( \entropy(T)-\terment(T)\right) \right).
	\end{aligned}$$
\end{lemma}
We emphasize that while Lemma~\ref{lem:classic_type_bound} counts vectors of type $T$, Lemma~\ref{lem:const_type_bound} adds the restriction that the vectors have to be consistent with the strategy $S$. The number of such vectors is smaller, as reflected by the $-\terment(T)$ factor in the upper bound given in Lemma~\ref{lem:const_type_bound}. 
\begin{proof}[Proof of Lemma~\ref{lem:const_type_bound}]
	Define 
	$$V = \left\{ (a_1,\ldots, a_n)\in \chi^n\,\middle| \, \type(a_1,\ldots, a_n) =T\textnormal{ and $(a_1,\ldots,a_n)$ is consistent with $S$} \right\}.$$
	Also, recall that $\chi_j=\{(j,i)\,|\,i\in [r_j]\}	 =\{(j',i')\in \chi\,|\,j'=j\}$.
	
	In order to upper bound $|V|$ we define a function $\varphi:V \rightarrow \chi_1^*\times \chi_2^*\times \ldots \times \chi_N^*$. We then bound the cardinality of its image\footnote{The {\em image} of a function $f:X\rightarrow Y$ is $\{f(x) |~ x\in X\}$ and denoted $\Ima(f)$.} and show it is an injection to get the required upper bound on $|V|$. 
	We use~$N$ functions, $\varphi^j$ for every $j\in [N]$, to define $\varphi$.

	For $1\leq j \leq N$ let $\varphi^j(a_1, \ldots, a_n)$ be the result of removing from $(a_1, \ldots, a_n)$ all entries that do not belong to $\chi_j$. Formally,
	$\varphi^j:V\rightarrow \chi_j$ is defined by  
	$\varphi^j(a_1, \ldots, a_n)=(a_{\pi(1)}, \ldots, a_{\pi(h)})$,
	where  $h= |\{ \ell~|~a_{\ell}\in \chi_j \}|$ is the number of entries of $(a_1,\ldots,a_n)$ in $\chi_j$, and $\pi:[h]\rightarrow  \{\ell~|~a_{\ell}\in \chi_j \} $ is the unique monotone one-to-one function from $\{1,\ldots, h\}$ to entries in $(a_1,\ldots,a_n)$ which are in $\chi_j$. 
	We  define $$\varphi(a_1,\ldots,a_n)= \left(\varphi^1(a_1,\ldots ,a_n), \ldots, \varphi^N(a_1,\ldots,a_n)\right)$$ 
	for every $(a_1,\ldots, a_n)\in V$. That is, the $j$-th entry of $\varphi(a_1,\ldots,a_n)$ is  $\varphi^j(a_1,\ldots,a_n)$, the substring of $a_1,\ldots,a_n$  which contains only entries from $\chi_j$.   For example, $\varphi^2( (2,1),(1,1), (1,3),(2,2)) = ((2,1),(2,2))$.

	\begin{claim}
		\label{claim:image_bound}
		$\left|\Ima(\varphi)\right|\, \leq \,\exp\left( n\cdot \left(  \entropy(T) - \terment(T)\right)\right)$.
	\end{claim}
	\begin{claimproof}
		We first show that the type of $\varphi^j(a_1,\ldots,a_n)$ is the same for every $(a_1,\ldots,a_n)\in V$. 
		For every $1\leq j \leq N$ define $\lambda_j = \sum_{i=1}^{r_j} T_{(j,i)} =\sum_{a\in \chi_j} T_a$. The value $\lambda_j$ is the frequency of elements from $\chi_j$ in a vector $(a_1,\ldots,a_n) \in\chi^n$ such that $\type(a_1,\ldots, a_n) = T$.
		By Observation~\ref{obs:types}, it holds that $n\lambda_j$ is integral, and it can be easily verifiedthat $\varphi^j(a_1,\ldots,a_n) \in \chi_j^{\lambda_j n }$ for every $1\leq j \leq N$ and $(a_1,\ldots,a_n)\in V$.
		
		For $1\leq j \leq N$ such that $\lambda_j\neq 0$, define  $T^j\in \nonneg^{\chi}$ by $T^j_{(j,i)}=\frac{1}{\lambda_j}\cdot  T_{(j,i)}$ for $(j,i)\in \chi_j$, and $T^j_{(j',i')}=0$ for $(j',i')\in \chi \setminus \chi_j$. For $1\leq j \leq N$ such that $\lambda_j=0$ define $T^j=\bar{0}\in \nonneg^r$ .
		
		For every $j\in [N]$ and $(a_1,\ldots,a_n)\in V$, the frequency of $(j,i)\in \chi$ in $\varphi^j(a_1,\ldots, a_n)$ is $\frac{T_{(j,i)}}{\lambda_j}$. This is the result of dividing   $n\cdot T_{(j,i)}$, the number of times $(j,i)$ appears in $(a_1,\ldots,a_n)$,  by $n\cdot \lambda_j$, the dimension of $\varphi^j(a_1,\ldots,a_n)$. Following this argument, 
		it is easy to verify that if $\lambda_j\neq 0$ then $\type(\varphi^j(a_1,\ldots,a_n)) =T^j$ for every $(a_1,\ldots,a_n)\in V$. That is, $T^j$ is the type of the substring of  $(a_1,\ldots,a_n ) \in V$ which contains only entries in $\chi_j$. 
		
		Define $$U^j =\left\{(a_1,\ldots, a_{\lambda_j\cdot n}) \in \chi^{\lambda_j \cdot n }~\middle|~ \type(a_1,\ldots, a_{\lambda_j\cdot n }) = T^j \right\}.$$ It follows that $\Ima(\varphi^j)\subseteq U^j$, and  $\Ima(\varphi)\subseteq U^1 \times U^2 \times \ldots \times U^N$,
		all the vectors of dimension $\lambda_j\cdot n$ with entries in $\chi$ whose type is $T^j$. 
		By Lemma~\ref{lem:classic_type_bound}, it holds that $$
		\begin{aligned}
			|U^j|\,&\leq\,  \exp\left(\lambda_j\cdot n \cdot \entropy(T^j)\right)\,\\
			&=\, \exp\left( \lambda_j\cdot n \sum_{i=1}^{r_j} T^j_{(j,i)} \cdot \ln \frac{1}{T^j_{(j,i)}}\right) \\
			&=\, \exp\left( \lambda_j\cdot n \sum_{i=1}^{r_j} \frac{T_{(j,i)}}{\lambda_j} \cdot \ln \frac{\lambda_j}{T_{(j,i)}}\right)\\
			&=\,\exp\left( n\cdot \left(\sum_{i=1}^{r_j} T_{(j,i) } \cdot \ln \frac{1}{T_{(j,i)}} -\lambda_j \cdot \ln \frac{1}{\lambda_j}\right)\right).
		\end{aligned}$$  
		The second equality follows from $T^j_{(j,i)} =\frac{T_{(j,i)}}{\lambda_j}$ (by definition), and the last equality holds as $\lambda_j=\sum_{i=1}^{r_j} T_{(j,i)}$. 
		Hence, 
		\begin{equation*}
			\begin{aligned}
				|\Ima(\varphi)|\, &\leq\, |U^1| \cdot |U^2| \cdot \ldots \cdot |U^N|\\ 
				&\leq\, \exp\left( n\cdot \sum_{j=1}^{N} \left(\sum_{i=1}^{r_j} T_{(j,i) } \cdot \ln \frac{1}{T_{(j,i)}} -\lambda_j \cdot \ln \frac{1}{\lambda_j}\right)\right),
				\\
				&=\,\exp\left( n\cdot \left(  \entropy(T) - \terment(T)\right)\right). 
			\end{aligned}
		\end{equation*}
	\end{claimproof}
	\begin{claim}
		\label{claim:injection}
		$\varphi$ is an injection.
	\end{claim}
	\begin{claimproof}
		Let $(a_1, \ldots, a_n ), (d_1, \ldots ,d_n)\in V$ such that $\varphi(a_1, \ldots, a_n) = \varphi(d_1, \ldots, d_n)$. Assume towards contradiction  that $(a_1, \ldots, a_n )\neq (d_1, \ldots ,d_n)$. Let $\ell$ be the minimal index such that $a_{\ell} \neq d_{\ell}$. 
		Since both $(a_1,\ldots,a_n)$ and $(d_1,\ldots,d_n)$ are consistent with $S$,  it holds that  $a_{\ell}, d_{\ell} \in \chi_j$
		where  $j=S(a_1, \ldots ,a_{\ell-1})=S(d_1, \ldots ,d_{\ell-1})$.

		The proof idea is that since $a_{\ell}\neq d_{\ell}$, it must also hold that $\varphi^j(a_1,\ldots, a_n) \neq \varphi^j(d_1,\ldots,d_n)$. On the other hand, we must have that $\varphi^j(a_1,\ldots, a_n) = \varphi^j(d_1,\ldots,d_n)$ since $\varphi(a_1,\ldots,a_n)=\varphi(d_1,\ldots,d_n)$. 
		
		As $\varphi(a_1,\ldots,a_n) = \varphi(d_1,\ldots,d_n)$, it follows that $\varphi^j(a_1,\ldots,a_n) = \varphi^j(d_1,\ldots, d_n )$. 
		Let $h=|\{\ell'~|~a_{\ell' }\in \chi_j\}| = |\{\ell'~|~d_{\ell' }\in \chi_j\}|$ (the number of entries in $(a_1,\ldots,a_n)$ and $(d_1,\ldots,d_n)$ from $\chi_j$ is the same since the type of both is $T$).
		Let $\pi:[h]\rightarrow \{\ell'~|~a_{\ell' }\in \chi_j\}$ and $\sigma:[h]\rightarrow \{\ell'~|~d_{\ell' }\in \chi_j\}$  be the unique monotone one-to-one functions such that $\varphi^j(a_1,\ldots,a_n) = (a_{\pi(1)},\ldots, a_{\pi(h)})$ and $\varphi^j(d_1,\ldots,d_n) = (d_{\sigma(1)},\ldots, d_{\sigma(h)})$.
		Let $k=|\{ \ell'<\ell\,| a_{\ell'} \in \chi_j\}| = |\{ \ell'<\ell\,| d_{\ell'} \in \chi_j\}|$ be the number of entries in $(a_1,\ldots,a_{\ell -1})=(d_1,\ldots,d_{\ell -1})$ from $\chi_j$. By the definition of $\sigma$ and $\pi$ it must hold that $\sigma(k+1) = \pi(k+1) = \ell$. Furthermore, $a_{\ell}=a_{\pi(k+1)} = d_{\sigma(k+1) }= d_{\ell}$ (since $\varphi^j(a_1,\ldots,a_n) = \varphi^j(d_1,\ldots, d_n )$), which contradicts the selection of~$\ell$. Thus, $(a_1,\ldots,a_n) =(d_1,\ldots,d_n)$. That is, $\varphi$ is an injection.
	\end{claimproof}
	
	Since $\varphi$ is an injective function (Claim~\ref{claim:injection}),  by Claim~\ref{claim:image_bound} we  have
	\begin{equation*}
		|V| =|\Ima(\varphi)|  
		\leq \exp\left( n\cdot \left(  \entropy(T) - \terment(T)\right)\right). 
	\end{equation*}
\end{proof}

The method of types can be used to upper bound the probability  that  a sequence of $n$ {\em i.i.d.} is of a given type (e.g., Theorem 11.1.4 in \cite{Co06}). 
Using Lemmas~\ref{lem:type_prob} and~\ref{lem:const_type_bound}, we obtain a similar upper bound which applies to our setting. 
\begin{lemma}
	\label{lem:all_type_prob}
	Let $n\in \mathbb{N}_{>0}$, $T\in \mK_n$, and $S\in\mS$. Then, $$\Pr_S\left(\type(A_1,\ldots,A_n) = T\right)\, \leq \, \exp\bigg(- n\cdot \big( \D{T}{\bgam} +\terment(T)\big)\bigg).$$ 
\end{lemma}
\begin{proof}
	The lemma follows from a simple calculation.
	$$
	\begin{aligned}
		\Pr_S&\left(\type(A_1,\ldots,A_n) = T\right)\,\\
		&=\,\sum_{\begin{array}{c}(a_1,\ldots,a_n)\in \chi^n \textnormal{ s.t. }\\ \type(a_1,\ldots,a_n) = T\end{array}} \Pr_S\left((A_1,\ldots,A_n)= (a_1,\ldots,a_n)\right)\\
		&=\,\sum_{\begin{array}{c}(a_1,\ldots,a_n)\in \chi^n \textnormal{ s.t. }\\ \type(a_1,\ldots,a_n) = T \textnormal{ and}\\
				\textnormal{ $(a_1,\ldots,a_n)$ is consistent with $S$}
		\end{array}}  \Pr_S\left((A_1,\ldots,A_n)= (a_1,\ldots,a_n)\right)\\
		&=\, \left| \left\{ (a_1,\ldots,a_n)\in\chi^n~\middle|~\begin{array}{c}\type(a_1,\ldots,a_n) = T \textnormal{ and}\\
			\textnormal{ $(a_1,\ldots,a_n)$ is consistent with $S$}
		\end{array} \right\}\right| 
		\cdot \exp\left( - n\cdot \left( \D{T}{\bgam}+\entropy(T)\right)\right)\\
		&\leq \,  \exp\left( n\cdot \left(\entropy(T) - \terment(T)\right)\right)
		\cdot \exp\left( - n\cdot \left( \D{T}{\bgam}+\entropy(T)\right)\right)\\
		& = \, \exp\bigg(- n\cdot \big( \D{T}{\bgam} +\terment(T)\big)\bigg).
	\end{aligned}
	$$
	The second equality follows from Lemma~\ref{lem:inconsistent_prob}, the third equality holds due to Lemma~\ref{lem:type_prob},  and the inequality follows from Lemma~\ref{lem:const_type_bound}. 
\end{proof}

In particular, Lemma~\ref{lem:all_type_prob} implies that if $\Pr_S\left(\type(A_1,\ldots,A_n) = T\right)$ is large then $ \D{T}{\bgam} +\terment(T)$ is {small}. 

The last property we need to establish provides a connection between $ \D{T}{\bgam} +\terment(T)$ and an $\ell_1$-like measure of distance between $T$ and $\bgam$. 
We will use this property, together with Lemma~\ref{lem:all_type_prob}, to argue that if $\Pr(\type(A_1,\ldots, A_n) =T)$ is high then, in some sense, the $\ell_1$ distance between $T$ and $\bgam$ must be small.

A  type $T\in \mK_n$ is a distribution, whereas $\bgam$ represents the $N$ vectors $\bgam^1, \ldots, \bgam^N$, and each of these is a distribution. To measure the distance between $T$ and $\bgam$ we should take this difference into account. One way to do that is to scale
$\bgam$. We measure the distance between $T$ and  $(\lambda_j \cdot \bgam^j_i)_{(j,i)\in \chi}$, where $\lambda_j = \sum_{i=1}^{r_j} T_{(j,i)}$. Intuitively, this matches the frequency of each of the alphabets $\chi_1,\ldots, \chi_N$ to its frequency in $T$. We prove the following.
\begin{lemma}
	\label{lem:norm}
	Let $n\in \mathbb{N}_{>0}$ and $T\in \mK_n$.  Then 
	$$
	\sum_{(j,i)\in \chi} \left| T_{(j,i)} -\lambda_j \cdot \bgam^j_i \right| \,\leq \, 2\cdot \sqrt{\D{T}{\bgam} + \terment(T)},
	$$
	where $\lambda_j = \sum_{i=1}^{r_j} T_{(j,i)}$ for every $i\in [N]$. 
\end{lemma}
The proof of Lemma \ref{lem:norm}  uses the  next result (Lemma 11.6.1 cf. \cite{Co06}). 
\begin{lemma}
	\label{lem:KLtol1}
	For every two distributions  $\bup^1,\bup^2 \in \nonneg^n$, it holds that 
	$$ \left( \sum_{i=1}^{n} |\bup_i^1-\bup_i^2| \right)^2 \,\leq\,   2\cdot \D{\bup^1}{\bup^2}.$$
\end{lemma}
\begin{proof}[Proof of Lemma \ref{lem:norm}]
	For every $j\in [N]$ define a distribution $\bt^{j }\in \mathbb{R}_{\geq 0}^{r_j}$ by $\bt^j_i = \frac{T_{(j,i)}}{\lambda_j}$ for all $i\in [r_j]$. In case $\blam_j=0$ we let $\bt^j$ be an arbitrary distribution.  By  Lemma~\ref{lem:KLtol1}, we have
	$$
	\left(\sum_{i=1}^{r_j} \left| \bt^j_i -\bgam^j_i \right|\right)^2 \, \leq \,2\cdot \D{\bt^j}{\bgam^j} 
	$$
	for every  $j\in [N]$. By the above,
	\begin{equation}
		\label{eq:norms_eq1}
		\begin{aligned}
			2\cdot 	\sum_{j=1}^{N} \lambda_j \cdot \D{\bt^j}{\bgam^j}\, &\geq\, \sum_{j=1}^{N} \lambda_j \cdot  \left(\sum_{i=1}^{r_j} \left| \bt^j_i -\bgam^j_i \right|\right)^2 \, \\
			&\geq \, \left(\sum_{j=1}^{N}\sum_{i=1}^{r_j }\blam_j \cdot \left| \bt^j_i -\bgam^j_i\right| \right)^2 \,\\
			&=\,  \left(\sum_{(j,i)\in \chi}  \left| T_{(j,i)} -\lambda_j\cdot \bgam^j_i\right| \right)^2,
		\end{aligned}
	\end{equation}
	where the second inequality follows from Jensen inequality as $x^2$ is convex, and the  equality follows from $\bt^j_i=\frac{T_{(j,i)}}{\lambda_j}$.  Furthermore,
	\begin{equation}
		\label{eq:norms_eq2}
		\begin{aligned}
			\sum_{j=1}^{N} \lambda_j \cdot \D{\bt^j}{\bgam^j} \, &= \, \sum_{j=1}^{N} \lambda_j \cdot \sum_{i=1}^{r_j} \bt^j_i \cdot \ln \frac{\bt^j_i}{\bgam^j_i} \,\\& = 
			\, \sum_{j=1}^{N} \lambda_j \cdot \sum_{i=1}^{r_j} \frac{T_{(j,i)}}{\lambda_j} \cdot \ln \frac{ \left(\frac{T_{(j,i)}}{\lambda_j}\right)}{\bgam^j_i} \\
			&=\,\sum_{j=1}^{N}  \sum_{i=1}^{r_j} {T_{(j,i)}} \cdot \ln \frac{ T_{(j,i)}}{\bgam^j_i} + \sum_{j=1}^{N} \lambda_j \cdot \ln \frac{1}{\lambda_j}\\
			&= \D{T}{\bgam} + \terment(T),
		\end{aligned}
	\end{equation}
	where the first equality follows from the definition of $\D{\cdot}{\cdot}$, the second  from the definition of $\bt^j_i$ and the last equality follows from the definitions of divergence \eqref{eq:div_type} and term entropy \eqref{eq:terment}. 
	
	By \eqref{eq:norms_eq1} and \eqref{eq:norms_eq2} we have, 
	$$\sum_{(j,i)\in \chi} \left| T_{(j,i)} -\lambda_j\cdot \bgam^j_i\right|  \, \leq \, \sqrt{ 2\cdot \sum_{j=1}^{N} \lambda_j \cdot \D{\bt^j}{\bgam^j}  } \, \leq  \, 2 \cdot \sqrt{ \D{T}{\bgam} +\terment(T)}.
	$$
\end{proof}

	\subsection{Random Walks with High Probability}
	\label{sec:trivial}
		We can now proceed to  the proof of Lemma~\ref{lem:trivial_rec}. We first restate the lemma.
\trivial*

We prove the lemma using the equivalence between the random walk and the composite recurrence.
Let $\terms=\{(\bb^j,\bk^j,\bdel^j)\,|\,j\in N\}$ and $\alpha$    which satisfy the condition in Lemma~\ref{lem:trivial_rec}. 
Also, consider the random walk associated with $\terms$ as defined in Section~\ref{sec:random_walk_formal}. 
We use the same notations as in Section~\ref{sec:random_walk_formal} to denote the random variables, sets, probability space etc, associated with the random walk. 
Also, recall that by Lemma~\ref{lem:comp_to_walk} it holds that 
\begin{equation}
	\label{eq:high_prop_eq}
	p(
	\floor{\alpha\cdot k},k) \,=\, \min_{S\in \mS} \Pr_S(G^{\floor{\alpha\cdot k},k})\,=\, \min_{S\in \mS} \Pr_S\left(\exists n\geq 0 :~X_n\leq \floor{\alpha \cdot k} \textnormal{ and }Y_n\leq k\right).
\end{equation}

Intuitively, in the $n$-th step of the random walk the adversary $S\in \mS$  selects $j\in [N]$ based on the history of the walk, and subsequently a value $I$ is selected such that $\Pr_S(I = i)=\bdel^j_i$; the position of the walk then moves by $(\beta(j,I) ,\kappa(j,I)) = (\bb^j_I, \bk^j_I)$.\footnote{Note that $ \Pr_S(I=j) =\bdel^j_i$ since we consider the random walk associated with the terms $(\bb^j,\bk^j,\bdel^j)$ for $j\in [N]$.}
Thus, the expected movement on the  $x$-axis is $\bb^j\cdot \bdel^j$, and the expected movement on the $y$-axis is $\bk^j \cdot \bdel^j$. By the condition of the lemma $\bb^j\cdot \bgam^j < \alpha \cdot \bk^j\cdot \bgam^j $ for all $j\in [N]$, thus we expect the ratio between the movement on the $x$-axis and the $y$-axis to be less than~$\alpha$, regardless of the adversary's choices. 
That is, we expect $X_n\leq \alpha Y_n$ to hold throughout the walk.  
In particular, we anticipate $G^{\alpha \cdot k,k }$ to occur with high probability against every adversary. By \eqref{eq:high_prop_eq} this means that $p(\floor{\alpha \cdot k },k)$ is expected to be high, which in turn implies Lemma~\ref{lem:trivial_rec}. 

For a fixed $k\in \mathbb{N}$ we take the ``worst'' adversary  $S^*\in \mS$, for which $\Pr_{S^*}(G^{\floor{\alpha \cdot k},k}) = \min_{S\in \mS}\Pr_{S}(G^{\floor{\alpha \cdot k},k})$.  The proof of Lemma~\ref{lem:trivial_rec} considers a specific value $n\in \mathbb{N}$ such that $\Pr_{S^*}( \alpha k - O(1)<X_n\leq\alpha k) \geq \Omega\left( \frac{1}{k}\right)$, and then focuses on a type $T\in \mK_n$ such that 
\begin{equation}
	\label{eq:trivial_proof_flow}
	\Pr_{S^*}( \type(A_1,\ldots, A_n ) =T \textnormal{ and }\alpha k - O(1)<X_n\leq\alpha k) \geq \frac{1}{\poly(k)}.
\end{equation}
Both $n$ and $T$ are found using the pigeon hole principle. Subsequently, the proof proceeds to show that $n\cdot \kappa(T) \geq k$. This suffices to show that $\Pr_S(G^{\floor{\alpha \cdot k,k}})\geq \frac{1}{\poly(k)}$, as it implies that the event $\left\{\type(A_1,\ldots,A_n) = T \textnormal{ and  } X_n\leq \alpha \cdot k\right\}$ is contained in $G^{\floor{\alpha\cdot k},k}$, and the probability of the former is lower bounded in \eqref{eq:trivial_proof_flow}. 
To show $n\cdot \kappa(T)\geq k$ we use Lemma~\ref{lem:all_type_prob} to argue that $\D{T}{\bdel} +\terment(T)$ is close to zero. Combining this with Lemma~\ref{lem:norm}, we have that $T$ and $\bdel$ are close. Once this is established, we use the property of the vectors $\bdel^j$ ($\bb^j\cdot \bdel^j < \alpha \cdot \bk^j\cdot\bdel^j $ for all $j\in [N]$) to show that $n\cdot \kappa(T)\geq k$.  
Finally, $\Pr_S(G^{\floor{\alpha \cdot k,k}})\geq \frac{1}{\poly(k)}$ together with \eqref{eq:high_prop_eq} imply the statement of Lemma~\ref{lem:trivial_rec}. 

\begin{proof}[Proof of Lemma~\ref{lem:trivial_rec}]
	As  $p(b,k)\in [0,1]$, it also holds that $\limsup_{k\rightarrow \infty} \frac{1}{k} \cdot \ln p(\floor{\alpha \cdot k},k)\leq 0$. This means that in order to prove the lemma, it suffices to show that $\liminf_{k\rightarrow \infty} \frac{1}{k} \cdot \ln p(\floor{\alpha \cdot k},k)\geq 0$.
	
	Let $k>M$, where $M>\frac{2}{\alpha}$ will be determined later in the proof.\footnote{The value of $M$ does not depend on $k$.} Let $S^*\in \mS$ such that $\Pr_{S^*}(G^{\floor{\alpha\cdot k},k}) = \min_{S\in \mS } \Pr_S(G^{\floor{\alpha\cdot k},k})$. 
	Also, let $s = \ceil{\alpha \cdot k  } $.
	Observe that $$X_{s+1}\,=\, \beta(A_1,\ldots, A_{s+1} ) \,=\,\sum_{\ell =1}^{s+1} \beta(A_\ell)\, \geq\, s+1\,>\,\alpha \cdot k,$$
	since $\beta(j,i)=\bb^j_i\geq 1$ for every $(j,i)\in \chi$. 
	Also, define $\bmax= \max_{a\in \chi} \beta(a) = \max_{j\in [N]}\max_{i\in [r_j]}\bb^j_i$ and~$t = \floor{\frac{\alpha \cdot k}{\bmax}}$. Then,
	$$
	X_t\,=\,\beta(A_1,\ldots,A_{t}) \, \leq t \cdot \bmax \leq \alpha k .
	$$
	Therefore, 
	$$ 1= \sum_{n=t}^{s} \Pr_{S^*} \left( X_n\leq \alpha\cdot  k \textnormal{ and } X_{n+1} > \alpha \cdot k \right).$$
	By the above inequality and the pigeon hole principle, there is $t\leq n \leq s$ such that 
	$$
	\Pr_{S^*} \left( X_{n}\leq \alpha\cdot  k \textnormal{ and } X_{n+1} > \alpha \cdot k \right) \,\geq \, \frac{1}{s}\, \geq \frac{1}{\alpha \cdot k+1}. 
	$$ 
	By Observation~\ref{obs:types}, the type of $A_1,\ldots, A_n$ is in $\mK_n$; thus,
	$$
	\begin{aligned}
		\sum_{T\in \mK_n}  \Pr_{S^*}& \left( X_{n}\leq \alpha\cdot k \textnormal{ and } X_{n+1}>\alpha\cdot  k \textnormal{ and } \type(A_1,\ldots, A_n)=T\right)
		\\
		&=\,\Pr_{S^*} \left( X_{n}\leq \alpha\cdot k \textnormal{ and } X_{n+1} > \alpha \cdot k \right) \,\geq \frac{1}{\alpha \cdot k+1}.
	\end{aligned}
	$$
	Since $|\mK_n|\leq (n+1)^{|\chi|}\leq (\alpha \cdot k +2)^{|\chi|}$ and by the pigeon hole principle, there is $T\in \mK_n$ such that 
	\begin{equation}
		\label{eq:piegon_T}
		\begin{aligned}
			\Pr_{S^*}&\left( X_{n}\leq \alpha \cdot k \textnormal{ and } X_{n+1}>\alpha \cdot k \textnormal{ and } \type(A_1,\ldots, A_n)=T\right) \\
			&\geq\, \frac{1}{(n+1)^{|\chi|}}\cdot \frac{1}{\alpha \cdot k+1} \geq \left(\frac{1}{ 2\cdot \alpha \cdot k }\right)^{|\chi|+1},
		\end{aligned}
	\end{equation}
	where the last inequality follows from  $\alpha \cdot k + 2 \leq 2\cdot \alpha \cdot k$ since $k>M>\frac{2}{\alpha}$. 
	Next, we want to show that $\{\type(A_1,\ldots, A_n ) =T\}\subseteq G^{\floor{\alpha \cdot k},k}$. This is done by showing that $n\cdot \beta(T) \leq\floor{\alpha \cdot k}$ and $n\cdot \kappa (T)\geq k$, using the next two claims.
	\begin{claim}
		\label{claim:betaT}
		$\alpha \cdot k-\bmax\leq n\cdot \beta(T) \leq \floor{\alpha \cdot k}$
	\end{claim}
	\begin{claimproof}
		By \eqref{eq:piegon_T}, it holds that 
		$$ 
		\begin{aligned}
			\Pr_{S^*}&\left(\beta(A_1,\ldots,A_n)\leq \alpha \cdot k \textnormal{ and } \beta(A_1,\ldots,A_{n+1})>\alpha \cdot k \textnormal{ and } \type(A_1,\ldots, A_n)=T \right)  \\
			&=\Pr_{S^*}\left(X_{n}\leq \alpha \cdot k \textnormal{ and } X_{n+1}>\alpha \cdot k \textnormal{ and } \type(A_1,\ldots, A_n)=T \right) >0,
		\end{aligned}$$
		Hence, there is $(a_1,\ldots ,a_{n+1})\in \chi^{n+1}$  such that $\beta(a_1,\ldots, a_n) \leq \alpha \cdot k$ ,$\beta(a_1,\ldots, a_{n+1}) >\alpha \cdot k$ and $\type(a_1,\ldots, a_n) =T$.
		Therefore,
		$$n\cdot \beta(T)  = \beta(a_1,\ldots, a_n) \leq \alpha \cdot k,$$
		where the first equality follows from \eqref{eq:beta_by_type}.  Since $n\cdot \beta(T)$ is integral, this implies $n\cdot \beta(T)\leq \floor{\alpha \cdot k}$. 
		
		Similarly,
		$$n\cdot \beta(T) \,=\,\beta(a_1,\ldots, a_n) \, = \, \beta(a_1,\ldots, a_{n+1})- \beta(a_{n+1}) \, \geq \alpha \cdot k - \bmax,
		$$
		where the last inequality holds as $\beta(a_{n+1})\leq \bmax$. 
	\end{claimproof}

	Let $\rho = \max_{j\in [N]} \frac{\bdel^j \cdot \bb^j }{\bdel^j\cdot\bk^j}$. By the condition of the lemma, it holds that $\frac{\bdel^j \cdot \bb^j}{\bdel^j\cdot \bk^j}<\alpha$ for every $j\in [N]$, therefore $\rho < \alpha$.   Also, define $\eta = \max_{(j,i)\in \chi }\left|  \bk_i^j - \frac{\bb^j_i}{\rho}\right|$. 
	\begin{claim}
		\label{claim:kappaT}
		$n\cdot \kappa(T) \geq k\cdot \left(  \frac{\alpha}{\rho}  - \frac{\bmax}{\rho} \cdot \frac{1}{k} -4\cdot  \alpha  \cdot \eta \cdot \sqrt{\D{T}{\bdel} -\terment(T)}\right) $.
	\end{claim}
	\begin{claimproof}
		Define $\lambda_j = \sum_{i=1}^{r_j } T_{(j,i)}$ for every $j\in [N]$. It holds that,
		\begin{equation}
			\label{eq:kappaT_first}
			\begin{aligned}
				n \cdot& \left( \kappa(T) - \frac{\beta(T)}{\rho }\right) = n\cdot \sum_{j=1}^{N} \sum_{i=1}^{r_j} T_{(j,i)}\cdot \left( \bk^j_i -\frac{\bb^j_i}{\rho} \right)\\
				&\geq n\cdot \sum_{j=1}^{N} \sum_{i=1}^{r_j} T_{(j,i)}\cdot \left( \bk^j_i -\frac{\bb^j_i}{\rho} \right) - n\cdot \sum_{j=1}^{N}
				\lambda_j \sum_{i=1}^{r_j}  \bdel^j_i\cdot \left( \bk^j_i -\frac{\bb^j_i}{\rho} \right)\\
				& = n\cdot \sum_{j=1}^{N} \sum_{i=1}^{r_j} \left( T_{j,i} - \lambda_j \cdot \delta^j_i\right)  \cdot \left(\bk^j - \frac{\bb^j_i}{\rho }\right)\\
				& \geq -n \cdot \eta \cdot \sum_{j=1}^{N} \sum_{i=1}^{r_j } \left|t_{(j,i)} -\lambda_j \cdot \bdel^j_i\right|\\
				&\geq  -4\cdot   \alpha \cdot k\cdot \eta \cdot \sqrt{\D{T}{\bdel} -\terment(T)} .
			\end{aligned}
		\end{equation}
		The first inequality holds since $\rho = \max_{j\in [N]} \frac{\bdel^j \cdot \bb^j }{\bdel^j\cdot\bk^j}$; thus, $\sum_{i=1}^{r_j} \delta^j_i \cdot \left(\bk^j_i -\frac{\bb^j_i}{\rho} \right) \geq 0 $ for every $j\in [N]$.  The second inequality follows from the triangle inequality and the definition of $\eta$. The last inequality follows from Lemma~\ref{lem:norm} and from $n\leq \ceil{\alpha \cdot k }\leq 2\cdot \alpha \cdot k$.  By rearranging the terms in \eqref{eq:kappaT_first}, we have
		$$
		\begin{aligned}
			n\cdot \kappa(T)  &\geq \frac{n\cdot \beta(T)}{\rho}  -4\cdot  \alpha \cdot k \cdot \eta \cdot \sqrt{\D{T}{\bdel} -\terment(T)}\\
			&\geq \frac{\alpha \cdot k -\bmax}{\rho}  -4\cdot  \alpha \cdot k \cdot \eta \cdot \sqrt{\D{T}{\bdel} -\terment(T)}, \\
			&= k \cdot \left( \frac{\alpha}{\rho} - \frac{\bmax}{\rho \cdot k} -4\cdot \alpha \cdot \eta   \cdot \sqrt{\D{T}{\bdel} -\terment(T)}\right) ,
		\end{aligned}
		$$
		where the second inequality is by Claim \ref{claim:betaT}.
	\end{claimproof}
	
	The lower bound for $n\cdot \kappa(T)$ in Claim \ref{claim:kappaT} depends on $  \sqrt{\D{T}{\bdel} -\terment(T)}$.   However, inequality~\eqref{eq:piegon_T} together with Lemma~\ref{lem:all_type_prob} imply that $\sqrt{\D{T}{\bdel} -\terment(T)}$ is arbitrarily small, depending on the selection of $M$.   Formally, by Lemma~\ref{lem:all_type_prob} we have
	$$
	\exp\bigg(- n\cdot \big( \D{T}{\bdel} +\terment(T)\big)\bigg)\, \geq\, \Pr_{S^*}\left(\Pr(A_1,\ldots ,A_n ) = T\right) \,\geq \, \left(\frac{1}{2\cdot \alpha \cdot k }\right)^{|\chi|+1},
	$$
	where the last inequality is by \eqref{eq:piegon_T}. Therefore,
	$$
	\D{T}{\bdel} +\terment(T) \leq \frac{|\chi|+1}{n}\cdot \ln\left( 2\cdot \alpha \cdot k\right) \,\leq \, \frac{ 4\cdot |\chi| \cdot \bmax\cdot \ln(2\cdot \alpha \cdot k)}{\alpha \cdot k} .
	$$
	The last inequality uses $n\geq t = \floor{\frac{\alpha \cdot k }{\bmax}}  > \frac{\alpha \cdot k}{2\cdot\bmax}  $, assuming $k>M > 2\cdot \bmax\cdot \alpha^{-1}$.
	We select $M >  \max\left\{2\cdot \bmax,\frac{2}{\alpha}\right\}$ such that 
	$$ \frac{\bmax}{\rho \cdot \ell }+4\cdot \alpha \cdot \eta\cdot \sqrt{\frac{ 4\cdot |\chi| \cdot \bmax\cdot \ln(2\cdot \alpha \cdot \ell)}{\alpha \cdot \ell}}  < \frac{\alpha}{\rho} -1$$ for every $\ell > M$. This is possible as the left-hand term in the above inequality converges to $0$ as $\ell$ goes to infinity, and the right-hand term is positive since $\rho<\alpha$.  Therefore, 
	$$
	\frac{\bmax}{\rho \cdot k} +4\cdot \alpha \cdot \eta\cdot  \sqrt{\D{T}{\bdel} +\terment(T)}  \leq  \frac{\bmax}{\rho \cdot k}+ 4\cdot \alpha \cdot \eta \cdot \sqrt{ \frac{ 4\cdot |\chi| \cdot \bmax\cdot \ln(2\cdot \alpha \cdot k)}{\alpha \cdot k}} < \frac{\alpha}{\rho }-1,
	$$
	and by Claim~\ref{claim:kappaT}, we have
	\begin{equation}
		\label{eq:kappaT}
		n\cdot \kappa(T) \geq k\cdot \left(  \frac{\alpha}{\rho}  - \frac{\bmax}{\rho} \cdot \frac{1}{k} -4\cdot  \alpha  \cdot \eta \cdot \sqrt{\D{T}{\bdel} -\terment(T)}\right)  >  k\cdot \left( \frac{\alpha}{\rho} -\frac{\alpha}{\rho}+1\right) = k.
	\end{equation}
	
	By Claim~\ref{claim:betaT} and \eqref{eq:kappaT}, we have 
	$$
	\begin{aligned}
		G^{\floor{\alpha \cdot k, k}} 
		&=\left\{\exists n': X_{n'}\leq \floor{\alpha \cdot k} \textnormal{ and } Y_{n} \geq k   \right\} \\
		&\supseteq \left\{ X_{n}\leq \floor{\alpha \cdot k} \textnormal{ and } Y_{n} \geq k   \right\}  \\
		&= \left\{ n\cdot \beta(\type(A_1,\ldots, A_n) ) \leq \floor{\alpha k} \textnormal{ and } n\cdot \kappa(\type(A_1,\ldots, A_n)) \geq k \right\}\\
		&\supseteq \left\{\type(A_1,\ldots, A_n) = T\right\}.
	\end{aligned}
	$$
	Therefore,
	$$p(\floor{\alpha\cdot k},k)= \min_{S\in \mS} \Pr_S(G^{\floor{\alpha\cdot k},k})  = \Pr_{S^*}(G^{\alpha\cdot k,k}) \geq \Pr_{S^*}(\type(A_1,\ldots,A_n) =T) \geq \left(\frac{1}{2\cdot \alpha \cdot k }\right)^{|\chi|+1}
	$$
	for every $k>M$.
	By the above inequality, 
	\begin{equation*}
		\label{eq:liminf}
		\liminf_{k\rightarrow \infty} \frac{1}{k} \cdot \ln p(\floor{\alpha \cdot k }, k) \geq  \liminf_{k\rightarrow \infty} \frac{1}{k} \cdot \ln \left(\left( \frac{1}{2\cdot \alpha \cdot k}\right)^{|\chi| +1}\right) = 0 
	\end{equation*}
	which completes the proof. 
\end{proof}

\subsection{Changing Probability Space}
\label{sec:translation}
Our next step is to prove Lemma~\ref{lem:translation}. We first restate the lemma.
\translation*
The lemma considers two recurrences which differ in their probability vectors: $\bgam^j$ vs. $\bdel^j$. We consider the two recurrences through the lens of the random walks, meaning we have two random walks to consider: the first is the one associated with $p_{\gamma}$ and the second is associated with $p_{\delta}$.  Denote by $(\Omega,\mF, \Pr_{\gamma, S})$ the probability space associated with the strategy $S\in \mS$ and the terms $(\bb^j,\bk^j,\bgam^j)$ for $j\in [N]$, and by  $(\Omega,\mF, \Pr_{\delta, S})$ the probability space associated with the strategy $S\in \mS$ and the terms $(\bb^j,\bk^j,\bdel^j)$ for $j\in [N]$. Observe that random variables such as $X_n$ and $A_n$ are defined in both probability spaces. 

For a fixed $k$, the proof of Lemma~\ref{lem:translation} first focuses on a strategy $S^*\in \mS$  such that 
$$p_{\gamma }(\floor{\alpha \cdot k},k)= \Pr_{\gamma, S^*} (G^{\floor{\alpha\cdot k},k}).$$ $S^*$ exists by Lemma~\ref{lem:comp_to_walk}.  Furthermore, the lemma implies that 
$$
p_{\delta}(\floor{\alpha \cdot k},k) = \min_{S\in \mS} \Pr_{\delta, S} \left( G^{\floor{\alpha \cdot k},k}\right)  \leq  \Pr_{\delta, S^*} \left( G^{\floor{\alpha \cdot k},k}\right).
$$
The proof then uses the pigeon hole principle to find $n>0$ and a type $T\in \mK_n$ such that $\{\type(A_1,\ldots ,A_n) =T\} \subseteq G^{\floor{\alpha \cdot k },k} $  and 
\begin{equation}
	\label{eq:translation_type_prob}\Pr_{\delta, S^*} \left(\type(A_1,\ldots, A_n)=T\right)  \geq \frac{p_{\delta}(\floor{\alpha \cdot k},k) }{\poly(k)}.
\end{equation}
The  main idea in the proof  is to evaluate the probability of the event $\{\type(A_1,\ldots ,A_n) =T\}$ in the probability space $(\Omega,\mF,\Pr_{\gamma, S^*})$ associated with the composite recurrence $p_{\gamma}$. Specifically, by Lemmas~\ref{lem:inconsistent_prob} and~\ref{lem:type_prob}, it can be shown that 
$$
\Pr_{\gamma,S^*}(\type(A_1,\ldots,A_n) = T) = \Pr_{\delta,S^*}(\type(A_1,\ldots, A_n ) =T)\cdot \exp\left( -n\left(\D{T}{\bgam} -\D{T}{\bdel}\right)\right).
$$
By \eqref{eq:translation_type_prob} and since $\lim_{k\rightarrow \infty} \frac{1}{k}\ln  p_{\delta}(\floor{\alpha\cdot k},k) =0$, the above probability is dominated by the expression  $\exp\left( -n\left(\D{T}{\bgam} -\D{T}{\bdel}\right)\right)$. 
To complete the proof, we show that $  -n\left(\D{T}{\bgam} -\D{T}{\bdel}\right)\gtrsim  -k\cdot \max_{j\in [N]} \frac{\D{\bdel^j}{\bgam^j}}{\delta^j\cdot \bk^j}$, which follows from the fact that the type $T$ and the vectors $\bdel^j$ for $j\in [N]$ must be close.

\begin{proof}[Proof of Lemma~\ref{lem:translation}]
	
	Define $M = \max_{j\in [N]} \frac{\D{\bdel^j }{\bgam^j} }{ \delta^j\cdot \bk^j}$, $\kmax=\max_{(j,i)\in \chi}\bk^j_i$ and $\eta = \max_{(j,i)\in \chi }\left| \ln\frac{\bdel^j_i}{\bgam^{j}_i}\right|$. Let $\eps>0$ and  define 
	\begin{equation}
		\label{eq:translate_Z_def}
		Z= \min\left\{ \frac{\eps^2}{2^{10}\cdot \eta^2\cdot \alpha^2\cdot \kmax}, ~\frac{\eps^2}{2^{10}\cdot M^2\cdot\alpha^2\cdot  \kmax^3},\frac{\eps}{8}\right\}.
	\end{equation}
	Select $K>\max\left\{2\cdot \kmax,\frac{4\cdot \kmax \cdot M}{\eps},\frac{1}{\alpha}\right\}$ such that for every $k > K$ it holds that 
	$$
	\frac{1}{k} \cdot \ln p_{\delta} \left( \floor{\alpha \cdot k},k\right) \geq -Z \iff p_{\delta}(\floor{\alpha \cdot k } , k) \geq \exp \left( - k\cdot Z \right) ,
	$$
	and 
	\begin{equation}
		\label{eq:Z_cond}
		\frac{1}{k}\cdot  (|\chi|+1) \cdot \ln (\alpha \cdot k +1) < Z.
	\end{equation}
	Such $K$ exists by the conditions of the lemma. The selection of $K$ and $Z$ will be made clearer later in the proof. Let $k>K$.  
	
	By Lemma~\ref{lem:comp_to_walk}, it holds that 
	$
	p_{\gamma} (\floor{\alpha\cdot k},k) = \min_{S\in \mS} \Pr_{\gamma, S}(G^{\floor{\alpha\cdot k },k })
	$; then, there exists a strategy $S^*\in \mS$ such that $p_{\gamma} (\floor{\alpha \cdot k } ,k) = \Pr_{\gamma, S^*} (G^{\floor{\alpha\cdot k},k})$.  Using Lemma~\ref{lem:comp_to_walk} w.r.t. $p_{\delta}$, we also have
	$$
	\exp\left( -k \cdot Z \right) \leq p_{\delta} (\floor{\alpha \cdot k} , k) = \min_{S\in \mS} \Pr_{\delta, S} (G^{\floor{\alpha\cdot k},k}) \leq \Pr_{\delta, S^*} (G^{\floor{\alpha\cdot k},k}).
	$$
	We can further expand the terms in the above inequality and get
	$$
	\begin{aligned}
		\exp\left( -k\cdot Z\right) &\leq \Pr_{\delta, S^*} \left( G^{\floor{\alpha \cdot k},k}\right) \\
		&= \Pr_{\delta, S^*} \left(\exists  n\in \mathbb{N}: X_n \leq \alpha \cdot k \textnormal{ and } Y_n \geq k   \right)\\
		&= \sum_{n=\floor{ k/\kmax } }^{\ceil{\alpha \cdot k}}
		\Pr_{\delta, S^*} \left(X_n \leq \alpha \cdot k \textnormal{ and } Y_n \geq k   \textnormal{ and } Y_{n-1} < k \right),
	\end{aligned}
	$$
	where the last equality holds as $X_{\ceil{\alpha k}+1} = \beta(A_1,\ldots, A_{\ceil{\alpha k} +1}) \geq  \alpha k +1$, and $$Y_{\floor{k/\kmax} -1} = \sum_{\ell = 1}^{\floor{k/\kmax} -1} \kappa(A_\ell )\leq \sum_{\ell = 1}^{\floor{k/\kmax} -1} \kmax <k.$$ 
	By the pigeon hole principle, there is $\floor{k/\kmax}\leq n \leq \ceil{\alpha \cdot k}$ such that 
	$$ 
	\Pr_{\delta, S^*} \left(X_n \leq \alpha \cdot k \textnormal{ and } Y_n \geq k   \textnormal{ and } Y_{n-1} < k \right) \geq \frac{\exp\left( -k\cdot Z \right)}{ \alpha k +1 }  = \exp\left( -k \cdot Z - \ln(\alpha \cdot k +1) \right).
	$$
	Since  $\type(A_1,\ldots ,A_n)\in \mK_n$  (by Observation~\ref{obs:types}), it follows that
	$$
	\begin{aligned}
		&\sum_{T\in \mK_n} \Pr_{\delta, S^*} \left(X_n \leq \alpha \cdot k \textnormal{ and } Y_n \geq k   \textnormal{ and } Y_{n-1} < k  \textnormal{ and } \type(A_1,\ldots ,A_n) = T\right)  \\
		=& \Pr_{\delta, S^*} \left(X_n \leq \alpha \cdot k \textnormal{ and } Y_n \geq k   \textnormal{ and } Y_{n-1} < k \right)  \geq \exp\left( -k \cdot Z - \ln(\alpha \cdot k +1) \right).
	\end{aligned}
	$$
	Recall that  $|\mK_n| = (n+1)^{|\chi|} \leq (\alpha \cdot k +1)^{|\chi|}$ (by Observation \ref{obs:types}); therefore, by the pigeon hole principle, there is $T\in \mK_n$  such that 
	\begin{equation}
		\label{eq:Tselect}
		\begin{aligned}
			&\Pr_{\delta, S^*} \left(X_n \leq \alpha \cdot k \textnormal{ and } Y_n \geq k   \textnormal{ and } Y_{n-1} < k  \textnormal{ and } \type(A_1,\ldots ,A_n) = T\right)  \\
			\geq&\frac{ \exp\left( -k \cdot Z - \ln(\alpha \cdot k +1) \right)}{ (\alpha \cdot k +1)^{|\chi|}} =  \exp\left( -k \cdot Z- (|\chi|+1)\ln(\alpha \cdot k +1) \right)\geq \exp \left( -2\cdot  k\cdot Z\right) .
		\end{aligned}
	\end{equation}
	The last inequality follows from \eqref{eq:Z_cond}. 
	Since the event in \eqref{eq:Tselect} has a positive probability, it holds that $ n\cdot \beta(T )\leq \alpha \cdot k$ and $k\leq n\cdot \kappa(T) \leq k+\kmax$. Therefore, 
	\begin{equation}
		\label{eq:pgam_to_prop}
		p_{\gamma}(\floor{\alpha \cdot k}, k) = \Pr_{\gamma,S^*}(G^{\floor{\alpha \cdot k},k} ) \geq \Pr_{\gamma,S^*}(\type(A_1,\ldots, A_n ) =T).
	\end{equation}
	
	Let
	$$C = \{ (a_1,\ldots, a_n)\in \chi^n ~|~\type(a_1,\ldots, a_n) =T \textnormal{ and } (a_1,\ldots, a_n) \textnormal{ is consistent with $S^*$} \}$$
	be the set of all strings of length $n$ of type $T$ which are consistent with $S^*$ (Definition~\ref{def:consistent}). By Lemmas~\ref{lem:inconsistent_prob} and~\ref{lem:type_prob}, it holds that
	$$
	\Pr_{\delta, S^*}(\type(A_1,\ldots,A_n) =T) = |C| \cdot \exp\left(-n\cdot \left( \entropy(T) + \D{T}{\bdel}\right) \right) 
	$$
	and 
	$$
	\Pr_{\gamma, S^*}(\type(A_1,\ldots,A_n) =T) = |C| \cdot \exp\left(-n\cdot \left( \entropy(T) + \D{T}{\bgam}\right) \right).
	$$
	Therefore, 
	\begin{equation}
		\label{eq:prob_gamma}
		\begin{aligned}
			p_\gamma(\floor{\alpha\cdot k},k) &\geq 
			\Pr_{\gamma, S^*}(\type(A_1,\ldots,A_n) =T)\\ &= \frac{\Pr_{\delta, S^*}(\type(A_1,\ldots,A_n) =T)}{   \exp\left(-n\cdot \left( \entropy(T) + \D{T}{\bdel}\right) \right) } \cdot  \exp\left(-n\cdot \left( \entropy(T) + \D{T}{\bgam}\right) \right) \\
			&\geq \exp\left(-2\cdot Z \cdot k  -n\left( \D{T}{\bgam} - \D{T}{\bdel}\right)  \right)\\
			&\geq \exp\left( -n\left( \D{T}{\bgam} - \D{T}{\bdel}\right) - \frac{\eps}{4}\cdot k \right),
		\end{aligned}
	\end{equation}
	where the first inequality follows from \eqref{eq:pgam_to_prop},  the second  inequality follows from \eqref{eq:Tselect}, and the last inequality holds as $Z\leq \frac{\eps}{8}$ (by~\eqref{eq:translate_Z_def}). In \eqref{eq:prob_gamma} the probability of an event in the probability space associated with $p_{\delta}$ is used to lower bound the probability of an event in the probability space associated with $p_{\gamma}$. This transition is the core of the proofs of Lemma~\ref{lem:translation} and Theorem~\ref{thm:rec}.

	We use the next claim to bound the last term in \eqref{eq:prob_gamma} (recall that $M = \max_{j\in [N]} \frac{\D{\bdel^j }{\bgam^j} }{ \delta^j\cdot \bk^j}$).
	\begin{claim}
		\label{claim:divs}
		$	-n\left( \D{T}{\bgam} - \D{T}{\bdel}\right) \geq  - M \cdot k -\eps\cdot\frac{3}{4}\cdot k$
	\end{claim}
	Before we prove Claim~\ref{claim:divs}, we show how it can be used to complete the proof of the lemma. 
	By  \eqref{eq:prob_gamma} and Claim~\ref{claim:divs}, we have
	$$
	\begin{aligned}
		p_{\gamma} (\floor{\alpha \cdot k} , k )
		&\geq \exp\left( -n\left( \D{T}{\bgam} - \D{T}{\bdel}\right) -\frac{\eps}{4}\cdot k  \right)\\
		&\geq \exp\left(- M \cdot k -\eps\cdot \frac{3}{4} - \eps\cdot  \frac{1}{4}\cdot k  \right)\\
		&= \exp\left( -(M+\eps )\cdot k \right).
	\end{aligned}
	$$
	Therefore,
	$$
	\frac{1}{k} \cdot p_{\gamma} (\floor{\alpha \cdot k}, k )  \geq\frac{1}{k}\cdot \ln \exp\left(-k(M-\eps)\right) \geq -M-\eps,
	$$
	where the first inequality follows from \eqref{eq:pgam_to_prop}. 
	This implies that $$\liminf_{k\rightarrow \infty} \frac{1}{k}\cdot \ln p_{\gamma} (\floor{\alpha\cdot k},k) \geq  -M= -\max_{j\in [N]} \frac{ \D{\delta^j}{\gamma^j}}{\delta^j\cdot \bk^j},$$ as required.

	We use the next claim in the proof of Claim~\ref{claim:divs}, which basically states that the type $T$ must be close to the vectors $\bdel^j$, due to \eqref{eq:Tselect}. Define $\lambda_j= \sum_{i=1}^{r_j} T_{(j,i)}$. 
	\begin{claim}
		\label{claim:TtoDel}
		$\sum_{j=1}^{N} \sum_{i=1}^{r_j}\left| T_{(j,i)}-\lambda_j\cdot \bdel^j_i\right|\leq 4 \cdot\sqrt {\kmax \cdot Z}$ 
	\end{claim}
	\begin{claimproof}
		By \eqref{eq:Tselect} and Lemma~\ref{lem:all_type_prob}, it holds that 
		$$
		\exp(-2\cdot k\cdot Z) \leq \Pr_{\delta, S^*}(\type(A_1,\ldots, A_n )=T)  \leq \exp\left(-n\cdot \left(\D{T}{\bdel} +\terment(T)\right)\right).
		$$
		Therefore, 
		$$
		\D{T}{\bdel} +\terment(T)\leq 2\cdot \frac{k}{n}\cdot Z \leq 4\cdot \kmax \cdot Z,
		$$
		where the last inequality holds as $n\geq k/\kmax-1\geq \frac{k}{2\cdot \kmax}$. 
		By Lemma~\ref{lem:norm}, we have
		$$
		\sum_{j=1}^{N}\sum_{i=1}^{r_j} \left|T_{(j,i)} - \lambda_j\cdot \bdel^j_i \right| \leq 2\cdot\sqrt{\D{T}{\bdel} +\terment(T) } \leq 4 \cdot\sqrt {\kmax \cdot Z}.
		$$
	\end{claimproof}

	\begin{claimproof}[Proof of Claim~\ref{claim:divs}]
		By the definition of $\D{}{}$ \eqref{eq:div_type}, we have
		\begin{equation}
			\label{eq:divs_first}
			\begin{aligned}
				-n\left( \D{T}{\bgam} - \D{T}{\bdel}\right) &= 
				-n \sum_{j=1}^{N}\sum_{i=1}^{r_j} T_{(j,i)}\cdot \ln \frac{T_{(j,i)}}{\bgam^j_i} + n \sum_{j=1}^{N}\sum_{i=1}^{r_j} T_{(j,i)}\cdot \ln \frac{T_{(j,i)}}{\bdel^j_i} \\
				&= -n \sum_{j=1}^{N}\sum_{i=1}^{r_j} T_{(j,i)}\cdot \ln \frac{\delta^j_i}{\bgam^j_i} \\
				&= -n \sum_{j=1}^{N} \lambda_j \sum_{i=1}^{r_j} \bdel^j_i \cdot \ln \frac{\bdel^j_i}{\bgam^j_i} -n\sum_{j=1}^{N} \sum_{i=1}^{r_j} \left( T_{(j,i)} - \lambda_j\cdot \bdel^j_i\right)\cdot \ln \frac{\bdel^j_i}{\bgam^j_i}\\
				& \geq -n\sum_{j=1}^{N} \lambda_j \cdot \D{\bdel^j}{\bgam^j} -n\sum_{j=1}^{N} \sum_{i=1}^{r_j} \left|T_{(j,i)} - \lambda_j\cdot \bdel^j_i\right| \cdot \eta \\
				&\geq  -n\sum_{j=1}^{N} \lambda_j \cdot \D{\bdel^j}{\bgam^j} - 2\cdot\alpha \cdot k\cdot \eta \cdot4\cdot  \sqrt{\kmax \cdot Z}\\
				&\geq -n\sum_{j=1}^{N} \lambda_j \cdot \D{\bdel^j}{\bgam^j} - k\cdot \frac{\eps}{4},
			\end{aligned}
		\end{equation}
		where the first inequality holds as  $\eta = \max_{(j,i)\in \chi }\left| \ln\frac{\bdel^j_i}{\bgam^{j}_i}\right|$,  the second inequality follows from Claim~\ref{claim:TtoDel} and since $n\leq \alpha\cdot k+1\leq 2\cdot \alpha \cdot k$, and the last inequality holds as $Z\leq \frac{\eps^2}{2^{10} \cdot \alpha^2\cdot \eta^2\cdot \kmax}$ \eqref{eq:translate_Z_def}. 
		
		Recall that $M = \max_{j\in [N]} \frac{\D{\bdel^j }{\bgam^j} }{ \delta^j\cdot \bk^j}$. Then,
		\begin{equation}
			\label{eq:divs_second}
			\begin{aligned}
				-n\sum_{j=1}^{N} \lambda_j \cdot \D{\bdel^j}{\bgam^j}  &= 	  -n\sum_{j=1}^{N} \lambda_j \cdot \frac{\bdel^j\cdot \bk^j}{\bdel^j\cdot \bk^j}\cdot \D{\bdel^j}{\bgam^j} \\
				&\geq -n\cdot \sum_{j=1}^N  \lambda_j\cdot \bdel^j\cdot \bk^j \cdot M \\
				&= -n\cdot  M\cdot  \sum_{j=1}^N   \sum_{i=1}^{r_j} T_{(j,i)} \cdot \bk^j_i -n \cdot M \cdot \sum_{j=1}^{r_j } \left( \lambda_j\cdot \bdel^j_i -T_{(j,i)}\right)\cdot  \bk^j_i \\
				&\geq  - M \cdot n\cdot \kappa(T) -n \cdot M \cdot \sum_{j=1}^{r_j } \left|\lambda_j\cdot \bdel^j_i -T_{(j, i)}\right|\cdot  \kmax\\
				&\geq  -M \cdot k - M\cdot\kmax   - \kmax\cdot M \cdot 2\cdot \alpha \cdot k \cdot 4\cdot \sqrt{\kmax \cdot Z}\\
				&\geq -M \cdot k -\frac{\eps}{4} \cdot k -\frac{\eps}{4}\cdot k.
			\end{aligned}
		\end{equation}
		The third inequality holds as  $n\cdot \kappa(T)\leq k + \kmax$ and $n\leq \alpha \cdot k+1\leq 2\cdot \alpha \cdot k$. The last inequality holds as $k>K>\frac{4\cdot M\cdot \kmax}{\eps}$ and  $Z\leq \frac{\eps^2}{2^{10}\cdot M^2\cdot\alpha^2\cdot  \kmax^3}$ \eqref{eq:translate_Z_def}. 
		By \eqref{eq:divs_first} and \eqref{eq:divs_second}, we have
		$$
		-n\left(\D{T}{\bgam} +\D{T}{\bdel}\right) \geq - M \cdot k -\eps\cdot\frac{3}{4}\cdot k.
		$$
	\end{claimproof}

\end{proof}

\subsection{The Lower Bound}
\label{sec:proof_lower}

Next, we prove Lemma~\ref{lem:rec_lower} that we now restate.
\reclower*
The proof of Lemma~\ref{lem:rec_lower} follows from a simple application of Lemmas~\ref{lem:trivial_rec} and~\ref{lem:translation}. 
\begin{proof}[Proof of Lemma~\ref{lem:rec_lower}]

For every $j\in [N]$ let $\bdel^j\in\mathbb{R}_{\geq 0}^{r_j}$ be a distribution such that $M_j = \frac{\D{\bdel^j}{\bgam^j}}{\bdel^j\cdot \bk^j}$ and $\bdel^j\cdot \bb^j \leq \alpha \cdot \bdel^j\cdot \bk^j$. 
Such distributions exist by the definitions of $\alpha$-branching numbers (Definition~\ref{def:alpha_branching}). 
Ideally, we would like to use Lemma~\ref{lem:trivial_rec} with respect to the composite recurrence $p_{\delta}$ of the terms $\{(\bb^j,\bk^j,\bdel^j)~|~j\in [N]\}$. However, the recurrence  $p_{\delta}$ does not satisfy the conditions of the Lemma~\ref{lem:trivial_rec}. It may be that $\bdel^j\cdot \bb^j =\alpha \cdot \bdel^j\cdot \bk^j$ while the lemma requires that $\bdel^j\cdot \bb^j <\alpha \cdot \bdel^j\cdot \bk^j$. Furthermore, $p_{\delta}$ may not be  well defined, as we require for a term $(\bb,\bk,\bgam)$ of a composite recurrence that $\bgam_i>0$ for every $i$. We use the next claim to overcome these technical obstacles. 
\begin{claim}
	\label{claim:calculus}
For every $j\in [N]$, there is a sequence  of distributions $\bdel^{j,\ell}\in \mathbb{R}^{r_j}_{+}$ for  $\ell\in \mathbb{N}$ such that:
\begin{enumerate}
	\item For every $\ell\in \mathbb{N}$  and $i\in [r_j]$ it holds that $ \bdel^{j,\ell}_i >0$.
	\item For every $\ell \in \mathbb{N}$ it holds that $\bdel^{j,\ell} \cdot \bb^j <  \alpha \cdot \bdel^{j,\ell} \cdot \bk^j$.
	\item For every $i\in [r_j]$ it holds that $\lim_{\ell \rightarrow \infty} \bdel^{j,\ell}_i =\bdel^{j}_i$. 
\end{enumerate}
\end{claim}
The proof of Claim~\ref{claim:calculus} uses simple calculus arguments and the fact that $\alpha$ is strictly greater than the critical ratio of each of the terms $(\bb^j,\bk^j,\bgam^j)$. We prove Claim~\ref{claim:calculus} below.

Let $\bdel^{j,\ell}$ be the vectors defined in Claim~\ref{claim:calculus} for every $j\in [N]$ and $\ell \in \mathbb{N}$. Let $p_{\delta,\ell}$ be the composite recurrence of $\{(\bb^j,\bk^j,\bdel^{j,\ell} ~|~j\in [N])\}$ for every $\ell\in \mathbb{N}$.  Since $\bdel^{j,\ell} \cdot \bb^j <  \alpha \cdot \bdel^{j,\ell} \cdot \bk^j$ for every $j\in \mathbb{N}$ and $\ell \in \mathbb{N}$, by Lemma~\ref{lem:trivial_rec} we have 
$$
\forall \ell\in\mathbb{N}:~~~~\lim_{k\rightarrow \infty} \frac{1}{k} \cdot \ln p_{\delta,\ell}(\floor{\alpha \cdot k},k) = 0. 
$$

Recall that $p$ is the composite recurrence of $\{(\bb^j,\bk^j,\bgam^j)~|~j\in [N]\}$.
Then, by Lemma~\ref{lem:translation},
$$
\forall \ell\in\mathbb{N}:~~~~\liminf_{k\rightarrow \infty} \frac{1}{k} \cdot \ln p(\floor{\alpha \cdot k},k) \geq - \max_{j\in [N]} \frac{\D{\bdel^{j,\ell}}{\gamma}}{\bdel^{j,\ell}\cdot \bk^j}. 
$$
Therefore,
$$
\begin{aligned}
\liminf_{k\rightarrow \infty} \frac{1}{k} \cdot \ln p(\floor{\alpha \cdot k},k) &=  \liminf_{\ell \rightarrow \infty } \liminf_{k\rightarrow \infty} \frac{1}{k} \cdot \ln p(\floor{\alpha \cdot k},k)\\
&\geq   \liminf_{\ell \rightarrow \infty } \left( - \max_{j\in [N]} \frac{\D{\bdel^{j,\ell}}{\gamma}}{\bdel^{j,\ell}\cdot \bk^j}\right) \\
&=  - \max_{j\in [N]} \frac{\D{\bdel^{j}}{\gamma}}{\bdel^{j}\cdot \bk^j}\\ &= -\max_{j\in [N]} M_j \\
&= -M.
\end{aligned}
$$
The second equality holds as $\lim_{\ell\rightarrow \infty} \bdel^{j,\ell}_i = \bdel^j_i$ for every $j\in [N]$ and $i\in [r_j]$. 

It remains to prove Claim~\ref{claim:calculus}.
\begin{claimproof}[Proof of Claim~\ref{claim:calculus}]
	Fix arbitrary  $j\in [N]$. By the definition of critical ratio (Definition~\ref{def:critical_ratio}) 
	there exists ${i^*}\in [r_j]$ such that $\critical(\bb^j,\bk^j,\bgam^j) = \frac{\bb^j_{i^*}}{\bk^j_{i^*}}$. Define $\bq\in [0,1]^{r_j}$ by $\bq_{i^*} = 1$ and $\bq_{i}=0$ for all $i\in [r_j]\setminus \{i^*\}$. Then, 
	\begin{equation}\label{eq:calculus_first}
	\bq \cdot \bb^j = \bb^j_{i^*}= \critical(\bb^j,\bk^j,\bgam^j) \cdot  \bk^j_{i^*} = \critical(\bb^j,\bk^j,\bgam^j) \cdot \bq\cdot  \bk^j  <\alpha \cdot \bq \cdot \bk^j.
	\end{equation}
	Observe that $\bq\cdot \bb^j > 0$ (since $\bb^j\in \mathbb{N}_{>0}^{r_j}$), therefore $\bq \cdot \bk^j>0$, which together with $\alpha>\critical(\bb^j,\bk^j,\bgam^j)$ justifies the strict inequality. By rearranging \eqref{eq:calculus_first} we have 
	\begin{equation}\label{eq:calculus_second}
	\bq \cdot \left( \bb^j-\alpha \cdot \bk^j\right)< 0.
	\end{equation}
	
	Also, let $\bt\in \mathbb{R}_{>0}^{r_j}$ be the distribution defined by $\bt_i = \frac{1}{r_j}$ for every $i\in [r_j]$.   By \eqref{eq:calculus_second} there is $a\in (0,1)$ such that 
	$$
	 a\cdot \bq \cdot \left( \bb^j-\alpha \cdot \bk^j\right) + (1-a)\cdot \bt \cdot \left( \bb^j-\alpha \cdot \bk^j\right)<0.
	 $$
	 The above inequality is equivalent to
	 \begin{equation}
	 	\label{eq:positive_convex}
	 a\cdot \bq \cdot \bb^j + (1-a) \cdot \bt \cdot \bb^j < \alpha \cdot a \cdot \bq \cdot \bk^j + \alpha \cdot (1-a)\cdot \bt \cdot \bk^j.
	 \end{equation}
	 
	Define $$\bdel^{j,\ell} = \left( 1-\frac{1}{\ell}\right) \cdot \bdel^j + \frac{a}{\ell} \cdot \bq + \frac{1-a}{\ell} \cdot \bt$$
	for every $\ell\in \mathbb{N}$. 
	It holds that $\delta^{j,\ell}$ is a distribution as it is a convex combination of distributions. 
	It remains to show  $\bdel^{j,\ell}$ satisfies the properties in the claim. 
	
	For every $\ell \in [\mathbb{N}]$ and $i\in [r_j]$ it holds that 
	$$
	\bdel^{j,\ell}_i = \left( 1-\frac{1}{\ell}\right) \cdot \bdel_i^j + \frac{a}{\ell} \cdot \bq_i + \frac{1-a}{\ell} \cdot \bt_i\geq \frac{1-a}{\ell} \cdot \bt_i >0.
	$$
	
	For every $\ell\in \mathbb{N}$ it holds that 
	$$
	\begin{aligned}
	\delta^{j,\ell}\cdot \bb^j &=\left( 1-\frac{1}{\ell}\right) \cdot \bdel^j\cdot \bb^j  + \frac{a}{\ell} \cdot \bq \cdot \bb^j+ \frac{1-a}{\ell} \cdot \bt\cdot \bb^j  \\
	&< 
	\left( 1-\frac{1}{\ell}\right) \cdot \alpha \cdot  \bdel^j\cdot \bk^j  + \frac{a}{\ell} \cdot \alpha \cdot \bq \cdot \bk^j+ \frac{1-a}{\ell}\cdot \alpha \cdot \bt\cdot \bk^j  \\
	&=  \alpha \cdot \bdel^{j,\ell} \cdot \bk^j,
	\end{aligned}
	$$
	where the inequality follows from \eqref{eq:positive_convex} and since $\delta^j\cdot \bb^j\leq \alpha \cdot\delta^j\cdot \bk^j$ by the definition of $\delta^j$. 
	
	Finally, for every $i\in [r_j]$ it holds that
	$$
	\lim_{\ell\rightarrow \infty } \delta^{j,\ell}_i = \lim_{\ell\rightarrow \infty} \left(  \left( 1-\frac{1}{\ell}\right) \cdot \bdel_i^j + \frac{a}{\ell} \cdot \bq_i + \frac{1-a}{\ell} \cdot \bt_i\right) = \bdel^{j}_i. 
	$$
\end{claimproof}

\end{proof}

\subsection{The Upper Bound}
\label{sec:proof_upper}

The final ingredient in the proof of Theorem~\ref{thm:rec} is the missing proof of Lemma~\ref{lem:rec_upper}.
\recupper*
The proof of Lemma~\ref{lem:rec_upper} uses the random walk associated with the recurrence $p$ using a specific strategy $S^*\in \mS$ which always selects a term $j^* \in [N]$ for which $M_{j^*}= \max_{j\in [N]} M_{j}$.  By Lemma~\ref{lem:comp_to_walk}, $p(\floor{\alpha \cdot k},k)\leq \Pr_{S^*}(G^{\alpha ,k})$. The proof focuses on a specific type $T$ and length $n$ such that $ \Pr_{S^*}(G^{\alpha ,k})\approx \Pr_{S^*} (\Pr(A_1,\ldots, A_n)=T)$. The probability of the last event is upper bounded using Lemma~\ref{lem:all_type_prob}, and the properties $T$ are used to show that this upper bound is at most $\exp(-k\cdot M)$. 

\begin{proof}[Proof of Lemma~\ref{lem:rec_upper}]
The proof considers the random walk associated with the composite recurrence $p$ of $\{(\bb^j,\bk^j,\bgam^j)~|~j\in[N]\}$ as defined in Section~\ref{sec:random_walk}, and uses the notation defining the random walk. This includes the random variables $X_n$, $Y_n$ and $A_n$ for any $n\in \mathbb{N}$, the set of strategies $\mS$ and the measure function $\Pr_S$ for the random walk when the adversary is $S$.  We also use the notation for types as given in Section~\ref{sec:types}. 
	
Fix arbitrary $j^*\in [N]$ such that $M_{j^*}= \max_{j\in [N]} M_j$, and define a strategy $S^*\in \mS$ by $S^*(a)=j^*$ for every $a\in \chi^*$. 
Let $k\in \mathbb{N}$ and assume $p(\floor{\alpha \cdot k }, k) > 0$. 
By Lemma~\ref{lem:comp_to_walk}, it holds that 
\begin{equation}
	\label{eq:upper_initial}
0<p(\floor{\alpha\cdot k},k) = \min_{S\in \mS} \Pr_{S}(G^{\floor{\alpha \cdot k},k} ) \leq \Pr_{S^*}(G^{\floor{\alpha \cdot k},k} )
\end{equation}
for every $k\in \mathbb{N}$. 

For every $n$ it holds that $X_n = \sum_{\ell=1}^{n} \beta(A_\ell) \geq n$ since $\beta(A_{\ell}) \geq 1$. Therefore,
$$
\begin{aligned}
0<\Pr_{S^*}(G^{\floor{\alpha \cdot k},k} ) &	\leq \Pr_{S^*} \left(\exists n: X_n\leq \alpha \cdot k  \textnormal{ and } Y_n\geq k\right)\\
& = \Pr_{S^*} \left(\exists 1\leq n\leq \floor{\alpha \cdot k}: X_n\leq \alpha \cdot k  \textnormal{ and } Y_n\geq k  \right).
\end{aligned}
$$

By the pigeon hole principle, there is $1\leq n \leq \alpha \cdot k$ such that 
$$
\begin{aligned}
\frac{\Pr_{S^*}(G^{\floor{\alpha \cdot k},k} ) }{{\alpha \cdot k}+1} \leq \Pr_{S^*}\left( X_n\leq \alpha \cdot k  \textnormal{ and } Y_n\geq k \right).
\end{aligned}
$$
Recall that the type of $A_1,\ldots,A_n$ is in $\mK_n$, and $|\mK_n|\leq (n+1)^{|\chi|}\leq (\alpha k +1)^{\chi}$  (Observation~\ref{obs:types}). Therefore,
$$
\begin{aligned}
	0<\frac{\Pr_{S^*}(G^{\floor{\alpha \cdot k},k} ) }{{\alpha \cdot k}+1} &\leq \Pr_{S^*}\left( X_n\leq \alpha \cdot k  \textnormal{ and } Y_n\geq k \right)\\
	&=\sum_{T\in \mK_n} \Pr_{S^*}\left( X_n\leq \alpha \cdot k  \textnormal{ and } Y_n\geq k \textnormal{ and } \type(A_1,\ldots, A_n) = T\right),
\end{aligned}
$$
and by the pigeon hole principle, there is $T\in \mK_n$ such that 
\begin{equation}
	\label{eq:upper_type}
0<\frac{\Pr_{S^*}(G^{\floor{\alpha \cdot k},k} ) }{(\alpha \cdot k+1)^{|\chi|+1}} \leq \Pr_{S^*}\left( X_n\leq \alpha \cdot k  \textnormal{ and } Y_n\geq k \textnormal{ and }  \type(A_1,\ldots, A_n) = T\right).
\end{equation}
By Lemma~\ref{lem:all_type_prob} it holds that 
\begin{equation}
	\label{eq:upper_to_div}
	\begin{aligned}
		\frac{\Pr_{S^*}(G^{\floor{\alpha \cdot k},k} ) }{(\alpha \cdot k+1)^{|\chi|+1}} &\leq \Pr_{S^*}\left( X_n\leq \alpha \cdot k  \textnormal{ and } Y_n\geq k \textnormal{ and }  \type(A_1,\ldots, A_n) = T\right)\\
		&\leq \exp\left(- n \cdot \left( \D{T}{\bgam} +\terment(T)\right) \right) . 
	\end{aligned}
\end{equation}
Furthermore, since $S^*$ is a constant function and the event in \eqref{eq:upper_type} has a positive probability, we can show the following claim, whose proof is given below.
\begin{claim}
	\label{claim:simple_type}
	For every $j\in [N]\setminus \{j^*\}$ and $i\in [r_j]$ it holds that $T_{(j,i)}=0$. 
\end{claim}

Define $\bt\in \mathbb{R}^{r_{j^*}}_{\geq 0}$ by $\bt_i = T_{(j^*, i)}$ for every $i\in [r_{j^*}]$.  By Claim~\ref{claim:simple_type}, we have
 \begin{equation}
 	\label{eq:t_is_dist}\sum_{i=1}^{r_{j^*}} \bt_i =\sum_{i=1}^{r_{j^*} }T_{(j^*,i)} = \sum_{j=1}^{N} \sum_{i=1}^{r_j} T_{(j,i)} = 1.
 	\end{equation}
 That is, $\bt$ is a distribution. 
Thus,
\begin{equation}
	\label{eq:upper_div_elim}\D{T}{\bgam} = \sum_{(j,i)\in \chi}  T_{(j,i)} \cdot \ln \frac{T_{(j,i)}}{\bgam^j_i}= \sum_{i=1}^{r_{j^*}}  T_{(j^*,i)} \cdot \ln \frac{T_{(j^*,i)}}{\bgam^{j^*}_i} = \sum_{i=1}^{r_{j^*}} \bt_i \cdot \ln \frac{\bt_i}{\bgam^{j^*}_i} = \D{\bt}{\bgam^{j^*}},
 \end{equation}
where the first equality  is due to \eqref{eq:div_type} and the second equality is due to  Claim~\ref{claim:simple_type}. For every $j\in [N]$ define $\lambda_j= \sum_{i=1}^{r_j} T_{(j,i)}$.  By Claim~\ref{claim:simple_type}, it holds that $\lambda_j=0$ for $j\neq j^*$ and $\lambda_{j^*}=1$. Therefore, by \eqref{eq:terment} we have
\begin{equation}
	\label{eq:terment_elim}
\terment(T) = \sum_{j=1}^{N} \lambda_j\cdot \ln \frac{1}{\lambda_j}= 0 .
\end{equation}

By \eqref{eq:upper_to_div}, \eqref{eq:upper_div_elim} and \eqref{eq:terment_elim}, we have 
\begin{equation}
	\label{eq:upper_ineq}
\begin{aligned}	
	\frac{\Pr_{S^*}(G^{\floor{\alpha \cdot k},k} ) }{(\alpha \cdot k+1)^{|\chi|+1}} &\leq \exp\left(- n \cdot \left( \D{T}{\bgam} +\terment(T)\right) \right) 
	&\leq \exp\left(  -n \cdot \D{\bt}{\bgam^{j^*}}\right) .
	\end{aligned}
\end{equation}
Recall that $X_n = n \cdot\beta(\type(A_1,\ldots, A_n))$ and $Y_n = n\cdot \kappa(\type(A_1,\ldots,A_n))$
by  \eqref{eq:kappa_by_type}.  Therefore, since the event in   \eqref{eq:upper_type} has a positive probability, it must holds that
 \begin{equation}
	\label{eq:T_type_props}
	n\cdot \beta(T) \leq \alpha \cdot k  \textnormal{ and  }k\leq n\cdot \kappa(T).
	\end{equation}
Furthermore, by Claim~\ref{claim:simple_type}, it holds that 
\begin{equation}
	\label{eq:ktimest} 
 \kappa(T) = \sum_{j=1}^{N} \sum_{i=1}^{r_j} \kappa((j,i)) \cdot T_{(j,i)} =  \sum_{i=1}^{r_{j^*}} \kappa((j^*,i))\cdot T_{(j^*,i)} = \sum_{i=1}^{r_{j^*}} \bk^{j^*}_i\cdot \bt_i = \bk^{j^*}\cdot \bt,
\end{equation}
where the first equality is by \eqref{eq:kappa_by_type_explicit}, and the third equality follows from the definition of $\bt$.  By plugging \eqref{eq:T_type_props} and \eqref{eq:ktimest} into \eqref{eq:upper_ineq} we get
\begin{equation}
	\label{eq:prob_to_t}
	\frac{\Pr_{S^*}(G^{\floor{\alpha \cdot k},k} ) }{(\alpha \cdot k+1)^{|\chi|+1}} 
\leq \exp\left(  -n \cdot \D{\bt}{\bgam^{j^*}}\right)
 \leq  \exp\left(  -\frac{k}{\kappa(T)} \cdot \D{\bt}{\bgam^{j^*}}\right)
	= \exp\left(-\frac{k}{\bk^{j^*}\cdot \bt} \cdot \D{\bt}{\bgam^{j^*}} \right).
\end{equation}

\begin{claim}
	\label{claim:t_to_number}
	$\frac{\D{\bt}{\bgam^{j^*}}}{\bk^{j^*}\cdot \bt}\geq M_{j^*}$. 
\end{claim}
\begin{claimproof}
We prove the claim by showing that $\bt$ is a feasible solution for \eqref{eq:alpha_num}, the optimization problem which defines the branching numbers, with respect to the term $(\bb^{j^*},\bk^{j^*},\bgam^{j^*})$. Similar to~\eqref{eq:ktimest}, by Claim~\ref{claim:simple_type}, it holds that
	$$
	\beta(T) = \sum_{j=1}^{N} \sum_{i=1}^{r_j} \beta((j,i))\cdot T_{(j,i)} = \sum_{i=1}^{r_{j^*}} \beta((j^*,i)) \cdot T_{(j^*,i)} = \sum_{i=1}^{r_{j^*} } \bb^{j^*}_i \cdot \bt_{i} = \bb^{j^*}\cdot \bt.
	$$
	Therefore, by \eqref{eq:T_type_props} we have
	$$
	\bb^{j^*} \cdot \bt =\beta(T)\leq \frac{\alpha \cdot k }{ n} \leq \alpha \cdot \kappa(T) = \alpha \cdot \bk^{j^*}\cdot \bt, 
	$$
	where the last equality is by \eqref{eq:ktimest}. 
	
	Furthermore,  by \eqref{eq:t_is_dist} it holds that  $\bt$ is a distribution. Overall, we showed that $\bt^*$ is a feasible solution for the optimization problem in \eqref{eq:alpha_num} with respect to the term $(\bb^{j^*},\bk^{j^*},\bgam^{j^*})$. The value of $\bt^*$ as a solution for the optimization problem is $\frac{1}{\bk^{j^*}\cdot \bt}\cdot \D{\bt}{\bgam^{j^*}}$, and since $M_{j^*}$ is the optimum, it follows that $\frac{1}{\bk^{j^*}\cdot \bt}\cdot \D{\bt}{\bgam^{j^*}} \geq M_{j^*}$. 
\end{claimproof}

By \eqref{eq:prob_to_t} and Claim~\ref{claim:t_to_number}, we have
$$
	\frac{p(\floor{\alpha k },k)}{(\alpha \cdot k+1)^{|\chi|+1} }\leq \frac{\Pr_{S^*}(G^{\floor{\alpha \cdot k},k} ) }{(\alpha \cdot k+1)^{|\chi|+1}} \leq 
\exp\left(-\frac{k}{\bk^j\cdot \bt} \cdot \D{\bt}{\bgam^{j^*}} \right) \leq \exp\left(-k\cdot M_{j^*}\right) = \exp(-k\cdot M ),
$$
where the first inequality is by \eqref{eq:upper_initial}.
As the above  inequality holds for every $k$ such that $p(\floor{\alpha \cdot k}, k)>0$, we have
$$
p(\floor{\alpha k}, k)\leq  (\alpha \cdot k+1)^{|\chi|+1} \cdot \exp\left(-k\cdot M\right)
$$
for all $k\in \mathbb{N}$. Therefore,
$$
\limsup_{k \rightarrow \infty} \frac{1}{k}\ln p\left(\floor{\alpha \cdot k},k\right) \leq \limsup_{k \rightarrow \infty} \frac{1}{k}\ln  \left((\alpha \cdot k+1)^{|\chi|+1} \cdot \exp\left(-k\cdot M\right)  \right) =-M.
$$

\begin{claimproof}[Proof of Claim~\ref{claim:simple_type}]
	Assume towards contradiction that there are $j\in [N]\setminus \{j^*\}$ and $i\in [r_j]$ such that $T_{(j,i)} > 0$. 
Let $(a_1,\ldots, a_n) \in \chi^{n}$ be a length $n$ vector of type $T$. That is, $\type(a_1,\ldots, a_n)=T$.  Since $T_{(j,i)}>0$, there is $1\leq \ell \leq n$ such that $a_{\ell} =(j,i)$. Therefore, $j\neq j^* = S^*(a_1,\ldots,a_{\ell-1})$ which implies that $(a_1,\ldots, a_n )$  is not consistent with $S^*$ (see Definition~\ref{def:consistent}). By Lemma~\ref{lem:inconsistent_prob}, we have that $\Pr_{S^*}\left(  (A_1,\ldots, A_n) = (a_1,\ldots,a_n)\right)  =0$ for every $(a_1,\ldots, a_n ) \in \chi^*$ such that $\type(a_1,\ldots, a_n) =T$. Therefore,
$$
\Pr_{S^*}(\type(A_1,\ldots,A_n)= T)  = \sum_{(a_1,\ldots, a_n )\in \chi^* \textnormal{ s.t. } \type(a_1,\ldots, a_n )= T}\Pr_{S^*}\left(  (A_1,\ldots, A_n) = (a_1,\ldots,a_n)\right)  =0.
$$
By \eqref{eq:upper_type}, we also have
$$
\begin{aligned}
0 &< \Pr_{S^*}\left( X_n\leq \alpha \cdot k  \textnormal{ and } Y_n\geq k \textnormal{ and } \type(A_1,\ldots, A_n) = T\right)\\
& \leq \Pr_{S^*} \left( \type(A_1,\ldots,A_n)= T\right) =0.
\end{aligned}
$$ 
A contradiction. Therefore, $T_{(j,i)}=0$ for every $j\in [N]\setminus \{j^*\}$ and $i\in [r_j]$. This completes the proof of the lemma.
\end{claimproof} 
 \end{proof}

\section{Discussion}
\label{sec:discussion}

In this paper we introduced a new technique for obtaining parameterized approximation algorithms leading to significant improvements in running times over existing algorithms. The analysis of our algorithms 
required the development of a mathematical machinery for the analysis
of a wide class of two-variable recurrence relations. Following the above results, several issues remain open:

\begin{itemize}
	\item From theoretical perspective, it is desirable to obtain deterministic variants of our algorithms. Derandomizing our technique is left for future work.

	\item
	Sanov's theorem also falls into the category of Large Deviation Theory. 
	There are some extensions of the theorem 
	from the viewpoint of probability theory. One of the most general of these is Gartner-Ellis theorem \cite{Ga77, El84} (see a unified claim in \cite{Ho00}). By using this theorem, some steps in the proof of Theorem \ref{thm:rec} may be skipped.
	We keep these steps to make the proof clearer and more accessible to readers outside the above areas.

	\item Often the analyses of branching algorithms use complex recurrence relations 
	involving two functions or more to obtain improved
	bounds on running times. Examples for such analyses can be 
	found in~\cite{CKJ01} and \cite{Fer10}. When transformed
	to the context of randomized branching, the analyses yield 
	recurrence relations in two functions, such as
	\begin{equation}
		\label{eq:multi_func_rec}
	\begin{aligned}
	&p(b,k) &=& \min \begin{cases}{ 0.5 \cdot p(b-1, k-1)+ 0.5 \cdot q(b-2, k)} \\
	0.5 \cdot p(b-1, k ) +  0.25\cdot q (b-2, k) + 	0.25 \cdot q(b-2, k-2) \end{cases} \\ 
	&q(b,k) &=& \min \begin{cases}{ 0.5 \cdot p(b-1, k-1)+ 0.5 \cdot q(b-3, k)} \\
	0.5 \cdot p(b-1, k ) + 0.25 \cdot q (b-3, k) + 	0.25 \cdot q (b-3, k-3) \end{cases}	\end{aligned}
	\end{equation}
	A tight analysis for such recurrences is likely to lead to improved parameterized approximations for small values of 
	$\alpha$ (for both Vertex Cover and $3$-Hitting Set), as the (exact) algorithms of \cite{CKJ01} and \cite{Fer10} have 
	better running times, compared to the running times of our
	algorithms for approximation ratios approaching $1$.  Our initial results suggest that it is possible to lower bound such recurrences using adaptation of the techniques presented in this paper. 
	
	Currently, the (exact) parameterized algorithm for Vertex Cover with best 
	running time is due to \cite{CKX10}. We were unable to obtain a randomized branching variant for this algorithm. One reason is that an incorrect branching can lead to an unbounded increase 
	in the mininmal vertex cover size.

	\item We showed the application of randomized branching to Vertex Cover and to $3$-Hitting Set.
	Following the publication of the conference version of this paper, a simple form of randomized branching has been used in  \cite{JLMRS23, EK24} to design parameterized approximation algorithms for Feedback Vertex Set and other Vertex Deletion problems on graphs. 
	
	In general, designing parameterized approximation algorithms for Vertex Deletion problems, such as Vertex Cover and $3$-Path Vertex Cover \cite{Tsur19}, seems 
	similar w.r.t. difficulty level to the design of {\em exact} parameterized branching algorithms for these problems. In both settings, the running times of natural algorithms can be improved by introducing more sophisticated branching rules. This
	holds also for many of the algorithms proposed in~\cite{EK24}.

\end{itemize}

\bigskip
\noindent {\bf Acknowledgments.}
We thank Henning Fernau and Daniel Lokshtanov for stimulating discussions on the paper. 
We are grateful to the Technion Computer Systems Laboratory for providing us the computational infrastructure used for the numerical evaluations in Section \ref{sec:HS}.

\bigskip 

\bibliographystyle{abbrv}

\bibliography{randbranch}
\end{document}